\documentclass[11pt]{article}

\usepackage[a4paper, total={6in, 8in}]{geometry}
\usepackage[utf8]{inputenc}
\usepackage[english]{babel}
\usepackage{graphicx}
\usepackage{url}
\usepackage{microtype}
\usepackage{booktabs} 
\usepackage{xspace}
\usepackage{color}
\usepackage{enumitem}
\usepackage{times}
\usepackage{appendix}
\usepackage{algorithmicx}
\usepackage{amsmath,amssymb,latexsym} 
\usepackage{algcompatible}
\usepackage{algpseudocode, algorithm}
\algrenewcommand\alglinenumber[1]{\tiny #1:}
\usepackage{listings}
\usepackage{verbatim}
\usepackage{multicol}
\usepackage{float}
\usepackage{pifont}
\usepackage{lineno}
\usepackage{array}
\usepackage{soul}
\usepackage{mathtools}
\usepackage{hyphenat}
\usepackage{setspace}
\usepackage{tcolorbox}
\usepackage{cite}
\usepackage{wrapfig}
\usepackage{lipsum}
\usepackage{layout}
\usepackage[english]{babel}
\usepackage[utf8]{inputenc}
\usepackage{fancyhdr}
\usepackage{tabularx}
\usepackage{caption}
\usepackage{ragged2e}
\usepackage{hyperref}
\usepackage{cleveref}
\usepackage{fancyvrb}
\usepackage{epstopdf}
\usepackage{epsfig}
\usepackage{etoolbox}
\usepackage{breqn}
\usepackage{relsize}
\usepackage{pgfplots}
\usepackage{subfiles}
\usepackage{authblk}
\usepackage{pdfpages}
\pgfplotsset{compat=1.7}
\usepackage{calligra}
\usepackage{calrsfs}
\usepackage{subfig}
\usepackage{todonotes}

\title{\textbf{Efficient means of Achieving Composability using Object based Conflicts on Transactional Memory}\thanks{A preliminary version of this work was accepted in AADDA 2018 as work in progress.}}

\author{Sathya Peri, Ajay Singh and Archit Somani\thanks{Author sequence follows lexical order of last names.}\\Department of Computer Science \& Engineering, IIT Hyderabad, India\\(sathya\_p, cs15mtech01001, cs15resch01001)@iith.ac.in}
\date{}

\newcommand{\punt}[1]{}
\newcommand{\cmnt}[1]{}

\newtheorem{theorem}{Theorem}

\newtheorem{lemma}[theorem]{Lemma}
\newtheorem{proof}{\textit{Proof:}}

\newtheorem{corollary}[theorem]{Corollary}

\newtheorem{definition}{Definition}

\newcounter{history}

\newtheorem{observation}[theorem]{Observation}

\newcommand{\secref}[1]{Section~\ref{sec:#1}}
\newcommand{\figref}[1]{Figure~\ref{fig:#1}}

\newcommand{\stref}[1]{step~\ref{step:#1}}

\newcommand{\lemref}[1]{Lemma~\ref{lem:#1}}

\newcommand{\corref}[1]{Corollary~\ref{cor:#1}}

\newcommand{\defref}[1]{Definition~\ref{def:#1}}

\newcommand{\obsref}[1]{Observation~\ref{obs:#1}}

\newcommand{\algoref}[1]{{Algo~\ref{algo:#1}}}

\newcommand{\subsecref}[1]{SubSection{\ref{subsec:#1}}}

\newcommand{\Secref}[1]{Section~\ref{sec:#1}}

\newcommand{\Corref}[1]{Corollary~\ref{cor:#1}}

\newcommand{\Eqref}[1]{eq(\ref{eq:#1})}

\newcommand{\Obsref}[1]{Observation~\ref{obs:#1}}

\newcommand{\Lineref}[1]{Line~\ref{lin:#1}}

\newcommand{\ignore}[1]{}

\newcommand{\tobj} {transaction object}
\newcommand{\txns}[1] {txns(#1)}
\newcommand {\comm}[1] {committed(#1)}
\newcommand {\aborted}[1] {aborted(#1)}
\newcommand {\incomp}[1] {incomp(#1)}

\newcommand {\term}[1] {term(#1)}
\newcommand {\key}[1] {key(#1)}

\newcommand {\conf}[1] {conflict(#1)}

\newcommand{\seq} {sequential}

\newcommand{\lupdt}[2] {#2.lastUpdt(#1)}

\newcommand{\mr} {MR}
\newcommand{\tr} {TR}

\newcommand{\legal} {legal}
\newcommand{\legality} {legality}

\newcommand{\op} {operation}
\newcommand{\mth} {method\xspace}
\newcommand{\mths} {methods\xspace}
\newcommand{\rvmt} {rv\_method\xspace}
\newcommand{\upmt} {upd\_method\xspace}
\newcommand{\rvp} {\emph{\rvmt execution} phase}
\newcommand{\cp} {\emph{\upmt execution} phase}

\newcommand{\cc} {correctness-criterion}

\newcommand{\inv} {$inv$}
\newcommand{\rsp} {$rsp$}

\newcommand{\evts}[1] {evts(#1)}
\newcommand{\met}[1] {methods(#1)}

\cmnt{\newcommand{\tbeg} {\emph{t\_begin}}
	\newcommand{\tread} {\emph{t\_read}}
	\newcommand{\twrite} {\emph{t\_write}}
	\newcommand{\tins} {\emph{t\_ins\xspace}}
	\newcommand{\tdel} {\emph{t\_del\xspace}}
	\newcommand{\tlook} {\emph{t\_look\xspace}}
	\newcommand{\tryc} {\emph{tryC\xspace}}
	\newcommand{\trya} {\emph{tryA\xspace}}
}

\newcommand{\tbg} {\emph{t\_begin\xspace}}
\newcommand{\tbeg} {\emph{STM\_begin\xspace}}
\newcommand{\tread} {\emph{t\_read\xspace}}
\newcommand{\twrite} {\emph{t\_write\xspace}}
\newcommand{\tins} {\emph{STM\_insert\xspace}}
\newcommand{\tdel} {\emph{STM\_delete\xspace}}
\newcommand{\tlook} {\emph{STM\_lookup\xspace}}
\newcommand{\tryc} {\emph{tryC}}
\newcommand{\trya} {\emph{tryA}}

\newcommand{\unaborted} {unaborted}

\newcommand{\inss} {\emph{list\_insert}}
\newcommand{\dell} {\emph{list\_del}}
\newcommand{\lookk} {\emph{list\_lookup}}

\newcommand{\ins} {\emph{rbl\_ins}}
\newcommand{\del} {\emph{rbl\_del}}
\newcommand{\look} {\emph{rbl\_Search}}

\newcommand{\rv} {\emph{rv}}
\newcommand{\up} {\emph{up}}

\newcommand{\opq} {opaque\xspace}
\newcommand{\opty} {opacity\xspace}
\newcommand{\lble} {linearizabale\xspace}
\newcommand{\lbty} {linearizability\xspace}
\newcommand{\coop} {co-opaque\xspace}
\newcommand{\coopty} {co\text{-}opacity\xspace}

\newcommand{\tab} {\texttt{hash-table\xspace}}

\newcommand{\llist} {list}
\newcommand{\otm} {\textit{HT-OSTM}\xspace}
\newcommand{\ltm} {\textit{list-OSTM}\xspace}
\newcommand{\rwtm} {\textit{RWSTMs}}
\newcommand{\lotm} {OSTM\xspace}
\newcommand{\lotms} {OSTMs\xspace}

\newcommand{\bto} {BTO\xspace}
\newcommand{\bst} {BST\xspace}
\newcommand{\lsl} {lazyrb-list\xspace}

\newcommand{\csr} {CSR\xspace}

\newcommand{\cg}[1] {CG(#1)}

\newcommand{\rvm} {\emph{rvm}\xspace}

\newcommand{\fevt}[1] {#1.firstEvt}
\newcommand{\levt}[1] {#1.lastEvt}

\newcommand{\fkmth}[3] {#3.firstKeyMth(#1, #2)}
\newcommand{\pkmth}[3] {#3.prevKeyMth(#1, #2)}

\newcommand{\udset}[1] {updtSet(#1)}
\newcommand{\rvset}[1] {rvSet(#1)}

\newcommand{\lin}[1]{#1.LP\xspace}

\newcommand\tabspace[1][1cm]{\hspace*{#1}}
\newcommand{\preds} {sh\_preds[]}
\newcommand{\currs} {sh\_currs[]}

\newcommand{\glslhead}{getRBLHead($obj\_id \downarrow, key \downarrow$)}
\newcommand{\llgopn}[1] {$le.getOpn(obj\_id \downarrow$, $ key \downarrow$)}
\newcommand{\llsopn}[1] {$le.setOpn(obj\_id \downarrow$, $key \downarrow$, #1)}

\newcommand{\llgval}[1] {$le.getValue(obj\_id \downarrow$, $ key \downarrow$)}
\newcommand{\llsval}[1] {$le.setValue(obj\_id \downarrow$, $ key \downarrow$, #1)}

\newcommand{\llspc} {$le.setPreds\&Currs(obj\_id \downarrow$, $ key \downarrow$, $\preds \downarrow$, $\currs \downarrow$)}

\newcommand{\llgaptc}[1] {$le.getAptCurr$(#1)}
\newcommand{\llgkeyobj} {$le.getKey\&Objid(le_i \downarrow$)}

\newcommand{\llsopst}[1] {$le.setOpStatus(obj\_id \downarrow$, $ key \downarrow$, #1)}

\newcommand{\txlfind} {\textup{txlog.findInLL}$(t\_id \downarrow, obj\_id \downarrow, key \downarrow, le \uparrow)$}

\newcommand{\txgllist} {txlog.getLlList($t\_id \downarrow$)}
\newcommand{\txsetst}[1] {txlog.setStatus($txstatus \downarrow$, $ OK \downarrow$)}

\newcommand{\lsls}[1] {rblSearch($t\_id \downarrow$, $obj\_id \downarrow$, $ key \downarrow$, #1, $\preds \uparrow$, $\currs \uparrow$, $op\_status \uparrow$ )}

\newcommand{\lslins}[1] {rblIns($\preds \downarrow$, $\currs \downarrow$, #1)}
\newcommand{\lsldel} {rblDel($\preds \downarrow$, $\currs \downarrow$)}

\newcommand{\rlsol} {releaseOrderedLocks($ordered\_ll\_list \downarrow$)}

\newcommand{\toval} {transValidation($t\_id \downarrow$, $ key \downarrow$, $\currs \downarrow$, $val\_type \downarrow$, $op\_status \uparrow$)}

\newcommand{\cld} {commonLu\&Del($t\_id\downarrow$, $obj\_id\downarrow$, $key\downarrow$, $value\uparrow,$ $op\_status\uparrow$)}

\newcommand{\cldd} {commonLu\&Del()}

\newcommand{\handlea}{handleAbort($t\_id \downarrow$)}
\newcommand{\llsort}{txlog.sort ($ll\_list \downarrow$)}

\newcommand{\nptov} {\emph{transValidation()}}
\newcommand{\npintv} {\emph{methodValidation()}}
\newcommand{\nptc} {\emph{STM\_tryC()}}
\newcommand{\npins} {\emph{STM\_insert()}}
\newcommand{\npdel} {\emph{STM\_delete()}}
\newcommand{\npluk} {\emph{STM\_lookup()}}
\newcommand{\nplsls} {\emph{rblSearch()}}
\newcommand{\nplsldel} {\emph{rblDel()}}
\newcommand{\nplslins} {\emph{rblIns()}}
\newcommand{\npbegin} {\emph{STM\_begin}}
\newcommand{\nppoval} {\emph{intraTransValidation()}}

\newcommand{\rn} {\textcolor{red}{rl\xspace}}
\newcommand{\bn} {\textcolor{blue}{bl\xspace}}

\newcommand{\rc} {\textcolor{red}{sh\_currs[0]}}
\newcommand{\bc} {\textcolor{blue}{sh\_currs[1]}}
\newcommand{\bp} {\textcolor{blue}{sh\_preds[0]}}
\newcommand{\rp} {\textcolor{red}{sh\_preds[1]}}

\newcommand {\cmntwa}[1] {\State{\textcolor{gray}{/* #1 */} }}

\begin{document}
\maketitle

\begin{abstract}
Composing together the individual atomic \mths of concurrent data-structures ($cds$) pose multiple design and consistency challenges. In this context composition provided by transactions in software transaction memory (STM) can be handy. However, most of the STMs offer read/write primitives to access shared $cds$. These read/write primitives result in unnecessary aborts. Instead, semantically rich higher-level \mths of the underlying $cds$ like lookup, insert or delete (in case of \tab{} or lists) aid in ignoring unimportant lower level read/write conflicts and allow better concurrency. 

In this paper, we adapt transaction tree model in databases to propose OSTM which enables efficient composition in $cds$. We extend the traditional notion of conflicts and legality to higher level \mths of $cds$ using STMs and lay down detailed correctness proof to show that it is \coop{}. We implement OSTM with concurrent closed addressed \tab{} (\emph{HT-OSTM}) and list (\ltm{}) which exports the higher-level operations as transaction interface.

In our experiments with varying workloads and randomly generated transaction operations, \otm{} shows speedup of 3 to 6 times and w.r.t aborts \otm{} is 3 to 7 times better than ESTM and read/write based STM, respectively.
Where as, \emph{list-OSTM} outperforms state of the art lock-free transactional list, NOrec STM list and boosted list by 30\% to 80\% across all workloads and scenarios. Further, \emph{list-OSTM} incurred negligible aborts in comparison to other techniques considered in the paper. 
\end{abstract}

\section{Introduction}
\label{sec:intro}
Software Transaction Memory Systems (\textit{STMs}) are a convenient programming interface for a programmer to access shared memory without worrying about concurrency issues \cite{HerlMoss:1993:SigArch,ShavTou:1995:PODC} and are natural choice for achieving composability\cite{Harretal:2005:PPoPP}.

Most of the \textit{STMs} proposed in the literature are specifically based on read/write primitive operations (or methods) on memory buffers (or memory registers). These \textit{STMs} typically export the following methods: \tbg{} which begins a transaction, \tread{} which reads from a buffer, \twrite{} which writes onto a buffer, \tryc{} which validates the \op{s} of the transaction and tries to commit. We refer to these as \textit{Read-Write STMs or \rwtm{}}. As a part of the validation, the STMs typically check for \emph{conflicts} among the \op{s}. Two \op{s} are said to be conflicting if at least one of them is a write (or update) \op. Normally, the order of two conflicting \op{s} cannot be commutated.  
On the other hand, \emph{Object STMs} or \textit{OSTM} operate on higher level objects rather than read \& write \op{s} on memory locations. They include more semantically rich \op{s} such as enq/deq on queue objects, push/pop on stack objects and insert/lookup/delete on sets, trees or \tab{} objects depending upon the underlying data structure used to implement OSTM. 

\ignore{We assume that the \tab{} object supports insert, delete and lookup \op{s} on $\langle$key, value$\rangle$ pairs. For showing correctness of our implementation, we will use \opty{}\cite{GuerKap:2008:PPoPP} since this \cc{} is generic enough to apply on \otm{s} unlike other criteria which are specific only to \rwtm{s}. }

It was shown in databases that object-level systems provide greater concurrency than read/write systems \cite[Chap 6]{WeiVoss:2002:Morg}. Along the same lines, we propose a model to achieve composability with greater concurrency for \textit{STMs} by considering higher-level objects which leverage the richer semantics of object level \mths. We motivate this with an interesting example.


Consider an \textit{OSTM} operating on the \tab{} object called as \textit{Hash-table Object STM} or \textit{HT-OSTM} which exports the following  methods - \tbeg: which begins a transaction (same as in \rwtm); \tins{} which inserts a value for a given key; \tdel{} which deletes the value associated with the given key; \tlook{} which looks up the value associated with the given key and \emph{STM\_tryC} which validates the \op{s} of the transaction.

\begin{figure}[H]
		\centering
	\centerline{\scalebox{0.5}{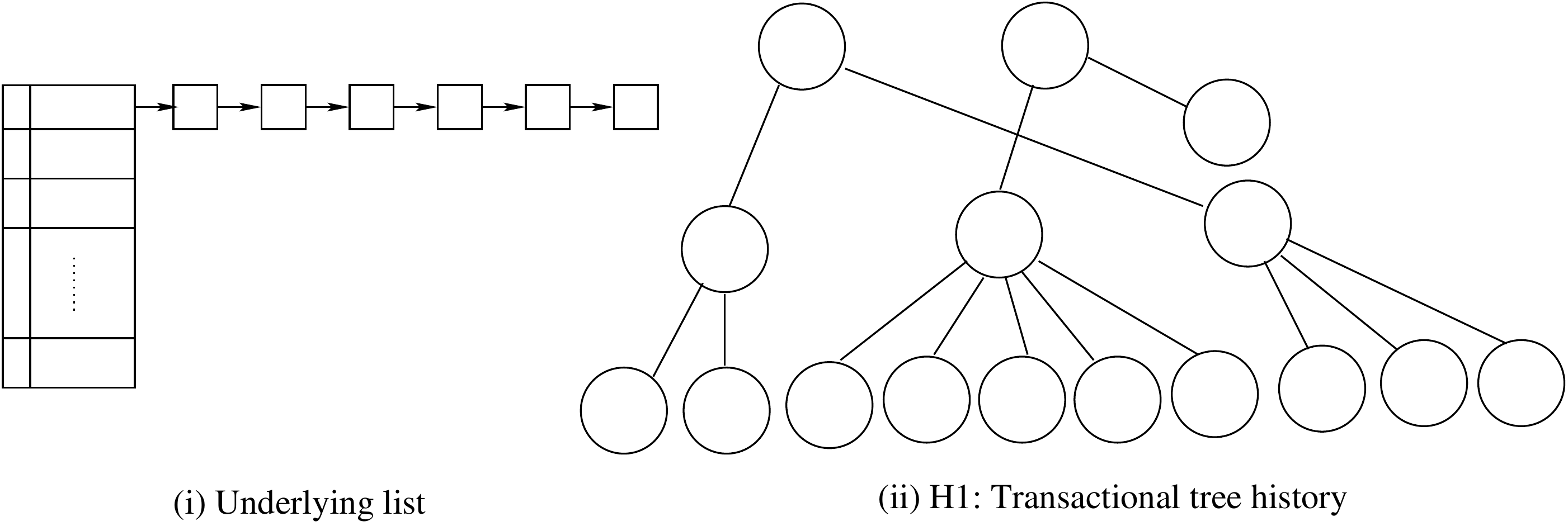}}
	\caption{Motivational example for OSTMs}
	\label{fig:tree-exec}
\end{figure}

A simple way to implement the concurrent \otm is using a \llist{} (a single bucket) where each element of the \llist{} stores the $\langle$key, value$\rangle$ pair. The elements of the \llist{} are sorted by their keys similar to the set implementations discussed in \cite[Chap 9]{Herlihy:ArtBook:2012}. It can be seen that the underlying \llist{} is a concurrent data-structure manipulated by multiple transactions. So, we may use the lazy-list based concurrent set \cite{Heller+:LazyList:PPL:2007} to implement the \op{s} of the \llist{} denoted as: \inss, \dell{} and \lookk{}. Thus, when a transaction invokes \tins{}, \tdel{} and \tlook{} methods, the STM internally invokes the \inss{}, \dell{} and \lookk{} methods respectively. 

Consider an instance of \llist{} in which the nodes with keys $\langle k_2~ k_5~ k_7~ k_8 \rangle$ are present in the \tab{} as shown in \figref{tree-exec}(i) and transactions  $T_1$ and $T_2$ are concurrently executing $\tlook_1(k_5)$ (shortened as l), $\tdel_2(k_7)$ (shortened as d) and $\tlook_1(k_8)$ as shown in \figref{tree-exec}(ii).
In this setting, suppose a transaction $T_1$ of \otm{} invokes methods \tlook{} on the keys $k_5, k_8$. This would internally cause the \otm{} to invoke \lookk{} method on keys $\langle k_2, k_5 \rangle$ and $\langle k_2, k_5, k_7, k_8 \rangle$ respectively.

\cmnt{
\begin{figure}[tbph]
\centerline{\scalebox{0.4}{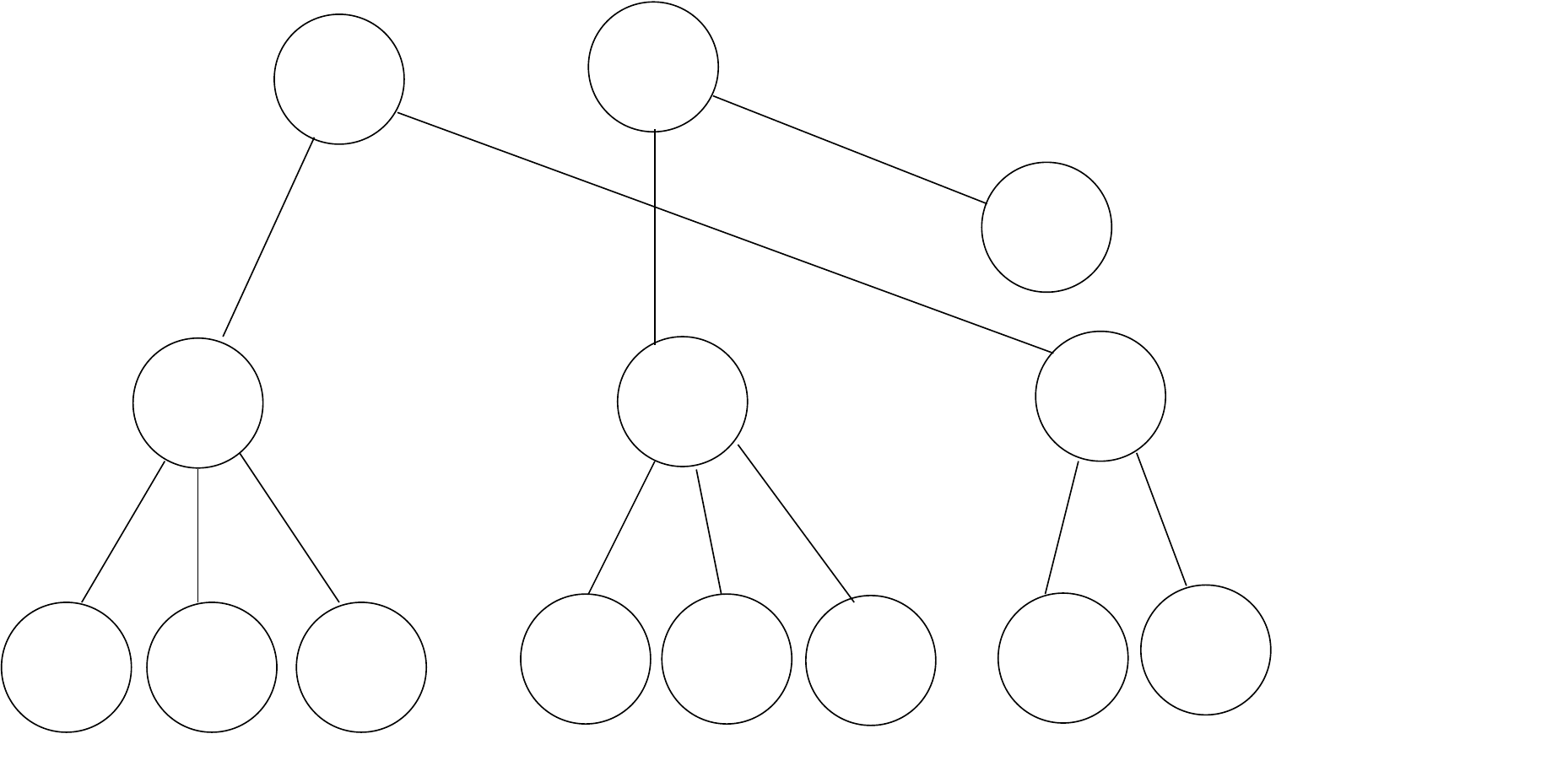}}
\caption{An execution on \lotm{}}
\label{fig:tree-exec1}
\end{figure}
}

Concurrently, suppose transaction $T_2$ invokes the method \tdel{} on key $k_7$ between the two \tlook{s} of $T_1$. This would cause, \otm{} to invoke \dell{} method of \llist{} on $k_7$. Since, we are using lazy-list approach on the underlying \llist, \dell{} involves pointing the next field of element $k_5$ to $k_8$ and marking element $k_7$ as deleted. Thus \dell{} of $k_7$ would execute the following sequence of read/write level operations- $r(k_2) r(k_5) r(k_7) w(k_5) w(k_7)$ where $r(k_5), w(k_5)$ denote read \& write on the element $k_5$ with some value respectively. The execution of \otm{} denoted as a \emph{history} can be represented as a transactional forest as shown in \figref{tree-exec}(ii). Here the execution of each transaction is a tree. 

In this execution, we denote the read/write \op{s} (leaves) as layer-0 and \tlook, \tdel{} methods as layer-1. Consider the history (execution) at layer-0 (while ignoring higher-level \op{s}), denoted as $H0$. It can be verified this history is not \opq \cite{GuerKap:2008:PPoPP}. This is because between the two reads of $k_5$ by $T_1$, $T_2$ writes to $k_5$. It can be seen that if history $H0$ is input to a \rwtm{} one of the transactions among $T_1$ \& $T_2$ would be aborted to ensure correctness (in this case opacity\cite{GuerKap:2008:PPoPP}). On the other hand consider the history $H1$ at layer-1 consisting of \tlook, \tdel{} methods while ignoring the underlying read/write \op{s}. We ignore the underlying read \& write \op{s} since they do not overlap (referred to as pruning in \cite[Chap 6]{WeiVoss:2002:Morg}). Since these methods operate on different keys, they are not conflicting and can be re-ordered either way. Thus, we get that $H1$ is \opq{}\cite{GuerKap:2008:PPoPP} with $T_1 T_2$ (or $T_2 T_1$) being an equivalent serial history.

\begin{figure}[H]
	\centering
	\captionsetup{justification=centering}
	\centerline{\scalebox{0.5}{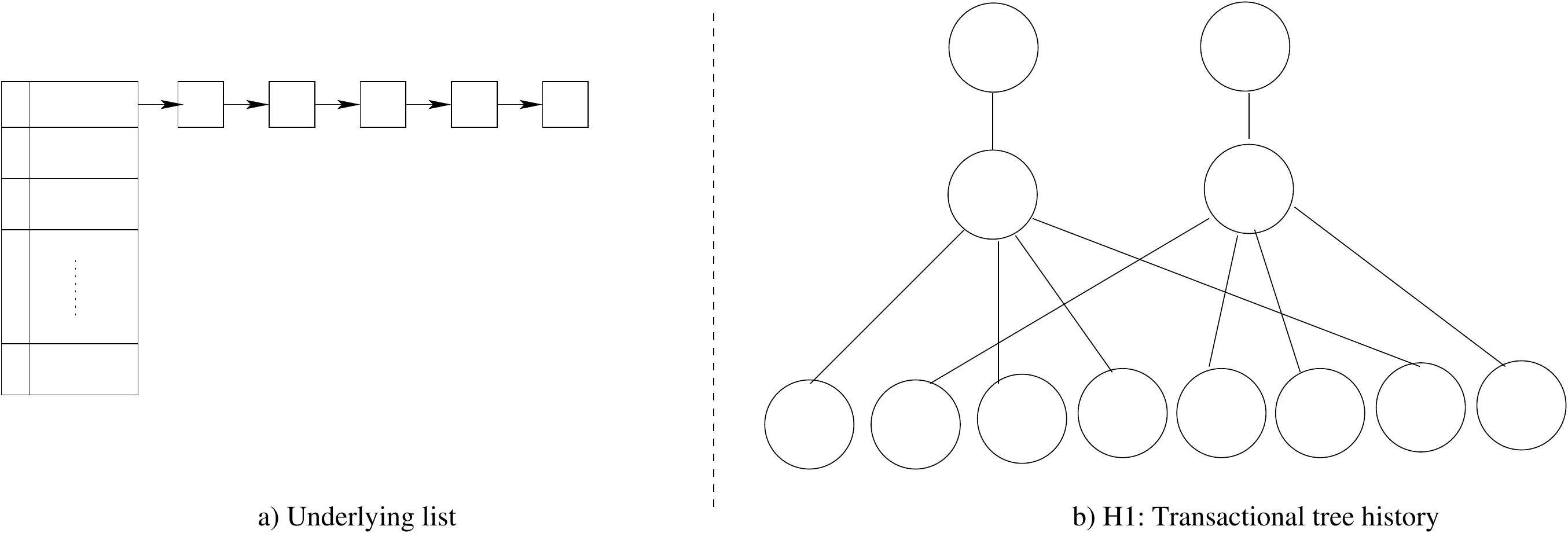}}
	\caption{Not linearizable at layer-0 due to cyclic conflicts  $r_2(k_2) w_1(k_2) w_2(k_2)$. Thus, lower level can not be isolated which causes no particular order at layer-1.}
	\label{fig:mvostm13}
\end{figure}

The important idea in the above argument is that some conflicts at lower-level \op{s} do not matter at higher level operations. Thus, such lower level conflicting operations may be ignored as shown in \figref{tree-exec}. Harris et al. referred to it as \emph{benign-conflicts}\cite{Harris_abstractnested}. On the other hand, \figref{mvostm13} shows that some lower level conflicts do matter at higher level. With object level modeling of histories, we get a higher number of acceptable schedules than read/write model.  The history, $H1$ in \figref{tree-exec}(ii) clearly shows the advantage of considering STMs with higher level \tins, \tdel{} and \tlook{} \op{s}. 

\cmnt{Now consider an application where $ht1$ and $ht2$ are two \tab{} objects. A process $p_1$ deletes $k_5$ from $ht1$ and inserts it into $ht2$. Another process $p_2$ looks up $k_5$. Now, $p_2$ may see an intermediate state of the system where $p_1$ has deleted the $k_5$ from $ht1$ but has not inserted in the $ht2$ violating compositionality.} 

The atomic property of transactions helps to correctly compose together several different individual operations. The above examples demonstrate that the concurrency in such STM can be enhanced by considering the object level semantics. To achieve this, in this paper:

\begin{enumerate}[label=(\alph*)]
    \item We propose a generic framework for composing higher level objects based on the notion of conflicts for objects in databases \cite[Chap 6]{WeiVoss:2002:Morg}.
\item For correctness our framework we consider, \opty\cite{GuerKap:2008:PPoPP} a popular \cc{} for STMs which is different from serializability commonly used in databases. It can be proved that verifying the membership of \opty similar to view-serializability is NP-Complete \cite{Papad:1979:JACM}. Hence, using conflicts we develop a subclass of \opty - \emph{conflict \opty} or \emph{\coopty} for objects. We then develop polynomial time graph characterization for \coopty based on conflict-graph acyclicity. The proposed \cc, \coopty is similar to the notion of conflict-opacity developed for \emph{\rwtm} by Kuznetsov \& Peri \cite{KuzPer:NI:TCS:2017}. 
\item To show the efficacy of this framework, we develop \emph{\otm} based on the idea of \emph{basic timestamp order (\bto)} scheduler developed in databases \cite[Chap 4]{WeiVoss:2002:Morg}. For showing correctness of \otm, we show that all the \mth{s} are \lble while the transactions are \emph{co-opaque} by showing that the corresponding conflict graph is acyclic. Although we have considered \otm here, we believe that this notion of conflicts can be extended to other high-level objects such as Stacks, Queues, Tries etc. 
\end{enumerate}
A simple modification of \otm gives us a concurrent list based STM  or \emph{\ltm}. Finally, we compared the     performance of \otm against a \tab{} application built using \rwtm: ESTM \cite{Felber:2017:jpdc} and \bto \cite{WeiVoss:2002:Morg,Singh2017PerformanceCO}. The \ltm{} is compared with lock-free transactional list\cite{Zhang:2016:LTW:2935764.2935780}, NOrec based RSTM list\cite{conf/ppopp/DalessandroSS10} and boosting list\cite{herlihy2008transactional}. The results show that \otm{} and \ltm{} reduces the number of aborts to minimal and show significant performance gain in comparison to other techniques.


\vspace{1mm}
\noindent
\textbf{Roadmap.} We explain the system model in \secref{model}. In \secref{opty}, we build the notion of \legality, conflicts to describe \opty, \coopty and the graph characterization. Based on the model we demonstrate the \otm{} design in \secref{htostm}. In \secref{pscode}, \secref{optm} and \secref{pocapp} we define \otm{} pseudocode, optimizations and proof sketch of \otm, respectively. In \secref{results} we show the evaluation results. Finally, we conclude in \secref{conc}.  
\section{Related Work} 
Our work differs from databases model in with regard to \cc{} used for safety. While databases consider \csr. We consider \lbty to prove the correctness of the \mth{s} of the transactions and \opty to show the correctness of the transactions. Earliest work of using the semantics of concurrent data structures for object level granularity include that of open nested transactions\cite{ni2007open} and transaction boosting of Herlihy et al.\cite{herlihy2008transactional} which is based on serializability(strict or commit order serializability) of generated schedules as correctness criteria. Herlihy's model is pessimistic and uses undo logs for rollback. Our model is more optimistic in that sense and the underlying data structure is updated only after there is a guarantee that there is no inconsistency due to concurrency. Thus, we do not need to do rollbacks which keeps the log overhead minimal. This also solves the problem of irrevocable operations being executed during a transaction which might abort later otherwise.

Hassan et al.\cite{Hassan+:OptBoost:PPoPP:2014} have proposed Optimistic Transactional Boosting (OTB) that extends original transactional boosting methodology by optimizing and making it more adaptable to STMs. They further have implemented OTB on set data structure using lazy-linked list\cite{Hassan+:OptBoost:PPoPP:2014}. Although there seem similarities between their work and our implementation, we differ w.r.t the \cc{} which is \coopty a subclass of \opty\cite{KuzPer:NI:TCS:2017} in our case. Furthermore, we also differ in the development of the conflict-based theoretical framework which can be adapted to build other object based STMs.

Transactional boosting idea of Herlihy et. al\cite{herlihy2008transactional} tries to utilize the object level semantics of linearizable datastructures. They assume $cds$ to be blackbox and try to transactify the base object(underlying datastructure); We in turn, consider the lower level operations (level-0) which aids to introduce $cds$ specific optimizations. Herlihy claims to differ from open nested transactions by providing a precise methodology and characterization of the mechanism. However, they maintain a log of each operation's inverse, which needs to execute once a transaction aborts. This incurs additional computational and memory cost. Moreover, many data structures do not provide reverse operations (for example, priority queue). The proposed \otm{} do not need reverse operation as we follow deferred update augmented with optimism of time-order based validation.
Moreover, transactional boosting is based on serlizabilty(strict or commit order serializabilty) of generated schedules as correctness critera. Herlihy's model is pessimistic and uses undo logs for rollback. Our model is more optimistic in that sense and underlying data structure is updated only after there is a guarantee that there is no inconsistency due to concurrency. Thus, we do not need to do rollbacks which keeps the log overhead minimal. This also solves the problem of irrevocable operations being executed during a transaction which might abort later otherwise.

Zhang et al.\cite{Zhang:2016:LTW:2935764.2935780} recently propose a method to transform lockfree $cds$ to transactional lockfree linked $cds$ and base the correctness on \emph{strict serializability}. The transactions are synchronized using CAS and they compare their work against STM based approaches. Our evaluation shows that \ltm{} implementation comprehensibly beats Zhang's transactional lock free list data structure. 

Fraser et. al.\cite{Fraser:2007:CPW:1233307.1233309} proposed OSTM based on shadow copy mechanism, which involves a level of indirection to access the shared objects through \emph{OSTMOpenForReading} and \emph{OSTMOpenForWriting} as exported \mths. Contrary to it, our OSTM model exports the higher object level methods like \npluk{}, \npins{} and \npdel{} while hiding the internal read and write lower level primitives. So, it seems that using the Fraser OSTM one can write the higher level methods transactionally using its read/write \mths{}. For example, one may implement a \emph{lookup} on the underlying list object using its transactional interface. But we differ here because we allow such multiple higher level operations to be grouped together atomically without requiring user to implement them explicitly. 
The exported \mth{s} in Fraser et.al's OSTM may allow \emph{OSTMOpenForReading} to see the inconsistent state of the shared objects but our OSTM model precludes this possibility by validating the access during execution of \rvmt{} (i.e. the methods which do not modify the underlying objects and only return some value by performing a search on them).Fraser's OSTM uses the transaction descriptors which stores the previous and new copies of the shared objects increasing the memory requirement to maintain the meta data. We, on the other hand, maintain a single copy of the underlying shared object and the meta information is augmented within each shared object. For example, in case of a list, each node is a shared object. Here we augment each shared node with the meta data (in our case the time-stamp of access by the other transactions) along with a unique \emph{key} and the \emph{value} pair (value may store any complex data type of any type). 
 Thus, we can say our motivation and implementation is different from Fraser OSTM\cite{Fraser:2007:CPW:1233307.1233309} and only the name happens to coincide.

\begin{figure}[H]
	\centering
	\includegraphics[scale=0.8]{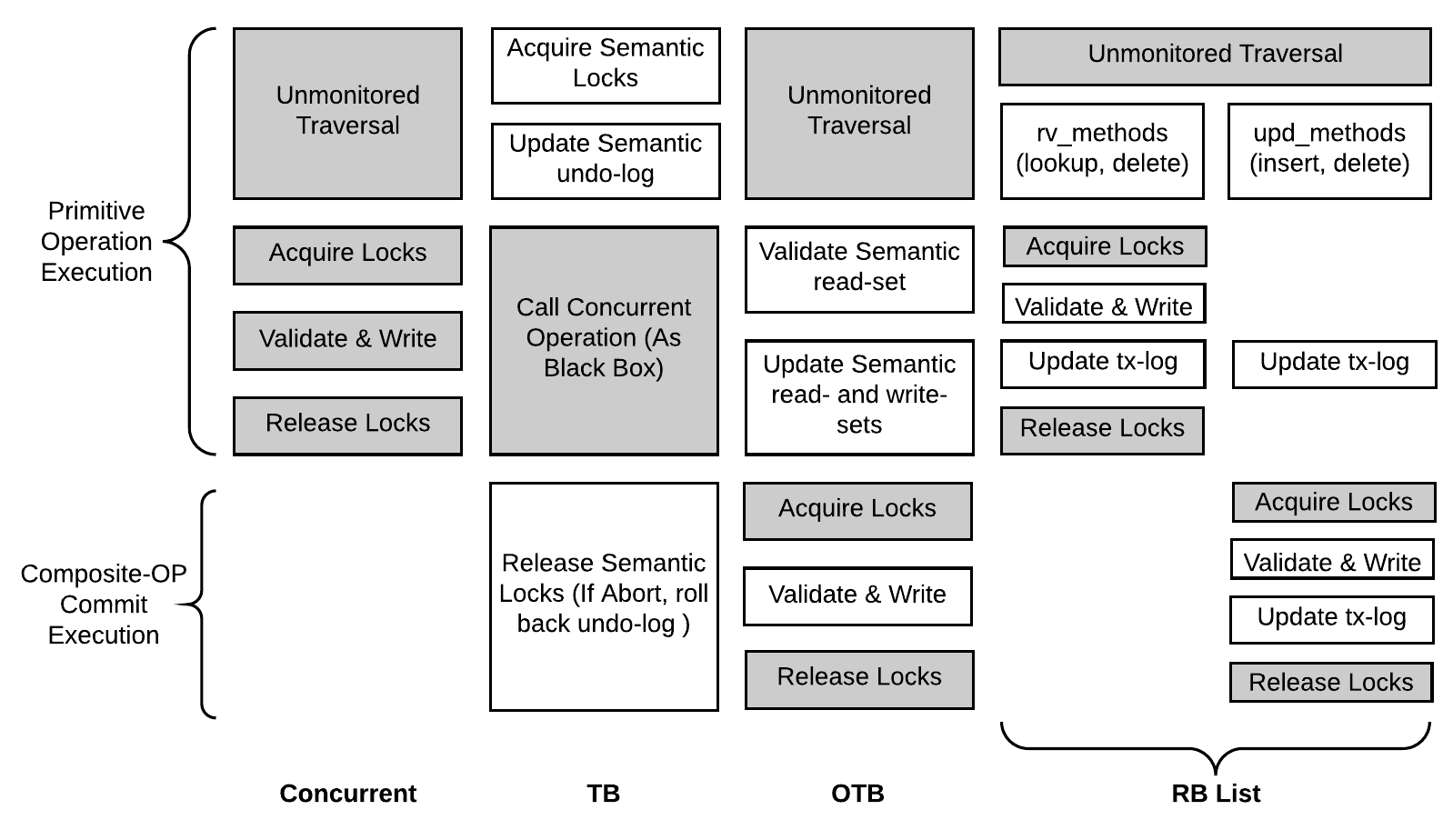}
	\caption{\otm{} design comparison against state of art techniques.}
	\label{fig:ostm-design-comp}
\end{figure}
 \figref{ostm-design-comp} compares the execution flow of normal concurrent data structure, boosted data structure, optimistically boosted data structure and the \otm{}.
\vspace{-4mm}
\section{Building System Model}
\label{sec:model}
\vspace{-4mm}
In this paper, we assume that our system consists of finite set of $P$ processors, accessed by a finite number of $n$ threads that run in a completely asynchronous manner and communicate using shared objects. The threads communicate with each other by invoking higher-level \mth{s} on the shared objects and getting corresponding responses. Consequently, we make no assumption about the relative speeds of the threads. We also assume that none of these processors and threads fail or crash abruptly.

\ignore{
	
\vspace{1mm}
\noindent
\textbf{\textsl{Events:}} 
Similar to Lev{-}Ari et. al.\cite{DBLP:conf/wdag/Lev-AriCK14, DBLP:conf/wdag/Lev-AriCK15}

\textbf{\textsl{Global States:}} We define the \emph{global state} or \emph{state} of the system as the collection of local and shared variables across all the threads in the system. The system starts with an initial global state. We assume that all the events executed by different threads are totally ordered. Each update event transitions the global state of the system leading to a new global state.\
}

\vspace{1mm}
\noindent
\textbf{\textsl{Events:}} We assume that the threads execute atomic \emph{events}. We assume that these events by different threads are (1) read/write on shared/local memory objects, (2) \mth{} invocations (or \emph{\inv}) event \& responses (or \emph{\rsp}) event on higher level shared-memory objects.

\vspace{1mm}
\noindent
\textbf{\textsl{Global States:}} We define the \emph{global state} or \emph{state} of the system as the collection of local and shared variables across all the threads in the system. The system starts with an initial global state. We assume that all the events executed by different threads are totally ordered. Each update event transitions the global state of the system leading to a new global state.

\vspace{1mm}
\noindent
\textbf{\textsl{Methods:}} The $n$ processes access a collection of \emph{\tobj{s}} via atomic \emph{transactions} supported by a \lotm. Each transaction has a unique identifier typically denoted as $T_i$. Within a transaction, a process can invoke transactional \mth{s} on a \emph{\tab} \tobj. A \tab{}($ht$) consists of multiple key-value pairs of the form $\langle k, v \rangle$. The keys and values are respectively from sets $\mathcal{K}$ and $\mathcal{V}$. The \mth{s} that a transaction $T_i$ can invoke are: (1) $\tins_i(ht, k, v)$: this \mth{} inserts the pair $\langle k,v \rangle$ into object $ht$ and return $ok$. If $ht$ already has a pair $\langle k,v' \rangle$ then $v'$ gets replaced with $v$. (2) $\tdel_i(ht, k, v)$: if $ht$ has a $\langle k,v \rangle$ pair then this \op{} deletes the pair and returns $v$. If no such $\langle k,v \rangle$ pair is present in $ht$, then the \op{} returns $nil$. (3) $\tlook_i(ht, k, v)$: if $ht$ has a $\langle k,v \rangle$ pair then this \op{} returns $v$. If no such $\langle k,v \rangle$ pair is present in $ht$, then the \mth{} returns $nil$. It can be seen that \tlook{} is similar to \tdel.

For simplicity, we assume that all the values inserted by transactions through \tins{} \mth{} are unique. We denote \tins{} and \tdel{} as \emph{update} \mth{s} since both these change the underlying data-structure.We denote \tdel{} and \tlook{} as \emph{return-value methods or \rvmt{s}} as these return values which are different from $ok$.


In addition to these return values, each of these methods can always return an abort value $\mathcal{A}$ which implies that the transaction $T_i$ is aborted. A \mth{} $m_i$ returns $\mathcal{A}$ if $m_i$ along with all the \mth{s} of $T_i$ executed so far are not consistent (w.r.t  \cc{} which is formally defined later). 

The \otm{} supports two other methods: (4) $\tryc_i$: this \mth{} tries to validate all the \op{s} of the $T_i$. \otm{} returns $ok$ if $T_i$ is successfully committed. Otherwise, \otm{} returns $\mathcal{A}$ implying abort. This method is invoked by a process after completing all its transactional \op{s}. (5) $\trya_i$: this method returns $\mathcal{A}$ and \otm{} aborts $T_i$. 

When any \mth{} of $T_i$ returns $\mathcal{A}$, we denote that method as well as $T_i$ as aborted. We assume that a process does not invoke any other \op{s} of a transaction $T_i$, once it has been aborted. We denote a \mth{} which does not return $\mathcal{A}$ as \emph{\unaborted}.

Having described about \mth{s} of a transaction, we describe about the events invoked by these \mth{s}. We assume that each \mth{} consists of a \inv{} and \rsp{} event. Specifically, the \inv{} \& \rsp{} events of the methods of a transaction $T_i$ are: (1) $\tins_i(ht, k, v)$: $\inv(\tins_i(ht, k, v))$ and $\rsp(\tins_i(ht, k, v, ok/\mathcal{A}))$. (2) $\tdel_i(ht, k, v)$: $\inv(\tdel_i(ht, k))$ and $\rsp(\\\tdel_i(h, k, v/nil/\mathcal{A}))$. (3) $\tlook_i(h, k, v)$: $\inv(\tlook_i(h, k))$ and $\rsp(\tlook_i\\(h, k, v/nil/\mathcal{A}))$. (4) $\tryc_i$: $\inv(\tryc_i())$ and $\rsp(\tryc_i(ok/\mathcal{A}))$. (5) $\trya_i$: $\inv(\trya_i())$ and $\rsp(\trya_i(\mathcal{A}))$. 

For clarity, we have included all the parameters of \inv{} event in \rsp{} event as well. In addition to these, each \mth{} invokes read/write primitives (operations) of $T_i$ are represented as: $r_i(x, v)$ implying that $T_i$ reads value $v$ for $x$; $w_i(x, v)$ implying that $T_i$ writes value $v$ onto $x$. Depending on the context, we ignore some of the parameters of the transactional methods and read/write primitives. We assume that the first event of a \mth{} is \inv{} and the last event is \rsp. 

Formally, we denote a \mth{} $m$ by the tuple $\langle \evts{m}, <_m\rangle$. Here, $\evts{m}$ are all the events invoked by $m$ and the $<_m$ a total order among these events. For instance, the \mth{} $l_{11}(k_5)$ of \figref{trans-tree} is represented as: $\inv(l_{11}(h, k_5))~ r_{111}(k_2, o_2) r_{112}(k_5, o_5)~ \rsp(l_{11}(h, k_5, o_5))$. In our representation, we abbreviate \tins{} as $i$, \tdel{} as $d$ and \tlook{} as $l$. From our assumption, we get that for any read/write primitive $rw$ of $m$, $\inv(m) <_m rw <_m \rsp(m)$.

\vspace{1mm}
\noindent
\textbf{\textsl{Transactions:}} Following the notations used in database multi-level transactions \cite{WeiVoss:2002:Morg}, we model a transaction as a two-level tree. \figref{trans-tree} shows a tree execution of a transaction $T_1$. The leaves of the tree denoted as \emph{layer-0} consist of read, write primitives on atomic objects. Hence, they are atomic. For simplicity, we have ignored the \inv{} \& \rsp{} events in level-0 of the tree. \emph{Level-1} of the tree consists of \mth{s} invoked by transaction. In the transaction shown in \figref{trans-tree}, level-1 consists of \tlook{} and \tdel{} \mth{s} operating on the \lsl{} as also shown in \figref{tree-exec}(i). 

\cmnt{
	\begin{Figure}
		\centering
		\scalebox{.5}{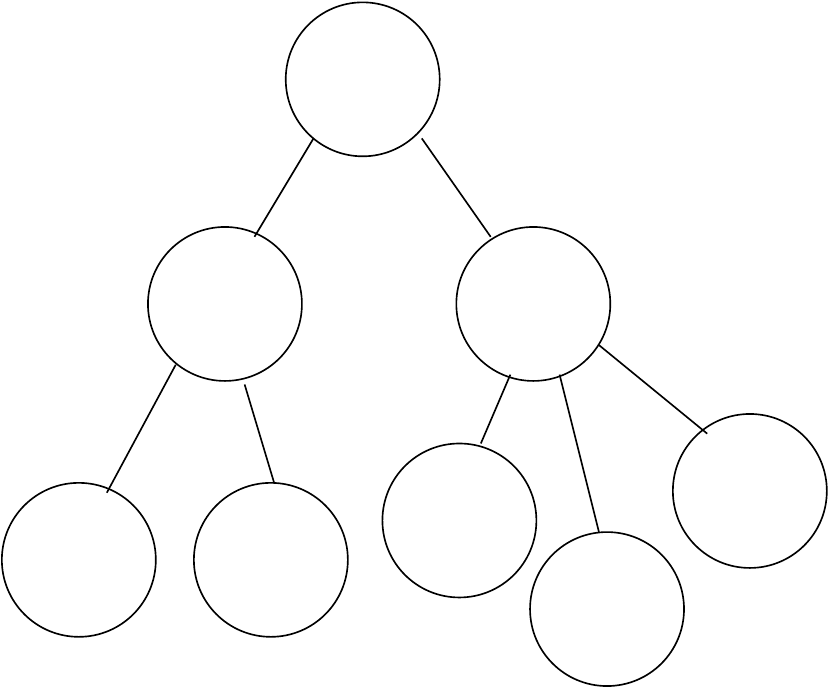\label{fig:trans-tree}}
		\vspace{-16cm}
		\captionof{figure}{T1 : A sample transaction on \bst{}(of \figref{listex}) representing a \tab{} object.}
	\end{Figure}

	\begin{wrapfigure}{R}{5cm}
		\centerline{\scalebox{0.4}{\input{figs/trans-tree.pdf_t}}}
		\caption{T1 : A sample transaction on \bst{}(of \figref{listex}) representing a \tab{} object.}
		\label{fig:trans-tree}
		
		\vspace{-10pt}
	\end{wrapfigure}
}

\begin{figure}[H]
	\centering
	\captionsetup{justification=centering}
	{\scalebox{0.5}{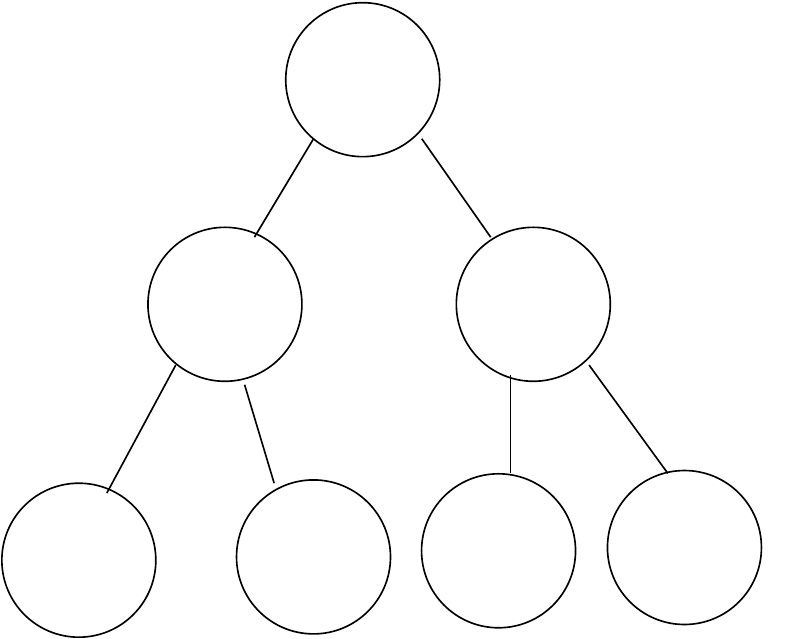}}
	\caption{T1 : A sample transaction on \lsl{} (of \figref{tree-exec}(i)) representing a \tab{} object.}
	\label{fig:trans-tree}
\end{figure}


Thus a transaction is a tree whose nodes are \mth{s} and leaves are events. Having informally explained a transaction, we formally define a transaction $T$ as the tuple $\langle \evts{T}, <_T\rangle$. Here $\evts{T}$ are all the read/write events (primitives) at level-0 of the transaction. $<_T$ is a total order among all the events of the transaction. For instance, the transaction $T_1$ of \figref{trans-tree} is: $\inv(l_{11}(ht, k_5))~ r_{111}(k_2, o_2) r_{112}(k_5, o_5)~ \rsp(l_{11}(ht, k_5, o_5))~  \inv(d_{12}(ht, k_2))~
r_{121}(k_2, o_2)~ w_{122}(k_2, o_2)~ \\\rsp(d_{12}(ht, k_2, o_2))$. Given all level-0 events, it can be seen that the level-1 \mth{s} and the transaction tree can be constructed. 

We denote the first and last events of a transaction $T_i$ as $\fevt{T_i}$ and $\levt{T_i}$. Given any other read/write event $rw$ in $T_i$, we assume that $\fevt{T_i} <_{T_i} rw <_{T_i} \levt{T_i}$. 

All the \mth{s} of $T_i$ are denoted as $\met{T_i}$. We assume that for any method $m$ in $\met{T_i}$, $\evts{m}$ is a subset of $\evts{T_i}$ and $<_m$ is a subset of $<_{T_i}$. Formally, $\langle \forall m \in \met{T_i}: \evts{m} \subseteq \evts{T_i} ~ \land <_m \subseteq <_{T_i} \rangle$.

We assume that if a transaction has invoked a \mth{}, then it does not invoke a new \mth{} until it gets the response of the previous one. Thus all the \mth{s} of a transaction can be ordered by $<_{T_i}$. Formally,  $(\forall m_{p}, m_{q} \in \met{T_i}: (m_{p} <_{T_i} m_{q}) \lor (m_{q} <_{T_i} m_{p}))\rangle$. 


\cmnt{
	\paragraph{Real-time Order \& Serial Histories:} Given a history $H$, $<_H$ orders all the events in $H$. Now, we define a precedence order on methods and transactions based on real-time ordering of events in the history $H$. Given two complete \mth{s} $m_{ij}, m_{pq}$ in $\met{H}$, we denote $m_{ij} \prec_H^{\mr} m_{pq}$ if $\rsp(m_{ij}) <_H \inv(m_{pq})$. Here \mr{} stands for method real-time order. It must be noted that all the \mth{s} of the same transaction are ordered. 
	
	Similarly, for two transactions $T_{i}, T_{p}$ in $\term{H}$, we denote $(T_{i} \prec_H^{\tr} T_{p})$ if $(\levt{T_{i}} <_H \fevt{T_{p}})$. Here \tr{} stands for transactional real-time order. 
	
	Thus, $\prec$ partially orders all the \mth{s} and transactions in $H$. It can be seen that if $H$ is sequential, then $\prec_H^{\mr}$ totally orders all the \mth{s} in $H$. Formally, $\langle (H \text{ is seqential}) \implies (\forall m_{ij}, m_{pq} \in \met{H}: (m_{ij} \prec_H^{\mr} m_{pq}) \lor (m_{pq} \prec_H^{\mr} m_{ij}))\rangle$. 
}

\vspace{1mm}
\noindent
\textbf{\textsl{Histories:}} A \emph{history} is a sequence of events belonging to different transactions. The collection of events is denoted as $\evts{H}$. Similar to a transaction, we denote a history $H$ as tuple $\langle \evts{H},<_H \rangle$ where all the events are totally ordered by $<_H$. The set of \mth{s} that are in $H$ is denoted by $\met{H}$. A \mth{} $m$ is \emph{incomplete} if $\inv(m)$ is in $\evts{H}$ but not its corresponding response event. Otherwise $m$ is \emph{complete} in $H$. 

Coming to transactions in $H$, the set of transactions in $H$ are denoted as $\txns{H}$. The set of committed (resp., aborted) transactions in $H$ is denoted by $\comm{H}$ (resp., $\aborted{H}$). The set of \emph{live} transactions in $H$ are those which are neither committed nor aborted.  On the other hand, the set of \emph{terminated} transactions are those which have either committed or aborted. 


We denote two histories $H_1, H_2$ as \emph{equivalent} if their events are the same, i.e., $\evts{H_1} = \evts{H_2}$. A history $H$ is qualified to be \emph{well-formed} if: (1) all the \mth{s} of a transaction $T_i$ in $H$ are totally ordered, i.e. a transaction invokes a \mth{} only after it receives a response of the previous \mth{} invoked by it (2) $T_i$ does not invoke any other \mth{} after it received an $\mathcal{A}$ response or after $\tryc(ok)$ \mth. We only consider \emph{well-formed} histories for \otm.

A \mth{} $m_{ij}$ ($j^{th}$ \mth of a transaction $T_i$) in a history $H$ is said to be \emph{isolated} or \emph{atomic} if for any other event $e_{pqr}$ belonging to some other \mth{} $m_{pq}$ (of transaction $T_p$) either $e_{pqr}$ occurs before $\inv(m_{ij})$ or after $\rsp(m_{ij})$. Here, $e_{pqr}$ stands for $r^{th}$ event of $m_{pq}$.

\begin{figure}[H]
	\centering
	\captionsetup{justification=centering}
	{\scalebox{0.5}{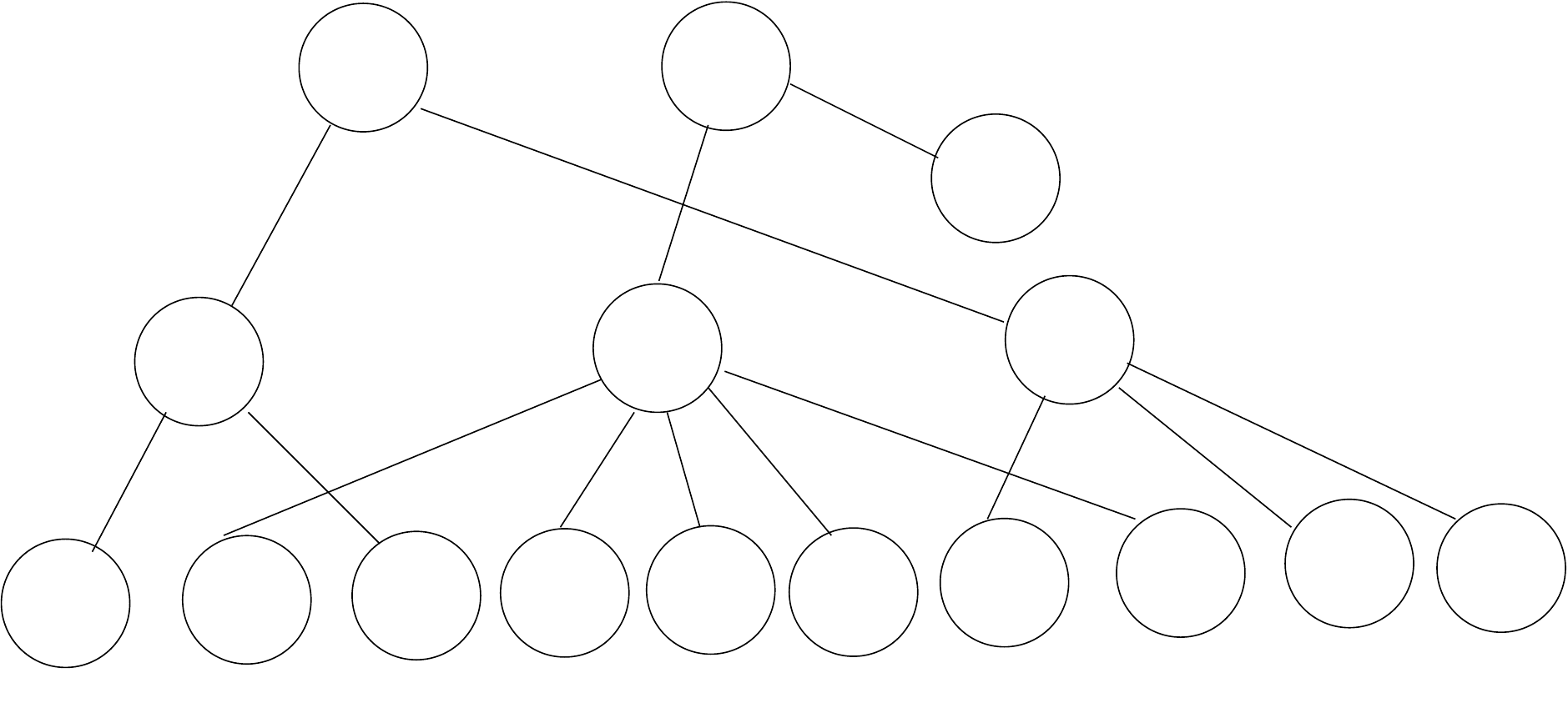}}
	\caption{H2 : A non-sequential History.}
	\label{fig:unseq}
\end{figure}
\vspace{1mm}
\noindent
\textbf{\textsl{Sequential Histories:}} A \mth{} $m_{ij}$ of a transaction $T_i$ in a history $H$ is said to be \emph{isolated} if for any other event $e_{pqr}$ belonging to some other \mth{} $m_{pq}$ (of transaction $T_p$) either $e_{pqr}$ occurs before $\inv(m_{ij})$ or after $\rsp(m_{ij})$. Formally, $\langle m_{ij} \in \met{H}: m_{ij} \text{ is isolated} \equiv (\forall m_{pq} \in \met{H}, \forall e_{pqr} \in m_{pq}: e_{pqr} <_H \inv(m_{ij}) \lor \rsp(m_{ij}) <_H e_{pqr})\rangle$. For instance in $H1$ shown in \figref{tree-exec}(ii), $d_{2}(k_2)$ is isolated. In fact all the \mth{s} of $H1$ are isolated.

Consider history $H2$ shown in \figref{unseq}. It can be seen that the all the three \mth{s} in $H2$, ($l_{11}, d_{21}, l_{12}$) are not isolated.

A history $H$ is said to be \emph{sequential} (term used in \cite{KuzPer:NI:TCS:2017, KR:2011:OPODIS}) or \emph{linearized} \cite{HerlWing:1990:TPLS} if all the methods in it are complete and isolated. Thus, it can be seen that $H1$ is sequential whereas $H2$ is not. From now onwards, most of our discussion would relate to sequential histories. 

Since in sequential histories all the \mth{s} are isolated, we treat each \mth{} as whole without referring to its inv and rsp events. For a sequential history $H$, we construct the \emph{completion} of $H$, denoted $\overline{H}$, by inserting $\trya_k(\mathcal{A})$ immediately after the last \mth{} of every transaction $T_k \in \incomp{H}$. Since all the \mth{s} in a sequential history are complete, this definition only has to take care of completing transactions.

Consider a sequential history $H$. Let $m_{ij}(ht, k, v/nil)$ be the first \mth{} of $T_i$ in $H$ operating on the key $k$. Since all the \mth{s} of a transaction are sequential and ordered, we can clearly identify the first \mth{} of $T_i$ on key $k$. Then, we denote $m_{ij}(ht, k, v)$ as $\fkmth{\langle ht, k \rangle}{T_i}{H}$. For a \mth{} $m_{ix}(ht, k, v)$ which is not the first \mth{} on $\langle ht, k \rangle$ of $T_i$ in $H$, we denote its previous \mth{} on $k$ of $T_i$ as $m_{ij}(ht, k, v) = \pkmth{m_{ix}}{T_i}{H}$.\\

\begin{figure}
	\centering
	\captionsetup{justification=centering}
	{\scalebox{0.5}{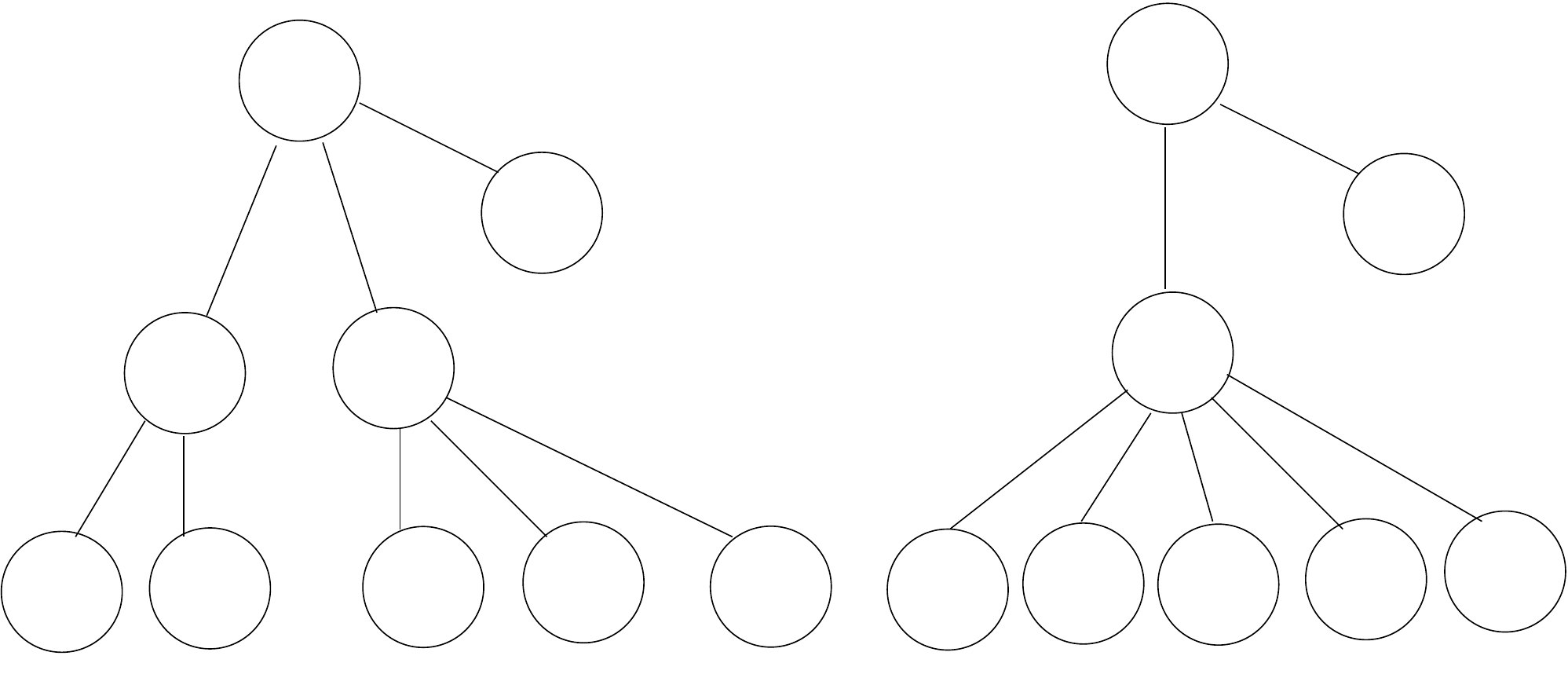}}
	\caption{A serial History}
	\label{fig:serial}
\end{figure}

\cmnt{
	\begin{figure}[H]
		\centering
		\subfloat[H2 : A non-sequential History.]{\scalebox{0.4}{\input{figs/nonseq.pdf_t}}\label{fig:unseq}}
	\end{figure}
	
	\begin{figure}[H]
		\centering
		\subfloat[A serial History.]{\scalebox{0.4}{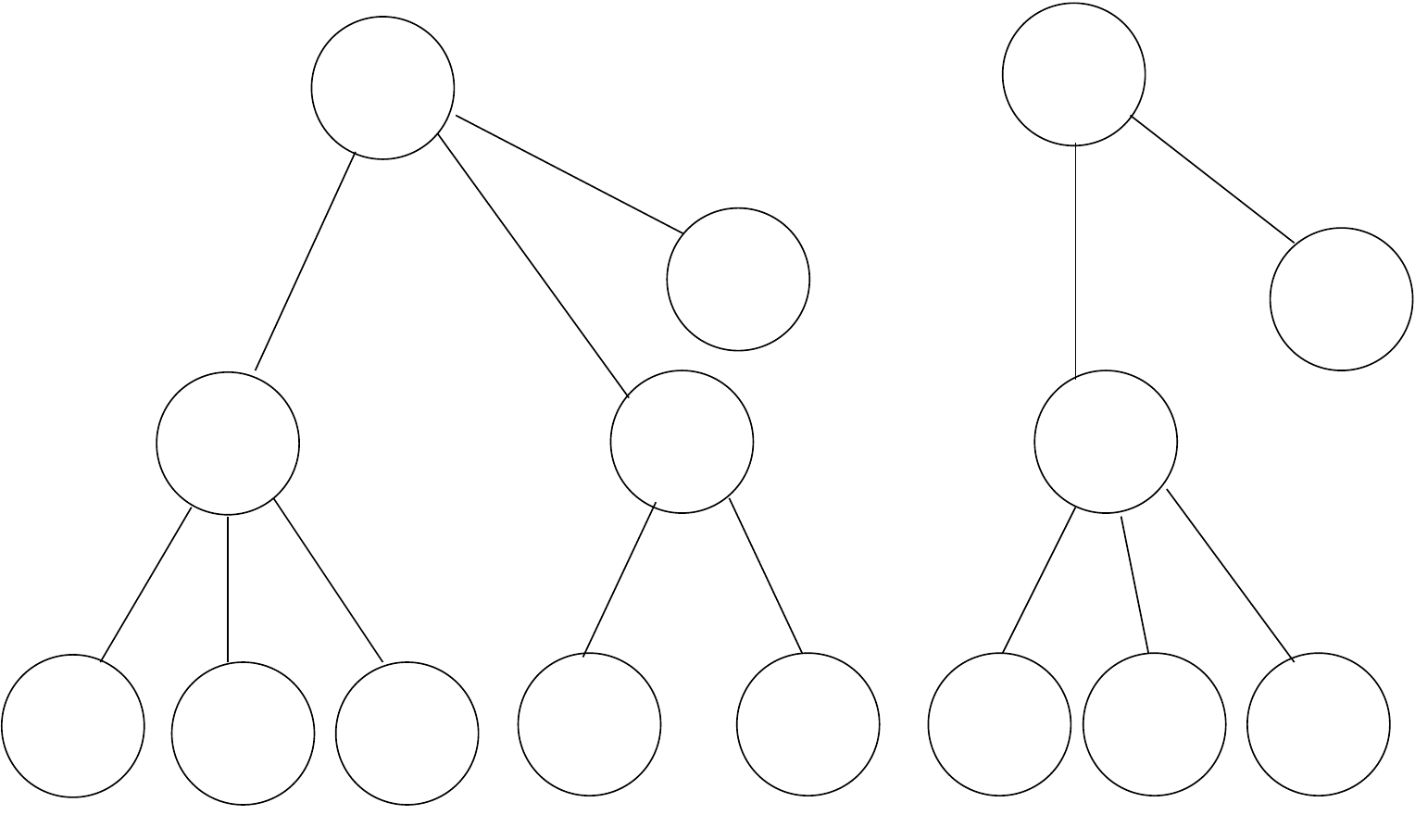}\label{fig:serial}}
		\caption{serial and non sequential History.}
	\end{figure}
}

\vspace{1mm}
\noindent
\textbf{\textsl{Real-time Order \& Serial Histories:}} Given a history $H$, $<_H$ orders all the events in $H$. For two complete \mth{s} $m_{ij}, m_{pq}$ in $\met{H}$, we denote $m_{ij} \prec_H^{\mr} m_{pq}$ if $\rsp(m_{ij}) <_H \inv(m_{pq})$. Here \mr{} stands for method real-time order. It must be noted that all the \mth{s} of the same transaction are ordered. Similarly, for two transactions $T_{i}, T_{p}$ in $\term{H}$, we denote $(T_{i} \prec_H^{\tr} T_{p})$ if $(\levt{T_{i}} <_H \fevt{T_{p}})$. Here \tr{} stands for transactional real-time order. 

\cmnt{
Thus, $\prec$ partially orders all the \mth{s} and transactions in $H$. It can be seen that if $H$ is sequential, then $\prec_H^{\mr}$ totally orders all the \mth{s} in $H$. Formally, $\langle (H \text{ is seqential}) \implies (\forall m_{ij}, m_{pq} \in \met{H}: (m_{ij} \prec_H^{\mr} m_{pq}) \lor (m_{pq} \prec_H^{\mr} m_{ij}))\rangle$. 
}

We define a history $H$ as \emph{serial} \cite{Papad:1979:JACM} or \emph{t-sequential} \cite{KR:2011:OPODIS} if all the transactions in $H$ have terminated and can be totally ordered w.r.t $\prec_{\tr}$, i.e. all the transactions execute one after the other without any interleaving. Intuitively, a history $H$ is serial if all its transactions can be isolated. Formally, $\langle (H \text{ is serial}) \implies (\forall T_{i} \in \txns{H}: (T_i \in \term{H}) \land (\forall T_{i}, T_{p} \in \txns{H}: (T_{i} \prec_H^{\tr} T_{p}) \lor (T_{p} \prec_H^{\tr} T_{i}))\rangle$. Since all the methods within a transaction are ordered, a serial history is also sequential. \figref{serial} shows a serial history. Here all the \emph{layer-1} \mths are isolated thus the involved transaction can be ordered as $T_1$ followed by $T_2$. Thus we attain a serial order $T_1$, $T_2$.

\cmnt{

\begin{figure}[!tbp]
	\centering
	\subfloat[H2 : A non-sequential History.]{\scalebox{0.4}{\input{figs/nonseq.pdf_t}}\label{fig:unseq}}
	\hfill
	\subfloat[A serial History.]{\scalebox{0.4}{\input{figs/serial-hist.pdf_t}}\label{fig:serial}}
	\caption{serial and non sequential History.}
\end{figure}
}

\section{Correctness of \otm: Opacity \& Conflict Opacity}
\label{sec:opty}
In this section, we define the correctness of \otm by extending \opty\cite{GuerKap:2008:PPoPP}. We then define a tractable subclass of \opty, \coopty which is defined using conflict like \csr\cite{WeiVoss:2002:Morg} in databases. We start with \legality{} and \opty. 

\vspace{-3mm}
\subsection{Legal Histories \& Opacity}
\label{subsec:legal}
\vspace{-2mm}
In this subsection, we start with defining legal histories. To simplify our analysis, we assume that there exists an initial transaction $T_0$ that invokes $\tdel$ \mth on all the keys of all the hash-tables used by any transaction. 

We define \emph{\legality{}} of \rvmt{s}  (\tdel{} \& \tlook{}) on sequential histories which we later use to define correctness criterion. Consider a sequential history $H$ having a \rvmt{} $\rvm_{ij}(ht, k, v)$ (with $v \neq nil$) belonging to transaction $T_i$. We define this \rvm \mth{} to be \emph{\legal} if: 
\vspace{-1mm}
\begin{enumerate}
	\item[LR1] \label{step:leg-same} If the $\rvm_{ij}$ is not first \mth of $T_i$ to operate on $\langle ht, k \rangle$ and $m_{ix}$ is the previous \mth of $T_i$ to operate on $\langle ht, k \rangle$. Formally, $\rvm_{ij} \neq \fkmth{\langle ht, k \rangle}{T_i}{H}$ $\land (m_{ix}(ht, k, v') = \pkmth{\langle ht, k \rangle}{T_i}{H})$ (where $v'$ could be nil). Then,
	\begin{enumerate}
		\setlength\itemsep{0em}
		\item if $m_{ix}(ht, k, v')$ is a \tins{} \mth i.e. $\tins_{ix}(ht, k, v')$ then $v = v'$. 
		\item if $m_{ix}(ht, k, v')$ is a \tlook{} \mth i.e. $\tlook_{ix}(ht, k, v')$ then $v = v'$. 
		\item if $m_{ix}(ht, k, v')$ is a \tdel{} \mth i.e. $\tdel_{ix}(ht, k, v'/nil)$ then $v = nil$. 
	\end{enumerate}
	
	In this case, we denote $m_{ix}$ as the last update \mth{} of $\rvm_{ij}$, i.e.,  $m_{ix}(ht, k, v') = \\\lupdt{\rvm_{ij}(ht, k, v)}{H}$. 
	
	\item[LR2] \label{step:leg-ins} If $\rvm_{ij}$ is the first \mth{} of $T_i$ to operate on $\langle ht, k \rangle$ and $v$ is not nil. Formally, $\rvm_{ij}(ht, k, v) = \fkmth{\langle ht, k \rangle}{T_i}{H} \land (v \neq nil)$. Then,
	\begin{enumerate}
		\setlength\itemsep{0em}
		\item There is a \tins{} \mth{} $\tins_{pq}(ht, k, v)$ in $\met{H}$ such that $T_p$ committed before $\rvm_{ij}$. Formally, $\langle \exists \tins_{pq}(ht, k, v) \in \met{H} : \tryc_p \prec_{H}^{\mr} \rvm_{ij} \rangle$. 
		\item There is no other update \mth{} $up_{xy}$ of a transaction $T_x$ operating on $\langle ht, k \rangle$ in $\met{H}$ such that $T_x$ committed after $T_p$ but before $\rvm_{ij}$. Formally, $\langle \nexists up_{xy}(ht, k, v'') \in \met{H} : \tryc_p \prec_{H}^{\mr} \tryc_x \prec_{H}^{\mr} \rvm_{ij} \rangle$. 		
	\end{enumerate}
	
	In this case, we denote $\tryc_{p}$ as the last update \mth{} of $\rvm_{ij}$, i.e.,  $\tryc_{p}(ht, k, v)$= $\lupdt{\rvm_{ij}(ht, k, v)}{H}$.
	
	\item[LR3] \label{step:leg-del} If $\rvm_{ij}$ is the first \mth of $T_i$ to operate on $\langle ht, k \rangle$ and $v$ is nil. Formally, $\rvm_{ij}(ht, k, v) = \fkmth{\langle ht, k \rangle}{T_i}{H} \land (v = nil)$. Then,
	\begin{enumerate}
		\setlength\itemsep{0em}
		\item There is \tdel{} \mth{} $\tdel_{pq}(ht, k, v')$ in $\met{H}$ such that $T_p$ (which could be $T_0$ as well) committed before $\rvm_{ij}$. Formally, $\langle \exists \tdel_{pq}\\(ht, k,$ $ v') \in \met{H} : \tryc_p \prec_{H}^{\mr} \rvm_{ij} \rangle$. Here $v'$ could be nil. 
		\item There is no other update \mth{} $up_{xy}$ of a transaction $T_x$ operating on $\langle ht, k \rangle$ in $\met{H}$ such that $T_x$ committed after $T_p$ but before $\rvm_{ij}$. Formally, $\langle \nexists up_{xy}(ht, k, v'') \in \met{H} : \tryc_p \prec_{H}^{\mr} \tryc_x \prec_{H}^{\mr} \rvm_{ij} \rangle$. 		
	\end{enumerate}
	In this case similar to \stref{leg-ins}, we denote $\tryc_{p}$ as the last update \mth{} of $\rvm_{ij}$, i.e., $\tryc_{p}(ht, k, v)$ $= \lupdt{\rvm_{ij}(ht, k, v)}{H}$. 
\end{enumerate}
We assume that when a transaction $T_i$ operates on key $k$ of a \tab{} $ht$, the result of this \mth is stored in \emph{local logs} of $T_i$ for later \mth{s} to reuse. Thus, only the first \rvmt{} operating on $\langle ht, k \rangle$ of $T_i$ accesses the shared-memory. The other \rvmt{s} of $T_i$ operating on $\langle ht, k \rangle$ do not access the shared-memory and they see the effect of the previous \mth{} from the \emph{local logs}. This idea is utilized in LR1. With reference to LR2 and LR3, it is possible that $T_x$ could have aborted before $\rvm_{ij}$. For LR3, since we are assuming that transaction $T_0$ has invoked a \tdel{} \mth{} on all the keys used of all \tab{} objects, there exists at least one \tdel{} \mth{} for every \rvmt on $k$ of $ht$. 
Coming to \tins{} \mth{s}, since a \tins{} \mth{} always returns $ok$ as they overwrite the node if already present therefore they always take effect on the $ht$. 
We explain the above formalized legality definitions with help of intuitive examples in following text:\\
\noindent
\textbf{Legality through examples:}
LR1 says that, for a given key (node), if $\rvmt{}$ is not the first method on the key in a transaction, then it will observe the value returned by the previous method of the same transaction. We show this in \figref{nostm12} for lookups, but wlog same behaviou holds for delete \mth as well. In \figref{nostm12}(i), previous method for $l_{ij}(ht, k_5, v_5)$ of transaction $T_i$ on key $k_5$ is $i_{ix}(ht, k_5, v_5)$. So, $l_{ij}(ht, k_5, v_5)$ will return the value $v_5$ which will be inserted by previous method $i_{ix}(ht, k_5, v_5)$. Same mechanism will be followed in \figref{nostm12}(ii) and \figref{nostm12}(iii) where previous method is a lookup and delete, respectively.
\begin{figure}[H]
	\centering
	\captionsetup{justification=centering}
	\centerline{\scalebox{0.5}{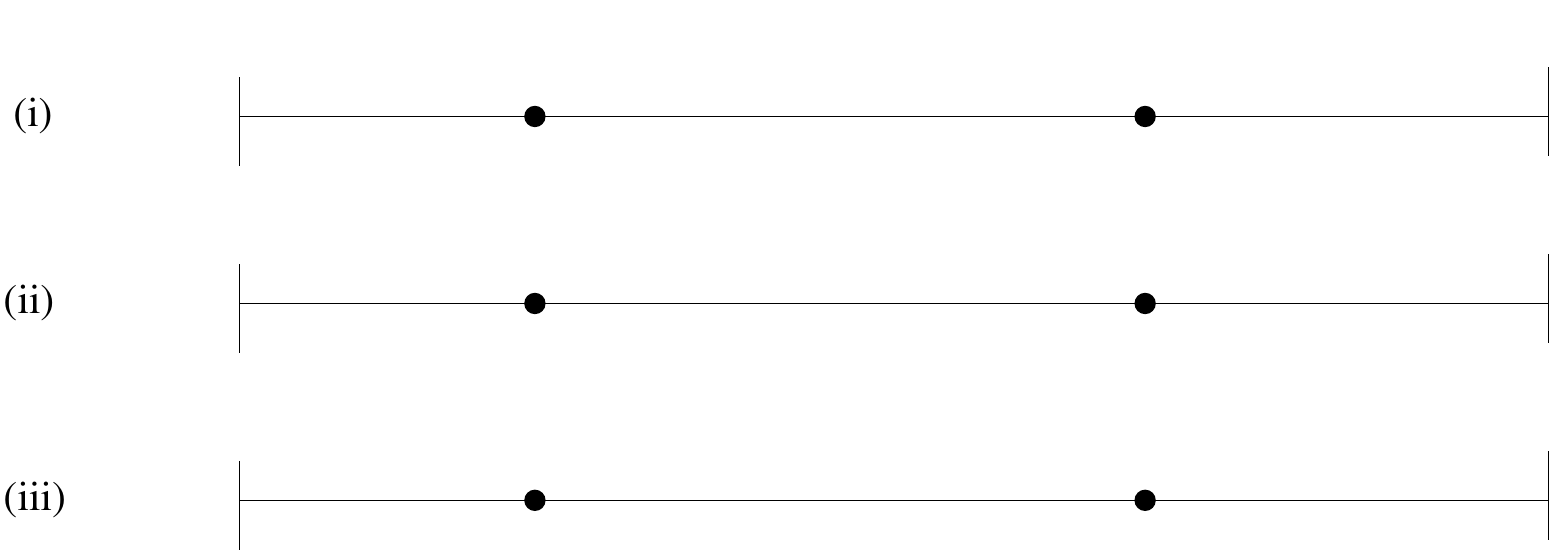}}
	\caption{Explanation for LR1}
	\label{fig:nostm12}
\end{figure}
LR2 says that, for a given shared key, if $\rvmt{}$ is the first method of the key in a transaction and it's value is not null then the previous closest method of committed transaction should be an insert on the key. In \figref{nostm13}, previous closest method for $l_{ij}(ht, k, v_p)$ of transaction $T_i$ on same key $k$ is $i_{pq}(ht, k, v_p)$ of transaction $T_p$. So, $l_{ij}(ht, k, v_p)$ will return the $v_p$ which has been inserted by $i_{pq}(ht, k, v_p)$ and there can't be any other transaction $\upmt{}$ working on the key $k$ between $T_p$ and $T_i$.
\begin{figure}[H]
	\centering
	\captionsetup{justification=centering}
	\centerline{\scalebox{0.5}{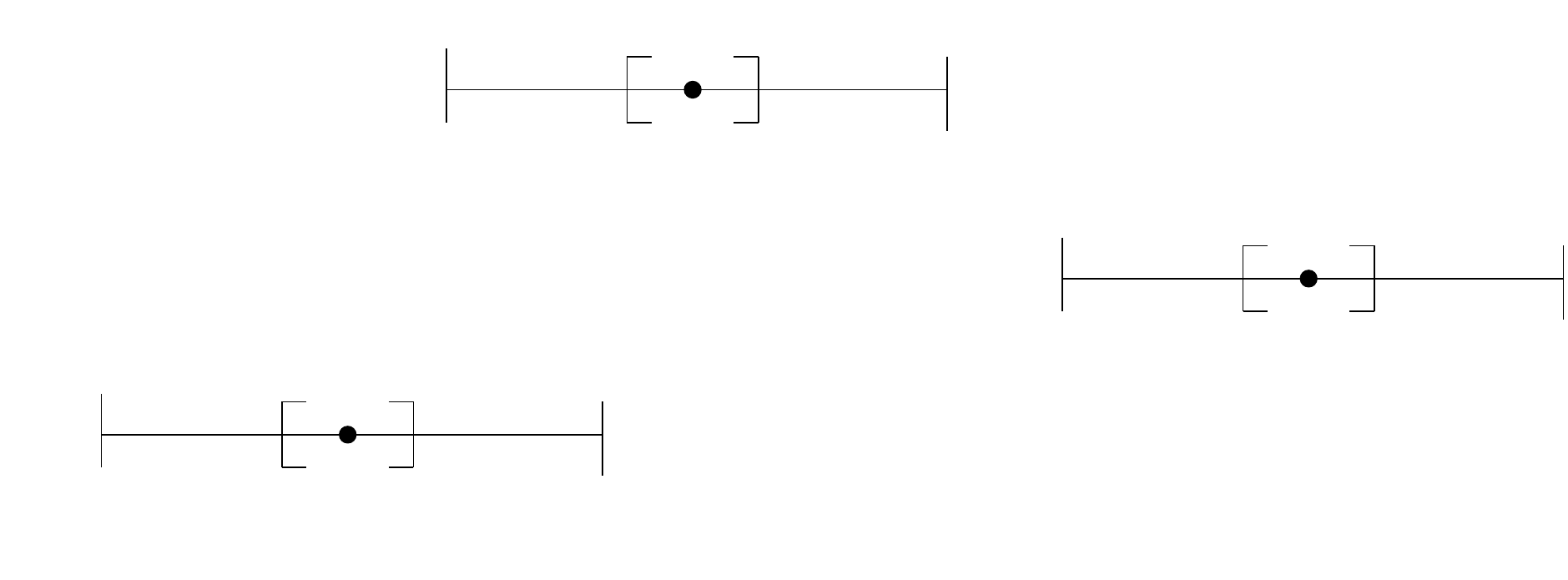}}
	\caption{Explanation for LR2.}
	\label{fig:nostm13}
\end{figure}

Finally LR3 says that, for a given shared key, if $\rvmt{}$ is the first method of the key in a transaction and it's value is null then the previous closest method of committed transaction should be a delete on the key. In \figref{nostm14}, previous closest method for $l_{ij}(ht, k, v_p)$ of transaction $T_i$ on key $k$ is $d_{pq}(ht, k, v_p)$ of transaction $T_p$. So, $l_{ij}(ht, k, v_p)$ will return the $v_p$ which has been returned by $d_{pq}(ht, k, v_p)$ and there can't be any other transaction $\upmt{}$ working on the same key between $T_p$ and $T_i$.

\begin{figure}[H]
	\centering
	\captionsetup{justification=centering}
	\centerline{\scalebox{0.5}{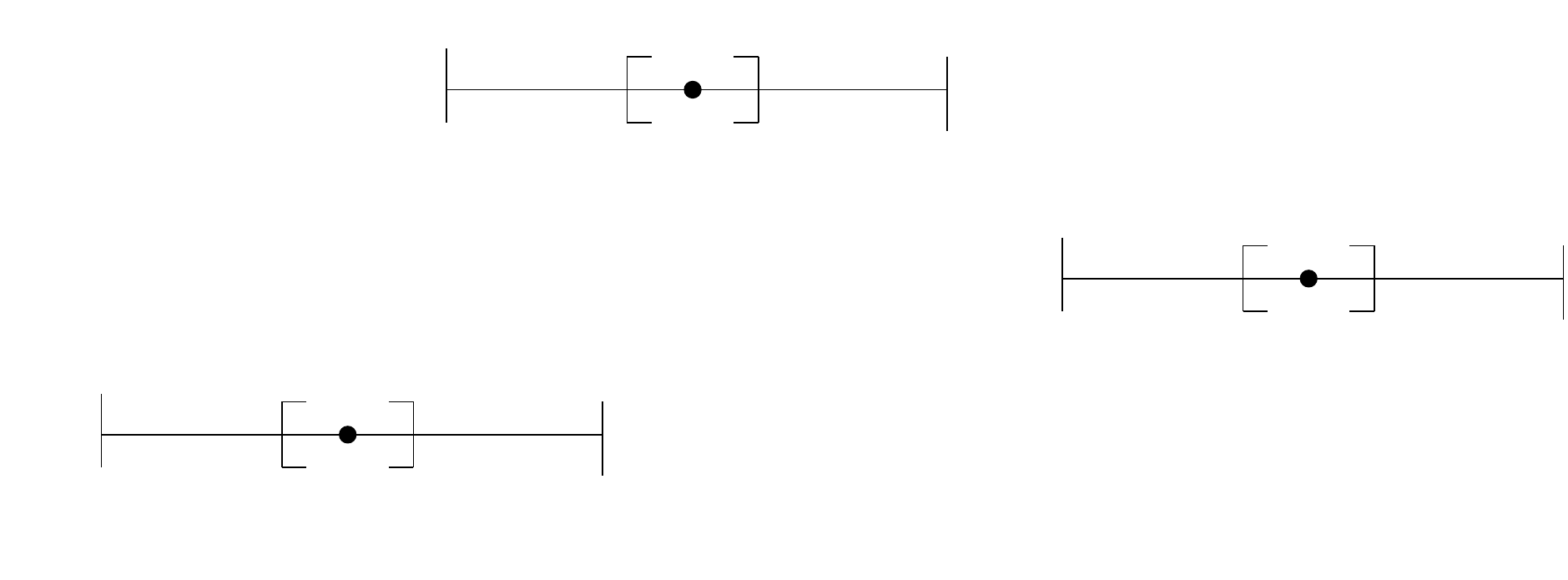}}
	\caption{Explanation for LR3}
	\label{fig:nostm14}
\end{figure}

Thus, we denote all \tins{} \mth{s} as \legal. We denote a sequential history $H$ as \emph{\legal} or \emph{linearized} \cite{HerlWing:1990:TPLS} if all its \rvm \mth{s} are \legal. While defining \legality{} of a history, we are only concerned about \rvm (\tlook{} and \tdel) \mth{s} since all \tins{} \mth{s} are by default \legal. 
History $H_2$ in \figref{legal} is legal because $l_{2}(ht, k_2, v_0)$ follows LR2, $d_{1}(ht, k_1, v_0)$ adheres to LR2 and $l_{2}(ht, k_1, nil)$ follows LR3. Thus all the \rvmt{} are legal.
\begin{figure}[H]
	\centering
	\captionsetup{justification=centering}
	{\scalebox{0.5}{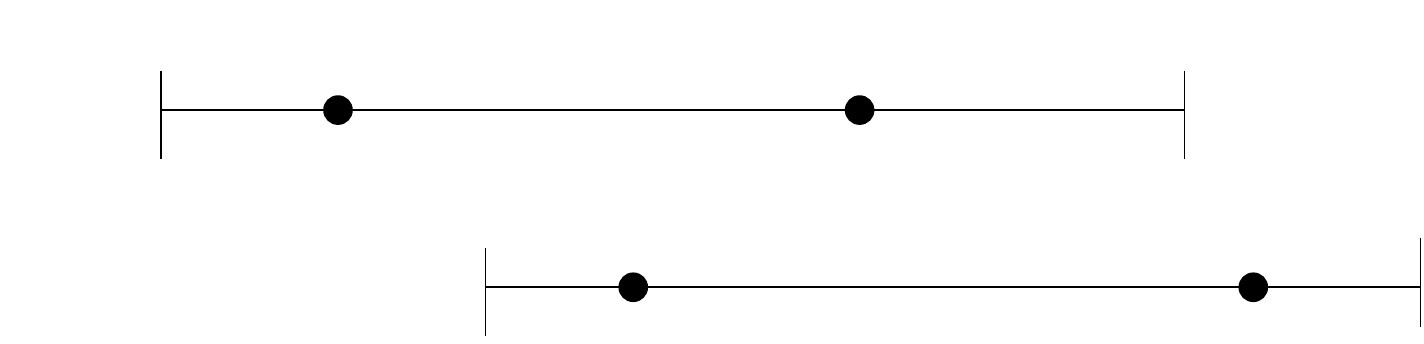}}
	\caption{Legal History H2}
	\label{fig:legal}
\end{figure}

We formally prove legality using \lemref{provingLegalityy}. \lemref{provingLegalityy} and then we finally show that \otm{} histories are \coop{} which is a subclass of opacity\cite{KuzPer:NI:TCS:2017}.

\vspace{1mm}
\noindent
\textbf{\textsl{Correctness-Criteria \& Opacity:}} A \emph{\cc} is a set of histories. A history $H$ satisfying a \cc{} has some desirable properties. A popular \cc{} is \emph{\opty} \cite{GuerKap:2008:PPoPP}. A sequential history $H$ is \opq{} if there exists a serial history $S$ such that: (1) $S$ is equivalent to $\overline{H}$, i.e. , $\evts{\overline{H}} = \evts{S}$ (2) $S$ is \legal{} and (3) $S$ respects the transactional real-time order of $H$, i.e., $\prec_H^{\tr} \subseteq \prec_S^{\tr}$. 
\vspace{-2mm}
\subsection{Conflict Notion \& Conflict-Opacity}
\label{sec:conflicts}
\vspace{-2mm}
Opacity is a popular \cc{} for STMs. But, as observed in \secref{intro}, it can be proved that verifying the membership of \opty similar to view-serializability (VSR) in databases is NP-Complete \cite{Papad:1979:JACM}. To circumvent this issue, researchers in databases have identified an efficient sub-class of VSR, called conflict-serializability or \csr, based on the notion of conflicts. The membership of \csr can be verified in polynomial time using conflict graph characterization. Along the same lines, we develop the notion of conflicts for \otm and identify a sub-class of \opty, \coopty. The proposed \cc{} is extension of the notion of conflict-opacity developed for \emph{\rwtm} by Kuznetsov \& Peri \cite{KuzPer:NI:TCS:2017}.

We say two transactions $T_i, T_j$ of a sequential history $H$ for \otm{} are in \emph{conflict} if atleast one of the following conflicts holds:
\vspace{-.2cm}
\begin{itemize}
	\setlength\itemsep{0em}
\item \textbf{tryC-tryC} conflict:(1) $T_i$ \& $T_j$ are committed and (2) $T_i$ \& $T_j$ update the same key $k$ of the \tab{}, $ht$, i.e., $(\langle ht,k \rangle \in \udset{T_i}) \land (\langle ht,k \rangle \in \udset{T_j})$, where $\udset{T_i}$ is update set of $T_i$. (3) $T_i$'s \tryc{} completed before $T_j$'s \tryc, i.e., $\tryc_i \prec_{H}^{\mr} \tryc_j$. 

\item \textbf{tryC-rv} conflict:(1) $T_i$ is committed (2) $T_i$ updates the key $k$ of \tab{}, $ht$. $T_j$ invokes a \rvmt{} $rvm_{jy}$ on the key same $k$ of \tab{} $ht$ which is the first \mth{} on $\langle ht, k \rangle$. Thus, $(\langle ht,k \rangle \in \udset{T_i}) \land (\rvm_{jy}(ht, k, v) \in \rvset{T_j}) \land (\rvm_{jy}(ht, k, v) = \fkmth{\langle ht, k \rangle}{T_j}{H})$, where $\rvset{T_j}$ is return value set of $T_j$. (3) $T_i$'s \tryc{} completed before $T_j$'s \rvm, i.e., $\tryc_i \prec_{H}^{\mr} \rvm_{jy}$.

\item \textbf{rv-tryC} conflict:(1) $T_j$ is committed (2) $T_i$ invokes a \rvmt on the key same $k$ of \tab{} $ht$ which is the first \mth{} on $\langle ht, k \rangle$. $T_j$ updates the key $k$ of the \tab{}, $ht$. Thus, $(\rvm_{ix}(ht, k, v) \in \rvset{T_i}) \land (\rvm_{ix}(ht, k, v) = \fkmth{\langle ht, k \rangle}{T_i}{H}) \land (\langle ht,k \rangle \in \udset{T_j})$ (3) $T_i$'s \rvm completed before $T_j$'s \tryc, i.e., $\rvm_{ix} \prec_{H}^{\mr} \tryc_j$. 
\end{itemize}

A \rvmt{} $\rvm_{ij}$ conflicts with a \tryc{} \mth{} only if $\rvm_{ij}$ is the first \mth{} of $T_i$ that operates on \tab{} with a given key. Thus the conflict notion is defined only by the \mth{s} that access the shared memory. $(\tryc_i{}, \tryc_j{})$, $(\tryc_i{}, \tlook_j)$, $(\tlook_i, \tryc_j{})$, $(\tryc_i{}, \tdel_j)$ and $(\tdel_i, \tryc_j{})$ can be the possible conflicting \mth{s}. For example, consider the history $H5: l_1(ht, k_1, NULL) l_2(ht, k_2$ $, NULL) i_2(ht, k_1, v_1) i_1(ht, k_4,$ $v_1) c_1 i_3(ht, k_3\\, v_3) c_3 d_2(ht, k_4, v_1) c_2 l_4(ht, k_4, NULL)$ $i_4(ht, k_2, v_4) c_4$ in \figref{netys}. $\langle l_1(ht, k_1, NULL), i_3(ht, k_1, v_1)\rangle$ and $\langle l_2(ht, k_2, NULL), i_4$ $(ht, k_2, v_4) \rangle$ are a conflict of type \emph{rv-tryC}. Conflict type of $ \langle i_1(ht, k_4, v_1), \\d_2(ht, k_4, v_1) \rangle$ and $ \langle i_1(ht, k_4, v_1),$ $l_4(ht, k_4, NULL) \rangle$ are \emph{tryC-tryC} and \emph{tryC-rv} respectively. 
\begin{figure}[H]
	\centering
	\scalebox{.5}{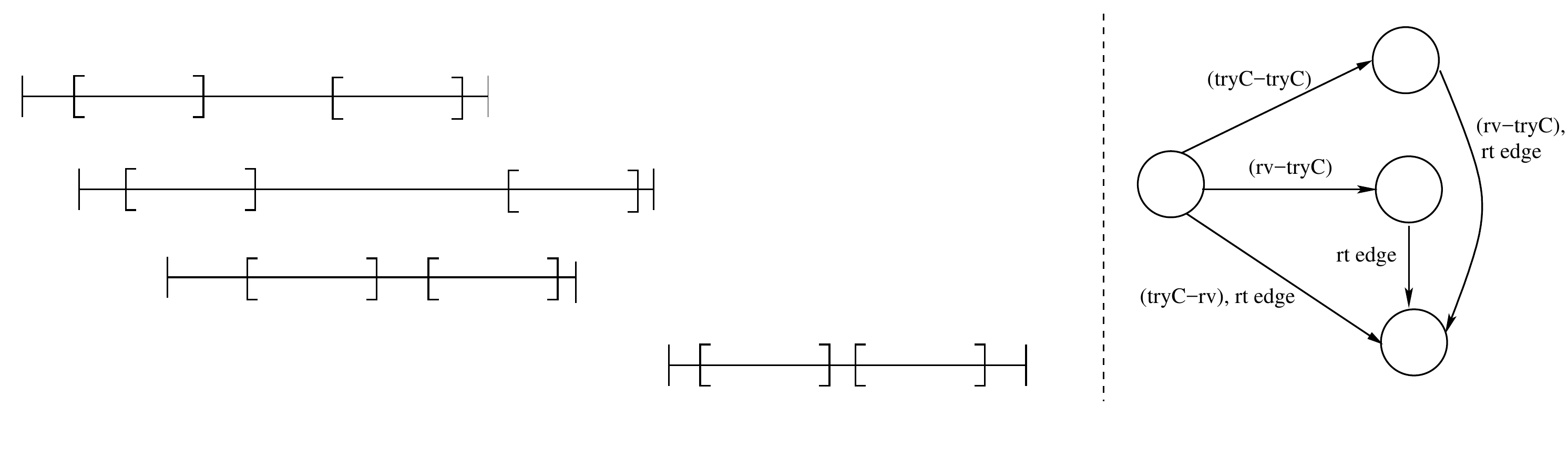}
	\caption{Graph Characterization of history $H5$}
	\label{fig:netys}
\end{figure}

\noindent
\textbf{Conflict Opacity:} Using this conflict notion, we can now define \coopty. A sequential history $H$ is conflict-opaque (or co-opaque) if there exists a serial history $S$ such that: 
\begin{enumerate}[nosep]
    \item $S$ is equivalent to $\overline{H}$, i.e. , $\evts{\overline{H}} = \evts{S}$, 
    \item $S$ is \legal{},
    \item $S$ respects the transactional real-time order of $H$, i.e., $\prec_H^{\tr} \subseteq \prec_S^{\tr}$ and
    \item S preserves conflicts (i.e. $\prec^{CO}_{H}\subseteq\prec^{CO}_{S}$). 
\end{enumerate}

\noindent Thus from the above definition, it can be seen that any history that is \coop is also \opq. 
\ignore{
\vspace{-.4cm}
\begin{definition}
	\label{def:co-opacity}
	\texttt{Co-opacity :} 	A sequential history $H$ is conflict-opaque (or co-opaque) if there exists a serial history $S$ such that: (1) $S$ is equivalent to $\overline{H}$, i.e. , $\evts{\overline{H}} = \evts{S}$ (2) $S$ is \legal{} and (3) $S$ respects the transactional real-time order of $H$, i.e., $\prec_H^{\tr} \subseteq \prec_S^{\tr}$ and (4) S preserves conflicts (i.e. $\prec$$^{CO}_{H}$ $\subseteq$ $\prec$$^{CO}_{S}$) \cite{KuzPer:NI:TCS:2017}. 
\end{definition}
\vspace{-.5cm}
}

\vspace{1mm}
\noindent
\textbf{Graph Characterization:} We now develop a graph characterization of \coopty. For a sequential history $H$, we define \emph{conflict-graph} of $H$, $\cg{H}$ as the pair $(V, E)$ where $V$ is the set of $\txns{H}$ and \emph{E} can be of following types: 
\begin{enumerate}[nosep]
\item \emph{conflict edges:} \{($T_i$, $T_j$) : ($T_i$, $T_j$) $\in$ \conf{H}\} where, \conf{H} is an ordered pair of transactions such that the transactions have one of the above pair of conflicts. 
\item \emph{real-time edge(or rt edge):} \{($T_i$, $T_j$): Transaction $T_i$ precedes $T_j$ in real-time, i.e., $T_i \prec_{H}^{\tr} T_j $\}.
\end{enumerate}

\ignore{
\vspace{-.2cm}
\begin{itemize}
	\item \emph{conflict edges:} \{($T_i$, $T_j$) : ($T_i$, $T_j$) $\in$ \conf{H}\} where, \conf{H} is an ordered pair of transactions such that the transactions have one of the above pair of conflicts. 

	\item \emph{real-time edge(or rt edge):} \{($T_i$, $T_j$) : $T_i \prec_{H}^{\tr} T_j $\}
\end{itemize}
}
Now, we have the following theorem which explains how graph characterization is useful.

\ignore{
\begin{lemma}
\otm{} histories are sequential and legal.
\vspace{-0.4cm}
\end{lemma}

In the accompanying report\cite{DBLP:journals/corr/abs-1709-00681} we show that at operational level \otm{} histories are linearizable and the conflicts \mth{} can be ordered in such a way that the generated sequential histories are legal.
}
\begin{theorem}
\label{thm:coop-graph}
A legal \otm{} history $H$ is \coop{} iff CG(H) is acyclic.
\end{theorem}
\noindent Using this framework, we next develop \otm using the notion of \bto. We show the correctness of the proposed algorithm by showing that all conflict graph of the histories generated by it are acyclic. 

\section{\otm{}}
\label{sec:htostm}

We design \otm{} a concurrent closed addressed \tab{} using above explained legality and conflict notion. The \otm{} exports \npbegin{\emph{()}}, \npins{}, \npdel{}, \npluk{} and \nptc{} and has $m$ number of buckets, which we refer to as size of the \tab{}. The main part of interest from concurrency perspective is each bucket of the \tab{} implemented as \lsl(lazy red-blue list), the shared memory data structure.

\subsection{Lazyrb-list}
It is a linked structure with immutable $head$ and $tail$ sentinel nodes of the form of a tuple \emph{$\langle$ key, value, lock, marked, max\_ts, rl, bl $\rangle$} representing a node. The $key$ represents unique id of the node so that a transaction could differentiate between two nodes. The $key$ values may range from $-\infty$ ( key of head node ) to $+\infty$ ( key of tail node ). The $value$ field may accommodate any type ranging from a basic integer to a complex class type. The $marked$ field is to have lazy deletion as popular in lazylists\cite{Heller+:LazyList:PPL:2007, Herlihy:ArtBook:2012} and $lock$ to implement exclusive access to the node. 

\cmnt
{
\begin{wrapfigure}{L}{0.4\textwidth}
    \begin{minipage}{0.4\textwidth}
    \begin{tcolorbox}
        \begin{algorithmic}[H]
            \State key, value;
	        \State lock, marked;
	        \State struct max\_ts;
	        \State node* rl; node* bl;
	    \end{algorithmic}
\end{tcolorbox}
\caption{\lsl{} node structure}
\label{fig:nodestruct}
\end{minipage}
\end{wrapfigure}
}

Lazyrb-list node have two links - $\bn$ (blue links) and $\rn$ (red links). First, the nodes which are not marked (not deleted) are reachable by \bn{} from the head. Second, the nodes which are marked (i.e. logically deleted) and are only reached by \rn{}. Thus, the name \lsl{}. All marked nodes are reachable via $\rn$ and all the unmarked nodes are reachable via $\bn$ \& $\rn$ from the head. Thus nodes reachable by $\bn$ are the subset of the nodes reachable by $\rn$. Every node of \lsl{} is in increasing order of its key.

Furthermore, every \lsl{} node also has a tuple $max\_ts\langle insert, delete, lookup \rangle$ to record the time-stamp of the transaction which most recently executed some \mth. Augmenting the underlying shared data structure with time-stamps help in identifying conflicts which can cause a cycle in the execution and hence violate \coopty{}\cite{KuzPer:NI:TCS:2017}. This is captured by the graph characterization of a generated history as discussed in \figref{netys} which implies that cyclic conflicts leads to non \coop{} execution.

\begin{figure}[H]
	\centering
	\begin{minipage}[b]{0.49\textwidth}
		\scalebox{.43}{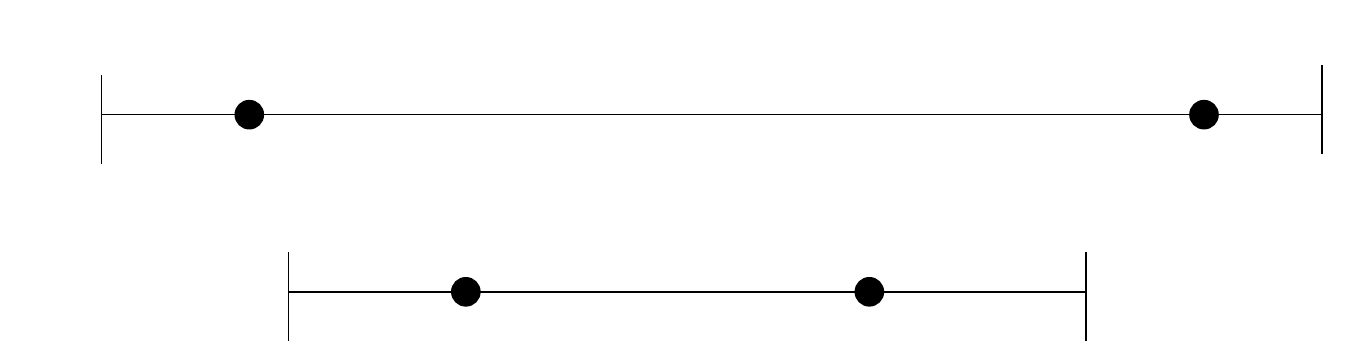}
		\centering
		\caption{History H is not \coop{}}
		\label{fig:nostm25}
	\end{minipage}
	\hfill
	\begin{minipage}[b]{0.49\textwidth}
		\centering
		\scalebox{.43}{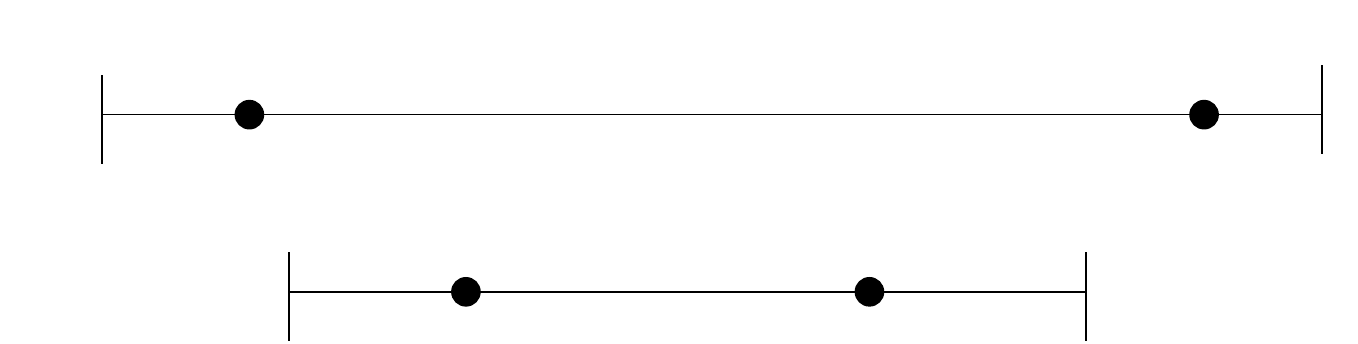}
		\caption{\coop{} History H1}
		\label{fig:nostm26}
	\end{minipage}   
\end{figure}

Now, we explain why we need to maintain deleted nodes through \figref{nostm25} and \ref{fig:nostm26}. History H shown in \figref{nostm25} is not \coop{} because there is no serial execution of T1 \& T2 that can be shown \coop{}. In order to make it \coop{} $l_1(ht, k_1, Nil)$ needs to be aborted. And $l_1(ht, k_1, Nil)$ can only be aborted if \otm{} scheduler knows that a conflicting operation $d_2(ht, k_1, v_0)$ has already been scheduled and thus violating \coopty. One way to have this information is that if the node represented by $k_1$ records the time-stamp of the delete \mth so that the scheduler realizes the violation of the time-order\cite{WeiVoss:2002:Morg} and aborts $l_1(ht, k_1, Nil)$ to ensure \coopty{}.

Thus, to ensure correctness, we need to maintain information about the nodes deleted from the \tab{}. This can be achieved by only marking node deleted from the list of \tab{}. But do not unlink it such that the marked node is still part of the list. This way, the information from deleted nodes can be used for ensuring \coopty{}. In this case, after aborting $l_1(ht, k_1)$, we get that the history is \coop{} with $T1$ \& $T2$ being the equivalent serial history as shown in \figref{nostm26}. The deleted keys (nodes with marked field set) can be reused if another transaction comes \& inserts the same key back.

\begin{figure}[H]
	\centering
		\begin{minipage}[b]{0.49\textwidth}
		\centering
		\scalebox{.41}{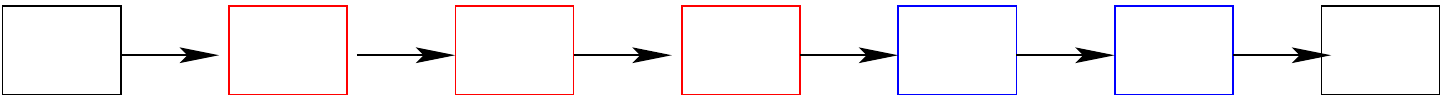}
		\caption{Searching $k_8$ over lazylist}
		\label{fig:nostm2}
	\end{minipage}   
	\hfill
	\begin{minipage}[b]{0.49\textwidth}
		\scalebox{.41}{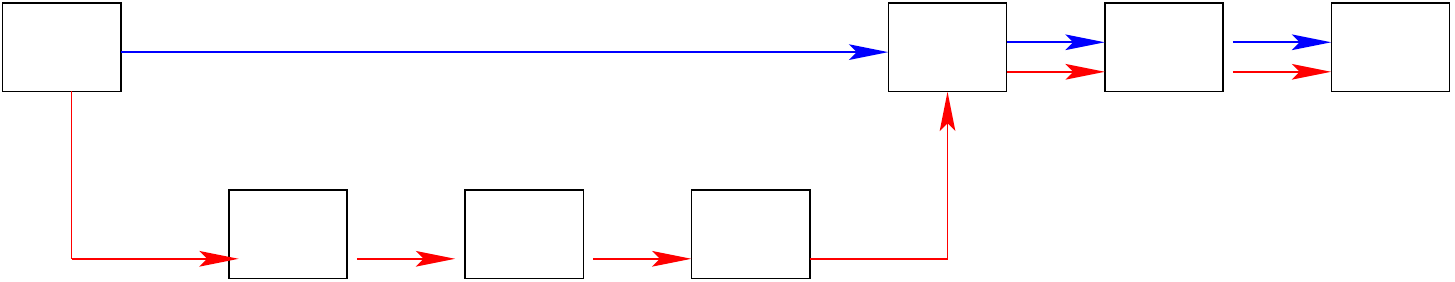}
		\centering
		\caption{Searching $k_8$ over \lsl{}}
		\label{fig:nostm}
	\end{minipage}
\end{figure}

But, the major hindrance in maintaining the deleted nodes as part of the ordinary lazy-list is that it would reduce search efficiency of the data structure. For example, in \figref{nostm2} searching $k_8$ would unnecessary cause traversal over marked ( marked for lazy deletion ) nodes represented by $k_1, k_3$ and $k_6$. We solve this problem in \lsl{} by using two pointers. 1) \bn (blue link): used to traverse over the actual inserted nodes and 2) \rn (red link) used to traverse over the deleted nodes. Hence, in \figref{nostm} to search for $k_8$ we can directly use \bn{} saving significant search computations.
\cmnt{
\begin{multicols}{3}
\begin{tcolorbox}
\begin{algorithmic}
    \State \textbf{txlog:}
    \State 
    \State t\_id; 
	\State tx\_status;
	\State vector $\langle key, le \rangle$;
\end{algorithmic}
\end{tcolorbox}
\begin{tcolorbox}
\begin{algorithmic}
    \State \textbf{max\_ts data structure:}
    \State lookup;
	\State insert;
	\State delete;
\end{algorithmic}
\end{tcolorbox}
\begin{tcolorbox}
\begin{algorithmic}
    \State \textbf{$le$ structure:}
    \State obj\_id,key,value;	
	\State preds, currs;
	\State op\_status;	
	\State operation\_name;
\end{algorithmic}
\end{tcolorbox}
\end{multicols}
}
A question may arise that how would we maintain the time-stamp of a node which has not yet been inserted? Such a case arises when \npluk{} or \npdel{} is invoked from $\rvmt{}$, and node corresponding to the key, say $k$ is not present in $\bn$ and $\rn$. Then the $\rvmt{}$ will create a node for key $k$ and insert it into underlying data structure as deleted (marked field set) node.

For example, lookup wants to search key $k_{10}$ in \figref{nostm} which is not present in the $\bn$ as well as $\rn$. Therefore, lookup method will create a new node corresponding to the key $k_{10}$ and insert it into $\rn$ (refer the \figref{nostm1}). So, we discuss in detail the invariants and properties of the \lsl{} and ensure that no duplicate nodes are inserted while proving the method level correctness in \secref{opnlevel}.

\begin{figure}
	\centering
	\scalebox{.4}{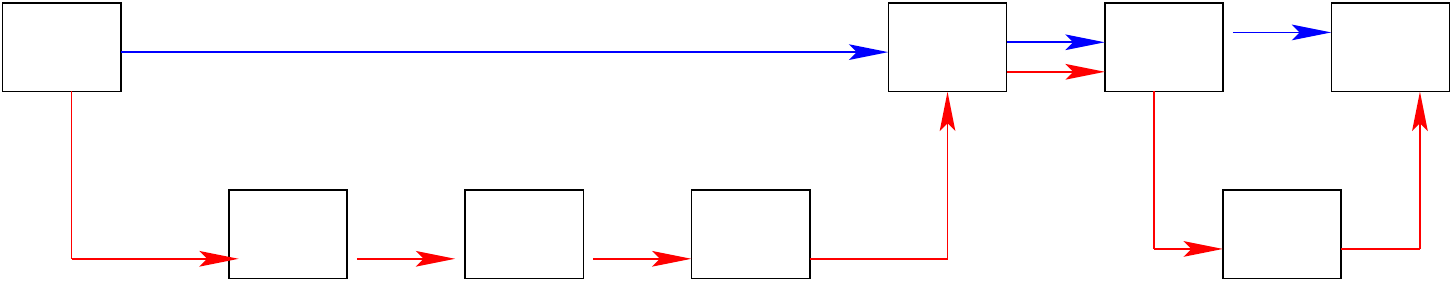}
	\caption{Execution under \lsl{}. $k_{10}$ is added in \lsl{} if not present.}
	\label{fig:nostm1}
\end{figure}


\subsection{Thread local log DS }
In proposed \otm{}, we use $thread$ $local$ $DS$ which is private to each thread for logging the local execution and $shared$ $memory$ $DS$ which is concurrently accessed by multiple transactions to communicate the meta information logged for validation of the \mth{s}.\cmnt{ On successful validation effect of each transaction will be visible into underlying shared DS.}

Each transaction $T_i$ maintains \emph{local log} which is a tuple of type $txlog\langle t\_id, tx\_status, le\rangle$, which consists of $t\_id$ and $tx\_status$ of the transaction. Transactions can have live, commit or abort as their status signifying that transaction is executing, has successfully committed or has aborted due to some \mth{} failing the validation respectively. 

\ignore{
\begin{tcolorbox}
\begin{lstlisting}{language = C++}
class txlog{
private :
int t_id;			STATUS tx_status;
/*a log entry is uniquely identified using key and obj_id
in le(log_entry)*/
vector <key, le> ll_list;  
public :
txlog();	~txlog();	findInLL();     setStatus();
getLlList();	  sort();       tryAbort();         };
\end{lstlisting}
\end{tcolorbox}
}
The $txlog$ also maintains a list $le$ (log\_entries) of meta information of each method a transaction executes in its life time. The $le$ is again a tuple $\langle key, value, opn, op\_status, preds, currs \rangle$ which records 1) $key$ and $value$ a method operates on, 2) $opn$: name of the \mth{}, 3) $op\_status$: \mth{'s} status ($OK$, $FAIL$) and 4) $preds$, $currs$: its \emph{location} over the \lsl{}.

We say a \mth{} identifies its \emph{location} over the \lsl{} when it finds the predecessor and successor nodes over the $\bn$ and $\rn$ respectively. We represent predecessor as $preds\langle \textcolor{blue}{k_m}, \textcolor{red}{k_n} \rangle$ ($k_m$ is unmarked node reachable by $\bn$ and $k_n$ is marked node reachable by $\rn$) and successor as $currs\langle\textcolor{red}{k_p}, \textcolor{blue}{k_q}\rangle$ ($k_p$ is marked for deletion node reachable by $\rn$ and $k_q$ is unmarked node reachable by $\bn$) respectively. Here, $ \langle \textcolor{blue}{k_m}, \textcolor{blue}{k_q} \rangle $ are predecessor (preds[0]) and current (currs[1]) node for $\bn$ and $ \langle\textcolor{red}{k_n}, \textcolor{red}{k_p}\rangle $ are predecessor (preds[1]) and current (currs[0]) node for $\rn$. We use word location with $preds$ and $currs$ interchangeably in rest of the paper. The $le$ is operated by getter and setter \mth{s} for each of the member variables as shown in table \ref{tabel:1}. Addtionally, we use following macros while explaining the pseudocode of \otm{} in \secref{pscode}.
\begin{tcolorbox}
\begin{lstlisting}{language = C++}
/*types of method exported by the HT-OSTM*/
enum OPERATION_NAME = {INSERT, DELETE, LOOKUP}

/*a transaction can ABORT/COMMIT and a method can ABORT, OK,
FAIL */
enum STATUS = {ABORT = 0, OK, FAIL, COMMIT}

/*to know whether validation is requested from TRYC or 
rv-method*/
enum VALIDATION_TYPE = {RV, TRYC}

/*To recognize on which list method has to be performed*/
enum LIST_TYPE = {RL, BL, RL_BL} 
\end{lstlisting}
\end{tcolorbox}

\begin{table}[H]
	\begin{tabular}{ || m{8em} | m{11cm}|| } 
		\hline
		\textbf{Functions} & \textbf{Description} \\ 
		\hline 		\hline
		setOpn() & store method name into ll\_list of the $txlog$\\ 
		\hline
		setValue() & store value of the key into ll\_list of the $txlog$\\ 
		\hline
		setOpStatus() & store status of method into ll\_list of the $txlog$\\ 
		\hline
		setPreds\&Currs() & store location of $preds$ and $currs$ according to the node corresponding to the key into ll\_list of the $txlog$\\ 
		\hline
		getOpn() & give operation name from ll\_list of the $txlog$\\ 
		\hline
		getValue() & give value of the key from ll\_list of the $txlog$\\ 
		\hline
		getOpStatus() & give status of the method from ll\_list of the $txlog$\\ 
		\hline
		getKey\&Objid() & give key and obj\_id corresponding to the method from ll\_list of the $txlog$\\ 
		\hline
		getAptCurr() & give the red or blue curr node from the log corresponding to the key of the $txlog$\\ 
		\hline
		getPreds\&Currs() & give location of $preds$ and $currs$ according to the node corresponding to the key from ll\_list of the $txlog$ \\ 
		\hline		
	\end{tabular}
	\caption{utility \mths to manipulate $txlog$.}
	\label{tabel:1}
\end{table}

\ignore{
\subsubsection{Shared memory DS:}
\otm{} shared memory is the chained \tab{} where each node of the chain (\lsl{}) is a key-value pairs of the form $\langle k, v \rangle$. Most of the notations used here are derived from \cite{Vafeiadis:2006:PCH:1122971.1122992}. A node $n$ when created is initialized as follows: (1) $key$ and $val$ is the key and value of the method that creates the node (2) $rednext$ and $bluenext$ are set to $nil$ (3) $marked$ is set to $false$ (4) $lock$ is null (5) $max\_ts$ is initialized to 0.
\begin{tcolorbox}
\begin{lstlisting}{language = C++}
struct node{
int key, value;		bool marked; 	struct max_ts;
lock;			node* rl;	node* bl;    };
/*hash table where each bucket is a lazyrb-list chain*/
node* shared_ht [];
\end{lstlisting}
\end{tcolorbox}
We adapt timestamp validation\cite{WeiVoss:2002:Morg} to ensure schedules generated by proposed \otm{} are serial. Therefore we maintain \textit{max\_ts\_lookup($ht,k$), max\_ts\_insert($ht,k$) and max\_ts\_delete($ht,k$)} that represents timestamp of last committed transaction which executed \tlook($ht,k$), \tins($ht,k$) and \tdel($ht,k$) respectively. $max\_ts$, $node$ and $le$ form the part of the meta information for the \otm{}.
\begin{tcolorbox}
\begin{lstlisting}{language = C++}
/*stores the time stamp of last transaction that performed 
lookup, insert or delete respectively  */
struct max_ts { lookup; insert; delete; }; 
\end{lstlisting}
\end{tcolorbox}
}
\ignore{
\subsection{Pseudocode} \label{sec:pcodee}

Through out its life an \otm{} transaction may execute \emph{STM\_begin()}, \npins{}, \npluk{}, \npdel{} and \nptc{} \mth{s} which are also exported to the user. Each transaction has a 1) \rvp: where \upmt{} \& \rvmt{} locally identify and logs the location to be worked upon and other meta information which would be needed for successful validation. Within \rvp{} \rvmt{s} do lock free traversal and then validate while \npins{} merely log their execution to be validated and updated during transaction commit. 2) \cp: where it validates the \upmt{} executed during its lifetime and validates whether the transaction will commit and finally make changes in \tab{} atomically or it will abort and flush its log. \figref{nostm24} depicts the transaction life cycle.
\begin{figure}[H]
	\centering
	\captionsetup{justification=centering}
	\centerline{\scalebox{0.6}{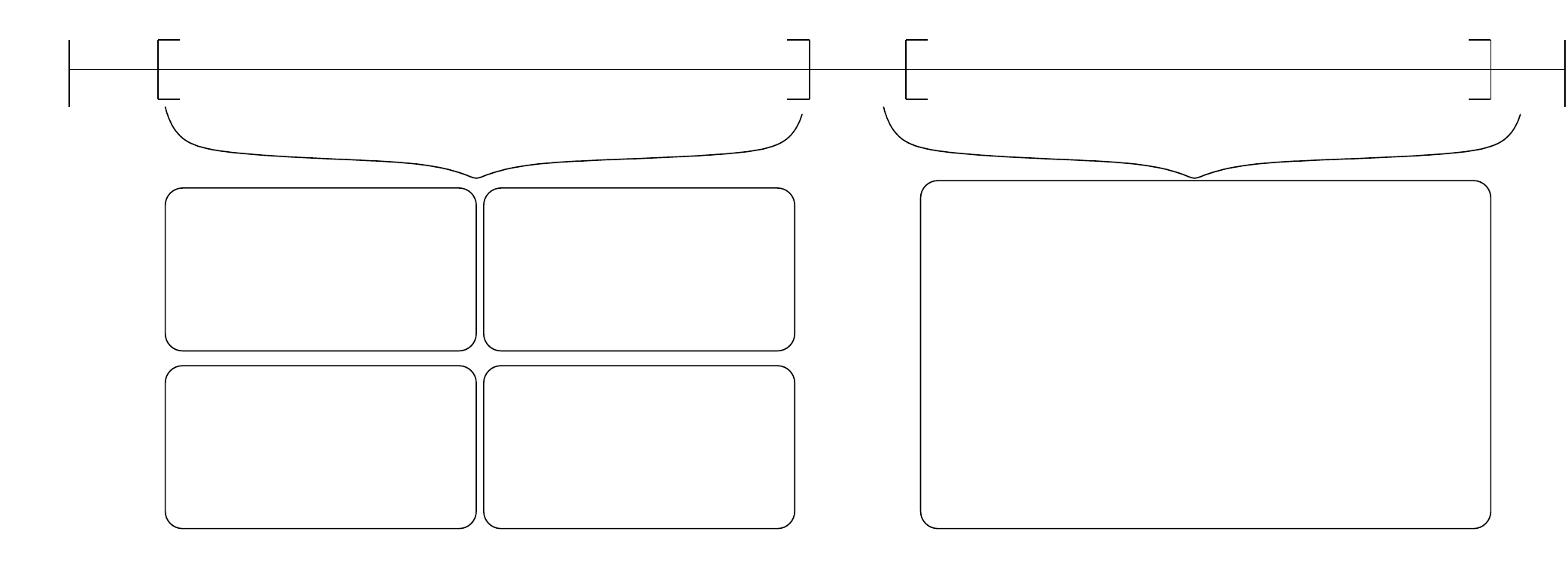}}
	\caption{Transaction lifecycle of \otm}
	\label{fig:nostm24}
\end{figure}
\textbf{Pseudocode convention:} In each algorithm $\downarrow$ represents the input parameter and $\uparrow$ shows the output parameter (or return value) of the corresponding methods (such in and out variables are italicized). Instructions in \emph{read()} and \emph{write()} with in each \mth{} denote that they touch the shared memory. Color of $preds$ \& $currs$ in algorithm depicts the red or blue node which are accessed by red or blue links respectively.

\textbf{\rvmt{} execution phase:}
Initially, 
\npbegin\emph{()} is the first function a transaction executes in its life cycle. It initiates the $txlog$ (local log) for the transaction (\Lineref{begin2}) and provides an unique id to the transaction (\Lineref{begin3}).

\begin{algorithm}[H]
	\caption{\tabspace[0.2cm]  STM\_begin($t\_id\uparrow$) : initiates local transaction log and return the transaction id.}     
	\label{algo:begin}
	
	\begin{algorithmic}[1]
		\Function{stm\_begin}{} 
		\label{lin:begin1}
		\cmntwa{ init the local log}
		\State txlog $\gets$ new txlog(); \label{lin:begin2}
		\cmntwa{atomic variable to assign transaction id i.e. TS initilized by OSTM as 0}
		\State t\_id$\gets$ $get\&inc(sh\_cntr \uparrow)$;//$\Phi_{lp}$ \label{lin:begin3}
		\State return $\langle t\_id\rangle$;\label{lin:begin4}
		\EndFunction 
	\end{algorithmic}
	
\end{algorithm}

Then transaction may encounter the \upmt{} or \rvmt{}. \npins{} method in \rvp{} simply checks if their is a previous \mth{} that executed on the same $key$. If their is already a previous \mth{} that has executed within the same transaction it simply updates the new $value$, $opn$ as insert and $op\_status$ to $OK$ (\Lineref{insert5}, \Lineref{insert6} and \Lineref{insert7} respectively). In case the \npins{} is the first \mth{} on $key$ it creates a new log entry for the $ll\_list$ of $txlog$ at \Lineref{insert3}. Finally the \npins{} gets to modify the underlying \tab{} using \nplslins{} at the \cp{}.


\begin{algorithm}[H]
	\caption{\tabspace[0.2cm] STM\_insert ($t\_id \downarrow, obj\_id \downarrow, key\downarrow, value\downarrow, op\_status\uparrow$) : updates log entry and return op\_status locally.
	}
	\label{algo:insert}
	\begin{algorithmic}[1]
		\makeatletter\setcounter{ALG@line}{7}\makeatother
		\Function{STM\_insert}{} \label{lin:insert1}
		\State STATUS op\_status $\gets$ OK;
		\cmntwa{get the txlog of the current transaction by t\_id}
		\State txlog $\gets$ getTxLog($t\_id \downarrow$);
		\If{$(!$\txlfind$)$} \label{lin:insert2}
		\cmntwa{no le present for this $\langle obj\_id, key\rangle$, create one}
		
		\State le $\gets$ new le$\langle obj\_id \downarrow, key\downarrow\rangle$\label{lin:insert3}; 
		\EndIf \label{lin:insert4}
		
		\cmntwa{le present for $\langle obj\_id, key\rangle$, merely update the log}
		\State \llsval{$value \downarrow$} \label{lin:insert5}; //$\Phi_{lp}$
		\State \llsopn{$INSERT \downarrow$}
		\label{lin:insert6};
		
		\State \llsopst{$OK \downarrow$} \label{lin:insert7};
		\cmntwa{return op\_status to the transaction that invoked insert}
		\State return $\langle op\_status\rangle$;
		\EndFunction \label{lin:insert8}
		
	\end{algorithmic}
\end{algorithm}
\cmnt{
	\begin{figure}[H]
		\centering
		\captionsetup{justification=centering}
		\includegraphics[width=6cm, height=6cm]{figs/insertfc.png}
		\caption{Flow of \npins{}}
		\label{fig:insertfc}
	\end{figure}
}

\nptc{} and \rvmt{} of \otm{} uses \nplsls{} to find the location at the \lsl{} (thus the name) in lock free manner. \Lineref{lslsearch3} to \Lineref{lslsearch8} and \Lineref{lslsearch11} to \Lineref{lslsearch15} of \algoref{lslsearch} find the location at \lsl{} for $\bn$ and $\rn$ respectively. This is motivated by the search in lazylist \cite[section 9.7]{Herlihy:ArtBook:2012}. The $preds$ and $currs$ thus identified are subjected to \npintv{} of \algoref{interferenceValidation} and \nptov{} of \algoref{tovalidation} after acquiring locks on the $preds$ and $currs$ (\Lineref{lslsearch17} of \algoref{lslsearch}). If the validation succeeds \nplsls{} returns the correct location to the operation which invoked it, otherwise \nplsls{} retries (if interference detected) or aborts (if time order violated) post releasing locks (\Lineref{lslsearch22}). 

\begin{algorithm}[H]
	\caption{rblSearch($t\_id \downarrow, obj\_id \downarrow, key \downarrow, val\_type \downarrow, \preds \uparrow, \currs \uparrow, op\_status \uparrow$) : finds location (\preds \& \currs) for given $ \langle obj\_id, key \rangle$ and returns them in locked state else returns ABORT. 
	}\label{algo:lslsearch}
		\begin{algorithmic}[1]
			\makeatletter\setcounter{ALG@line}{22}\makeatother
			\Function{rblSearch}{} \label{lin:lslsearch1}
			
			\State STATUS $op\_status$ $\gets$ RETRY;
			\While{($op\_status$ = \textup{RETRY})} \label{lin:lslsearch2}
			\cmntwa{get the head of the bucket in hash-table}
			\State head $\gets$ \glslhead \label{lin:lslsearch3};
			\cmntwa{init $\bp$ to head}
			\State $\bp$ $\gets$ $\texttt{read}$($head$) \label{lin:lslsearch4}; 
			\cmntwa{init $\bc$ to $\bp.\bn$}
			\State $\bc$ $\gets$ $\texttt{read}$($\bp$.\bn) \label{lin:lslsearch5};
		\algstore{testcont} 
	\end{algorithmic}
\end{algorithm}

\begin{algorithm}[H]
	\begin{algorithmic}[1]
		\algrestore{testcont} 

			\cmntwa{search node $ \langle obj\_id, key \rangle$ location in blue list}
			\While{$(\texttt{read}($$\bc$.\textup{key}$)$ $<$ $key)$} \label{lin:lslsearch6}
			\State $\bp$ $\gets$ $\bc$ \label{lin:lslsearch7};
			
			\State $\bc$ $\gets$ $\texttt{read}$($\bp$.\bn) \label{lin:lslsearch8};
			
			\EndWhile \label{lin:lslsearch9}
			\cmntwa{ init $\rp$ to $\bp$}
			\State $\rp$ $\gets$ $\bp$ \label{lin:lslsearch11};
			\cmntwa{init $\rc$ to $\bp.\rn$}
			\State $\rc$ $\gets$ $\rp$.\rn \label{lin:lslsearch12};
			\cmntwa{search node $ \langle obj\_id, key \rangle$ location in red list between \bp \& \bc}			\While{$(\texttt{read}($$\rc$.\textup{key}$)$ $<$ $key)$} \label{lin:lslsearch13}
			
			\State $\rp$ $\gets$ $\rc$ \label{lin:lslsearch14};
			\State $\rc$ $\gets$ $\texttt{read}$($\rp$.\rn) \label{lin:lslsearch15};
			
			\EndWhile \label{lin:lslsearch16}
			\cmntwa{acquire the locks on increasing order of keys}
			\State acquirePred\&CurrLocks($ \preds \downarrow$, $ \currs \downarrow$)\label{lin:lslsearch17};
			
			
			\cmntwa{validate the location recorded in \preds \& \currs. Also verify if the transaction has to be aborted.}
			\State validation($t\_id \downarrow$, $key$ $\downarrow$, $\preds$ $\downarrow$, $\currs$ $\downarrow$, $val\_type$ $\downarrow$, $op\_status \uparrow$)\label{lin:lslsearch21};	
			\cmntwa{if validation returns op\_status as RETRY or ABORT then release all the locks}
			\If{(($op\_status$ = \textup{RETRY}) $\lor$ ($op\_status$ = \textup{ABORT}))} \label{lin:lslsearch22}
			\cmntwa{release all the locks}
			\State{releasePred\&CurrLocks($\preds \downarrow$, $\currs \downarrow$)}
			
			\EndIf \label{lin:lslsearch27}
			
			\EndWhile \label{lin:lslsearch28}
			
			\State return $\langle \preds, \currs, op\_status \rangle$ \label{lin:lslsearch29};
			
			\EndFunction \label{lin:lslsearch30}
		\end{algorithmic}
\end{algorithm}

Interference validation helps detecting the execution where underlying data structure has been changed by second concurrent transaction while first was under execution without it realizing. This can be illustrated with \figref{nostm8}. Consider the history in \figref{nostm8}(iii) where two conflicting transactions $T_1$ and $T_2$ are trying to access key $k_5$, here $s_1$, $s_2$ and $s_3$ represent the state of the \lsl{} at that instant. Let at $s_1$ both the methods record the same $preds \langle k_1, k_3 \rangle $ and $currs \langle k_5, k_5 \rangle$ with the help of $\nplsls{}$ for key $k_5$ (refer \figref{nostm8}(i)). Now, let $d_1(k_5)$ acquire the lock on the $preds$ and $currs$ before the $l_2(k_5)$ and delete the node corresponding to the key $k_5$ from $\bn$ leading to state $s_2$ (in \figref{nostm8}(iii)) and commit. \figref{nostm8}(ii) shows the state $s_2$ where key $k_5$ is the part of $\rn$. Now, \npintv{} (in \algoref{interferenceValidation}) will identify that location of $l_2(k_5)$ is no more valid due to ($\bp.\bn$ $\neq$ $\bc$) at \Lineref{iv2} of \algoref{interferenceValidation}. Thus, $\nplsls{}$ will retry to find the updated location for $l_2(k_5)$ at state $s_3$ (in \figref{nostm8}(iii)) and eventually $T_2$ will commit.
\begin{figure}[H]
	\centering
	\captionsetup{justification=centering}
	\centerline{\scalebox{0.5}{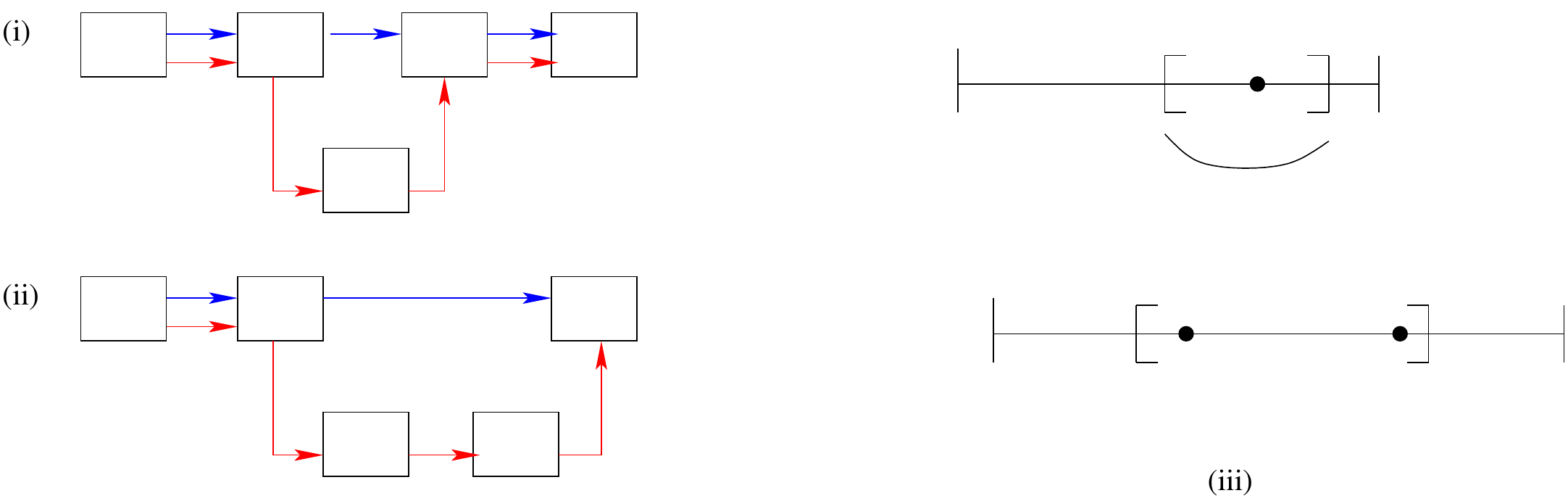}}
	\caption{Interference Validation for conflicting concurrent methods on key $k_5$}
	\label{fig:nostm8}
\end{figure}

\begin{algorithm}[H]
	\caption{\tabspace[0.2cm] STM\_lookup($t\_id \downarrow, obj\_id \downarrow, key \downarrow, value \uparrow, op\_status \uparrow$ ) 
	}
	\label{algo:lookup}
		
		\begin{algorithmic}[1]
			\makeatletter\setcounter{ALG@line}{22}\makeatother
			\Function{STM\_lookup}{} \label{lin:lookup1}
			\State STATUS $op\_status$ $\gets$ RETRY \label{lin:lookup2};
			\cmntwa{get the txlog of the current transaction by t\_id}
			\State txlog $\gets$ getTxLog($t\_id \downarrow$);
			
			\cmntwa{If already in log update the le with the current operation}
			\If{$($\txlfind$)$} \label{lin:lookup3}
			\State opn $\gets$ \llgopn{} \label{lin:lookup4}; 
			
			\cmntwa{if previous operation is insert/lookup then current method would have value/op\_status same as previous log entry}
			\If{$(($\textup{INSERT} $=$ \textup{opn} $)||($ \textup{LOOKUP} $=$ \textup{opn}$))$} \label{lin:lookup5}
			
			\State $value$ $\gets$ \llgval{} \label{lin:lookup6};
			\State $op\_status$ $\gets$  $le.getOpStatus$($obj\_id \downarrow$, $key \downarrow$) \label{lin:lookup7};
			\cmntwa{if previous operation is delete then current method would have value as NULL and op\_status as FAIL}
			\ElsIf{$($\textup{DELETE} $=$ \textup{opn}$)$} \label{lin:lookup8}
			\State $value$ $\gets$ NULL \label{lin:lookup9}; 
			\State $op\_status$ $\gets$ FAIL \label{lin:lookup10}; 
			\EndIf \label{lin:lookup11}
			\Else \label{lin:lookup12}
			\cmntwa{common function for \rvmt{}, if node corresponding to the key is not the part of underlying DS}
			\State \cld{};
			\cmnt{            
				\cmntwa{le$\langle obj\_id,key \rangle$ is not in log, search correct location for the operation over lsl and lock the corresponding \preds and \currs.}
				\State \lsls{$RV \downarrow$} \label{lin:lookup13}; 
				\If{$(op\_status$ $=$ \textup{ABORT}$)$} \label{lin:lookup14}
				\cmntwa{release local memory in case lslSearch returns abort}
				\State \handlea{} \label{lin:lookup15};
				\State return $\langle op\_status\rangle$;
				\Else \label{lin:lookup16}
				\cmntwa{if node$ \langle obj\_id, key \rangle$ is present update its lookup timestamp}	\If{$(\texttt{read}($$\bc$.\textup{key}$)$ $=$ key$)$} \label{lin:lookup17}
				\cmntwa{$node<obj\_id, key>$ is part of blue list}
				\State $op\_status$ $\gets$ OK \label{lin:lookup18};
				\State $\texttt{write}$($\bc$.max\_ts.lookup, TS($t\_id$)) \label{lin:lookup19};
				\State $value$ $\gets$ $\bc.value$ \label{lin:lookup20};
				\ElsIf{$(\texttt{read}($$\rc$.\textup{key}$)$ $=$ key$)$} \label{lin:lookup21}
				\cmntwa{$node<obj\_id, key>$ is part of red list}
				\State $op\_status$ $\gets$ FAIL \label{lin:lookup22};
				\State $\texttt{write}$($\rc$.max\_ts.lookup, TS($t\_id$)) \label{lin:lookup23};
				
				\State $value$ $\gets$ NULL \label{lin:lookup24};
				\Else \label{lin:lookup25}
				\cmntwa{if node$<obj\_id,key>$ is neither in blue or red list add the node in red list and update timestamp}
				\State \lslins{\textcolor{red}{$RL$} $\downarrow$} \label{lin:lookup26}; 
				\State $op\_status$ $\gets$ FAIL \label{lin:lookup27};
				\State $\texttt{write}$(\textcolor{red}{$sh\_node$}.max\_ts.lookup, TS($t\_id$)) \label{lin:lookup28};
				\State $value$ $\gets$ NULL \label{lin:lookup29};
				\EndIf \label{lin:lookup30}
				
				\cmntwa{release all the locks}
				\State{releasePred\&CurrLocks($\preds \downarrow$, $\currs \downarrow$);}

				\cmntwa{new log entry created to help upcoming method on the same key of the same tx}
				\State le $\gets$ new le$\langle obj\_id \downarrow, key\downarrow\rangle$ \label{lin:lookup31}; 
				\State \llsval{$NULL \downarrow$};
				\State \llspc{} \label{lin:lookup32};
				
				\State \llsopn{$LOOKUP \downarrow$} \label{lin:lookup33};
				\EndIf \label{lin:lookup38}
			}					
			\EndIf \label{lin:lookup39}
			\cmntwa{update the local log}
			\State \llsopn{$LOOKUP \downarrow$} \label{lin:lookup33};
			\State \llsopst{$op\_status \downarrow$} \label{lin:lookup40};
			\State return $\langle value, op\_status\rangle$\label{lin:lookup41}; 
			
			\EndFunction \label{lin:lookup42}
		\end{algorithmic}
		
	
\end{algorithm}

Consider $STM\_lookup_i(ht, k)$. If this is the subsequent operation by a transaction $T_i$ for a particular key $k$ on hash-table $ht$ i.e. an operation on $k$ has already been scheduled with in the same transaction $T_i$, then this \npluk{} return the value from the ll\_list and does not access shared memory (\Lineref{lookup3} to \Lineref{lookup10} in \algoref{lookup}). If the last operation was an \npins{} (or \npluk{}) on same key then the subsequent \npluk{} of the same transaction returns the previous value (\Lineref{lookup6} in \algoref{lookup}) inserted (or observed) without accessing shared memory, and if the last operation was an \npdel{} then \npluk{} returns the value NULL (\Lineref{lookup9} in \algoref{lookup}). Thus in this process subsequent \mth{s} also have same conflicts as the first \mth{} on same key within the same transaction (\emph{conflict inheritance}).

If \npluk{} is the first operation on a particular key then it has to do a wait free traversal (\Lineref{delete21} in \algoref{commonLu&Del}) with the help of $\nplsls$ (\algoref{lslsearch}) to identify the target node ($preds$ and $currs$) to be logged in ll\_list for subsequent \mth{s} in \rvp{} (discussed above for the case where \npluk{} is the subsequent \mth{}). If the node is present as blue (or red) node then it updates the operation status as OK (or FAIL) and returns the value respectively (\Lineref{delete25} to \Lineref{delete32} in \algoref{commonLu&Del}). If node corresponding to the key is not found  then it inserts that node (\Lineref{delete33} to \Lineref{delete37} in \algoref{commonLu&Del}) corresponding to the key into  $\rn$ of \lsl{}. The inserted node can be accessed only via red links. Hence, it will not visible to any subsequent \npluk{}. The node is inserted to take care of situations as  illustrated in \figref{nostm25} \& \figref{nostm26} . Finally, it updates the meta information in ll\_list and releases the locks acquired inside $\nplsls$ (\Lineref{delete433}).   

We prefer \npluk{} to be validated instantly and is never validated again in \nptc{} as the design choice to aid performance. Let's consider \otm{} history in \figref{nostm19}(i), if we would have validated $l(ht, k_1, v_0)$ again during $tryC$, $T_1$ would abort due to time order violation\cite{WeiVoss:2002:Morg}, but we can see that this history is acceptable where $T_1$ can be serialized before $T_2$ (\figref{nostm19}(ii)). Thus, \otm{} prevents such unnecessary aborts. Another advantage for this design choice is that $T_1$ doesn't have to wait for $tryC$ to know that the transaction is bound to abort as can be seen in \figref{nostm19}(iii). Here $l(ht, k_1, Abort)$ instantly aborts as soon as it realizes that time order is violated and schedule can no more be ensured to be correct saving significant computations of $T_1$.  This gain becomes significant if the application is lookup intensive where it would be inefficient to wait till \nptc{} to validate the \npluk{} only to know that transaction has to abort.

\begin{figure}[H]
	\centering
	\captionsetup{justification=centering}
	{\scalebox{0.5}{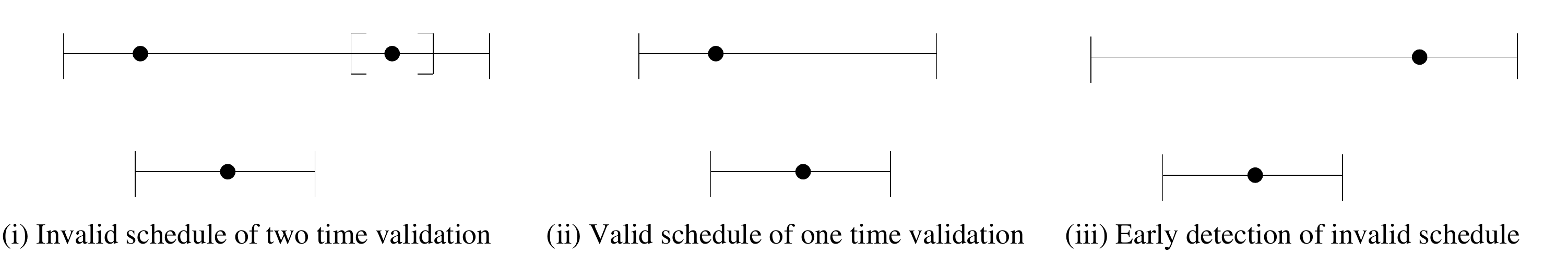}}
	\caption{Advantages of lookup validated once}
	\label{fig:nostm19}
\end{figure}

\setlength{\textfloatsep}{0pt}



\npdel{} (\algoref{delete}) in \rvp{} executes as similar to \rvmt{} and in \cp{} executes as \upmt{}. In \rvp{}, the \npdel{} first checks if their is already a previous \mth{} on same $key$ using the local log. In case their is already a \mth{} that executed on same $key$, \npdel{} does not need to touch shared memory and sees the effect of the previous \mth{} and returns accordingly (\Lineref{delete3} to \Lineref{delete18}). For example if previous executed \mth{} is an insert then the current \npdel{} \mth will return $OK$ (\Lineref{delete5} to \Lineref{delete9}). If the previous executed \mth{} is an \npdel{} then the current \npdel{} should return $FAIL$ (\Lineref{delete10} to \Lineref{delete13}). In case previous \mth{} was \npluk{} then current \npdel{} returns the status same as that of the previous \npluk{} \mth{} also overwriting the log for the $value$ and $opn$.

\begin{algorithm}[H]
	\caption{\tabspace[0.2cm] STM\_delete($t\_id \downarrow, obj\_id \downarrow, key \downarrow, value \uparrow, op\_status \uparrow$ ) 
	}
	\label{algo:delete}
	\begin{algorithmic}[1]
		\makeatletter\setcounter{ALG@line}{74}\makeatother
		\Function{STM\_delete}{} \label{lin:delete1}
		
		\State STATUS $op\_status$ $\gets$ RETRY;
		\cmntwa{get the txlog of the current transaction by t\_id}
		\State txlog $\gets$ getTxLog($t\_id \downarrow$);
		\cmntwa{If le$\langle obj\_id, key \rangle$ already in log, update the le with the current operation}
		\If{$($\txlfind$)$} \label{lin:delete3}
		\State opn $\gets$ \llgopn{} \label{lin:delete4}; 
		\cmntwa{if previous local method is insert and current operation is delete then overall effect should be of delete, update log accordingly}
		\If{$($\textup{INSERT} $=$ \textup{opn}$)$} \label{lin:delete5}
		\State $value$ $\gets$ \llgval{} \label{lin:delete6};
		
		\State \llsval{$NULL \downarrow$} \label{lin:delete7};
		
		\State \llsopn{$DELETE \downarrow$} \label{lin:delete8};
		
		\State $op\_status$ $\gets$ OK \label{lin:delete9};
		\cmntwa{if previous local method is delete and current operation is delete then overall effect should be of delete, update log accordingly}
		\ElsIf{$($\textup{DELETE} $=$ \textup{opn}$)$} \label{lin:delete10}
		\State \llsval{$NULL \downarrow$} \label{lin:delete11};
		\State $value$ $\gets$ NULL \label{lin:delete12}; 
		\State $op\_status$ $\gets$ FAIL \label{lin:delete13}; 
		\Else \label{lin:delete14}
		\cmntwa{if previous local method is lookup and current operation is delete then overall effect should be of delete, update log accordingly}
		\State $value$ $\gets$ \llgval{} \label{lin:delete15}; 
		
		\State \llsval{$NULL \downarrow$} \label{lin:delete16};

		\State \llsopn{$DELETE \downarrow$} \label{lin:delete17};
		\State $op\_status$ $\gets$  $le.getOpStatus$($obj\_id \downarrow$, $key \downarrow$) \label{lin:delete18};
		
		\EndIf \label{lin:delete19}
		\Else \label{lin:delete20}
		\cmntwa{common function for \rvmt{}, if node corresponding to the key is not the part of underlying DS}
		\algstore{testcont} 
	\end{algorithmic}
\end{algorithm}

\begin{algorithm}[H]
	\begin{algorithmic}[1]
		\algrestore{testcont} 
		
		\State \cld{};
		\cmnt{		
			\cmntwa{le$ \langle obj\_id,key \rangle$ is not in log, search correct location for the operation over lsl and lock the corresponding \preds and \currs.}
			\State \lsls{$RV \downarrow$} \label{lin:delete21};
			\If{$(op\_status$ $=$ \textup{ABORT}$)$} \label{lin:delete22}
			\cmntwa{release local memory in case lslSearch returns abort}
			\State \handlea{} \label{lin:delete23}; 
			\State return $\langle op\_status\rangle$;

			\Else \label{lin:delete24}
			\cmntwa{if node$ \langle obj\_id, key \rangle$ is present update its lookup timestamp as delete in rv phase behaves as lookup}
			\If{$(\texttt{read}($$\bc$.\textup{key}$)$ $=$ key$)$} \label{lin:delete25}
			\cmntwa{$node<obj\_id, key>$ is part of blue list}
			\State $op\_status$ $\gets$ OK \label{lin:delete26};
			\State $\texttt{write}$($\bc$.max\_ts.lookup, TS($t\_id$)) \label{lin:delete27};
			\State $value$ $\gets$ $\bc.value$ \label{lin:delete28};
			
			\ElsIf{$(\texttt{read}($$\rc$.\textup{key}$)$ $=$ key$)$} \label{lin:delete29}
			\cmntwa{$node<obj\_id, key>$ is part of red list}		
			
			\State $op\_status$ $\gets$ FAIL \label{lin:delete30};
			
			\State $\texttt{write}$($\rc$.max\_ts.lookup, TS($t\_id$)) \label{lin:delete31};
			\State $value$ $\gets$ NULL \label{lin:delete32};
			\Else \label{lin:delete33}
			\cmntwa{if node$<obj\_id,key>$ is neither in blue or red list add the node in red list and update timestamp}
			\State \lslins{\textcolor{red}{$RL$} $\downarrow$} \label{lin:delete34};
			\State $op\_status$ $\gets$ FAIL \label{lin:delete35};
			
			\State $\texttt{write}$(\textcolor{red}{$sh\_node$}.max\_ts.lookup, TS($t\_id$)) \label{lin:delete36};
			
			\State $value$ $\gets$ NULL \label{lin:delete37};
			\EndIf \label{lin:delete38}
			\cmntwa{release all the locks}
			\State{releasePred\&CurrLocks($\preds \downarrow$, $\currs \downarrow$)}\label{lin:delete433};
			\cmntwa{create new log entry in log}
			
			\State le $\gets$ new le$\langle obj\_id \downarrow, key\downarrow\rangle$\label{lin:delete39};
			\State \llsval{$NULL \downarrow$} \label{lin:delete40};
			
			\State \llspc{} \label{lin:delete41};
			
			\State \llsopn{$DELETE \downarrow$} \label{lin:delete42};
			
			\EndIf\label{lin:delete47}
		}
		\EndIf \label{lin:delete48}
		\cmntwa{update the local log}
		\State \llsopn{$DELETE \downarrow$} \label{lin:delete42};
		
		\State \llsopst{$op\_status \downarrow$} \label{lin:delete49};
		\State return $\langle value, op\_status\rangle$\label{lin:delete50};
		\EndFunction \label{lin:delete51}
	\end{algorithmic}
\end{algorithm}

\textbf{\upmt{} execution phase:} Finally a transaction after executing the designated operations reaches the \emph{\upmt{} execution} phase executed by the \nptc{} method. It starts with modifying the log to $ordered\_ll\_list$ which contains the log entries in sorted order of the keys (so that locks can be acquired in an order, refer \Lineref{tryc4} of \algoref{trycommit}) and contains only the \upmt{} (because we do not validate the lookup again for the reasons explained above for \figref{nostm19}). From \Lineref{tryc5} to \Lineref{tryc12} (in \algoref{trycommit}) we re-validate the modified log operation to ensure that the location for the operations has not changed since the point they were logged during \rvp{}. If the location for an operation has changed this block ensures that they are updated. Now, \nptc{} enters the phase where it updates the shared memory using local data stored from \Lineref{tryc14} to \Lineref{tryc37} in \algoref{trycommit}. \figref{nostm10} \& \figref{nostm9} explain the execution of insert and delete in update phase of \nptc{} using \nplslins{} and \nplsldel{} respectively. \figref{nostm10}(i) represents the case when $k_5$ is neither present in $\bn$ and nor in $\rn$ (\Lineref{tryc28} to \Lineref{tryc31} in \algoref{trycommit}). It adds $k_5$ to \lsl{} at location $preds\langle k_3, k_4 \rangle $ and $currs\langle k_8, k_8 \rangle$. \figref{nostm10}(i)(a) is \lsl{} before addition of $k_5$ and \figref{nostm10}(i)(b) is \lsl{} state post addition. Similarly, \figref{nostm10}(ii) represents the case when $k_5$ is present in $\rn$ (\Lineref{tryc24} to \Lineref{tryc27} in \algoref{trycommit}). It adds $k_5$ to \lsl{} at location $pred \langle k_3, k_4 \rangle $ and $curr \langle k_5, k_8 \rangle$. \figref{nostm10}(i)(c) is \lsl{} before addition of $k_5$ into \bn{} and \figref{nostm10}(i)(d) is \lsl{} state post addition. In case of $d(k_5)$ from \lsl{} when $k_5$ is present in $\bn$ (\Lineref{tryc33} to \Lineref{tryc37} in \algoref{trycommit}) \figref{nostm9}(i) represent the \lsl{} state before $k_5$ is deleted at location $preds \langle k_1, k_3 \rangle $ and $currs \langle k_5, k_5 \rangle $ and \figref{nostm9}(ii) represents the \lsl{} state after deletion.       
\begin{figure}[H]
	\centering
	\captionsetup{justification=centering}
	\centerline{\scalebox{0.55}{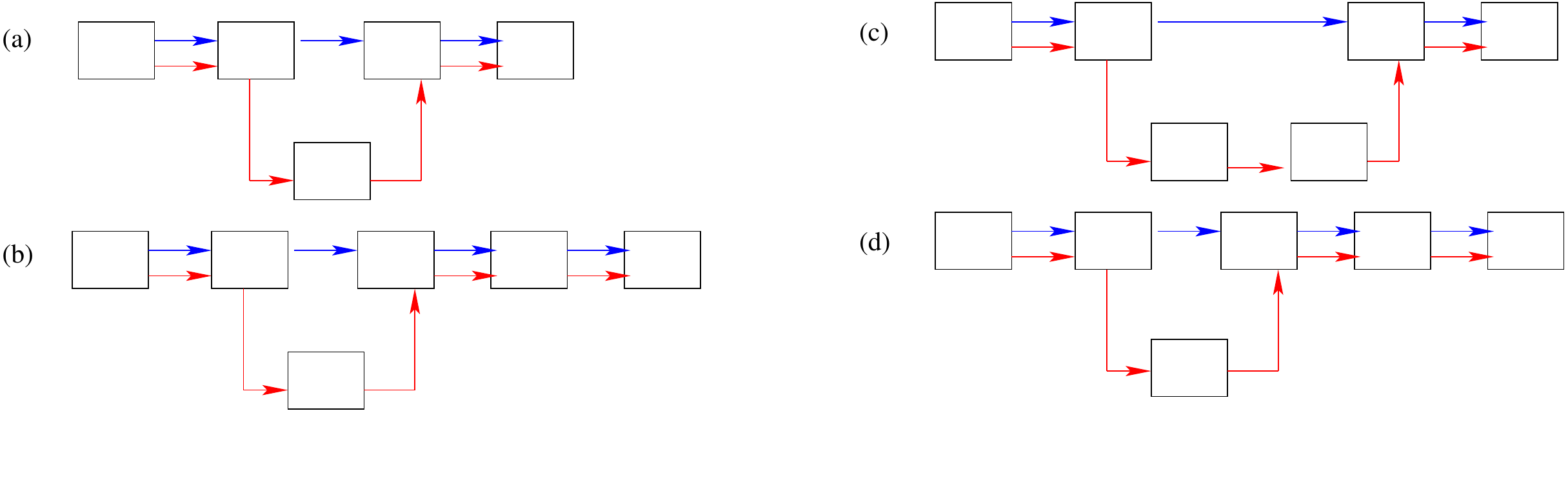}}
	\caption{$i(k_5)$ using \nplslins{} in $\nptc{}$}
	\label{fig:nostm10}
\end{figure}
\cmnt{
	Optimistically, \nptc{} of \algoref{trycommit}, find preds and currs for each $\upmt{}$ with the help of \nplsls{} method of \algoref{lslsearch}. Then \nplsls{} validates preds and currs for each $\upmt{}$. 
	After successful validation of \nplsls{}, \nptc{} performs \nppoval{} of \algoref{povalidation}, to get the program order validation. If all the validations are successful then shared memory is modified aptly before locks will release.  Otherwise \nptc{} returns abort. We avoid the overhead of rollbacks in case of transaction aborts as transaction are not allowed to modify the shared memory. After successful validation of \nptc{}, the actual effect of \npins{} and \npdel{} takes place atomically.
}
\cmnt{
	Consider \figref{nostm10} and \figref{nostm9} to illustrate more about \nptc{}. In \figref{nostm10}, we are taking two cases. First, last $\upmt{}$ of $\nptc{}$ wants to insert key $k_5$ in \figref{nostm10}(i), where $k_5$ is not the part of $\bn$ as well as $\rn$. Then it will create the node corresponding to the key $k_5$ and insert it into $\bn$ as well as $\rn$ with the help of $\nplslins{}$ of \algoref{lslins} (refer \figref{nostm10}(ii)). Second, last $\upmt{}$ of $\nptc{}$ wants to insert key $k_5$ in \figref{nostm10}(iii), where key $k_5$ is part of $\rn$ but not $\bn$. So after finding the location with the help of $\nplsls{}$ of \algoref{lslsearch}, $\nptc{}$ will add that node into $\bn$ as well by $\nplslins{}$ of \algoref{lslins} (refer \figref{nostm10}(iv)).  
}
\begin{figure}[H]
	\centering
	\captionsetup{justification=centering}
	\centerline{\scalebox{0.55}{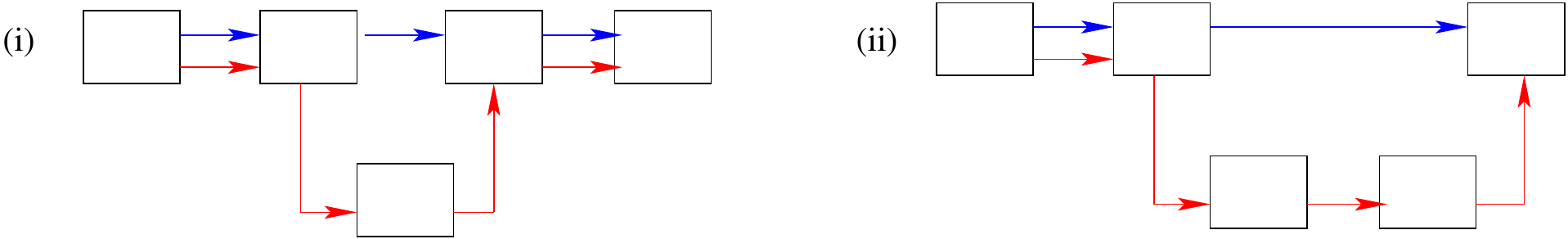}}
	\caption{$d(k_5)$ using \nplsldel{} in $\nptc{}$}
	\label{fig:nostm9}
\end{figure}
\cmnt{
	Last $\upmt{}$ of $\nptc{}$ wants to delete key $k_5$ in \figref{nostm9}(i), where $k_5$ is the part of $\bn$ as well as $\rn$. After finding the location with the help of \nplsls{} of \algoref{lslsearch}, $\nptc{}$ deletes the node corresponding to the key $k_5$ from the $\bn$ by $\nplsldel{}$ of \algoref{lsldelete} (refer \figref{nostm9}(ii)).
}
\cmnt{
	We correct this through $\nppoval{}$ which is $\nppoval{}$ is invoked before every $\upmt{}$ over the \lsl{} in update phase of \nptc{}(\Lineref{tryc14} to \Lineref{tryc40} of \algoref{trycommit}). \figref{nostm16} represents the functionality of $\nppoval{}$ of \algoref{povalidation}. Here, If \nppoval{} fails for any \upmt{} then as a corrective measure the $preds$ and $currs$ of the \upmt{} under execution will be updated using the previous \upmt{}'s $preds$ and $currs$ with the help of its $le$}
\cmnt{In \figref{nostm16}(i), same transaction is trying to insert key $k_5$ and $k_7$ in between the node $k_3$ and $k_8$. Before inserting the key $k_5$ both the operations recorded the same preds and currs (deferred write approach). But after successful insertion of key $k_5$ (refer \figref{nostm16}(ii)), preds for key $k_7$ has been changed. To solve this issue $\nppoval{}$ of \algoref{povalidation}, will check the preds and currs after every $\upmt{}$ of $\nptc{}$. If any changes occur then it updates the preds and currs with the help of previous $\upmt{}$. The same technique is used in the case of delete methods (\figref{nostm16}(iv)). After the successful insertion on key $k_5$ in \figref{nostm16}(v), it will search the new preds and currs with the help of $\nppoval{}$. New inserted node should be lock until the unlock of other preds and currs of same transaction. For example, all nodes ($k_5$ and $k_7$ in \figref{nostm16}(iii)) are locked so other transaction can't acquire the lock simultaneously. Please refer \secref{algorithmm} in appendix for more detail for pseudocode. 
}
\begin{figure}[H]
	\centering
	\captionsetup{justification=centering}
	\scalebox{.5}{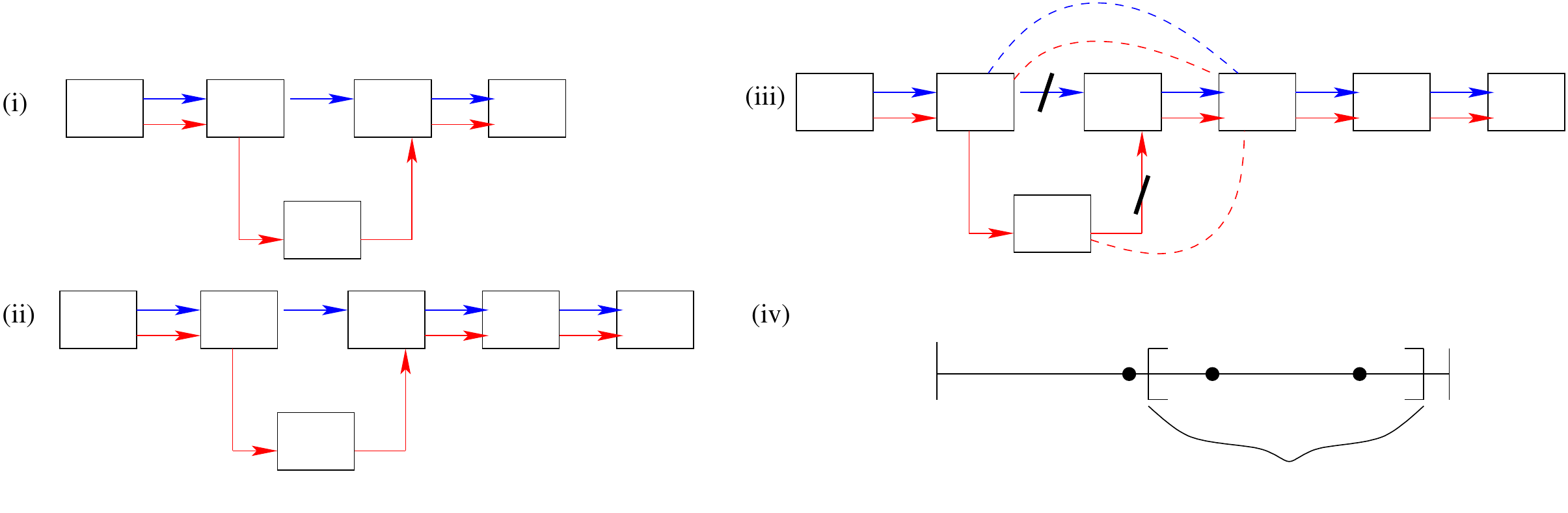}
	\captionof{figure}{Problem in execution without lostUpdateValidation() ($i_1(k_5)$ and $i_1(k_7)$). (i) \lsl{} at state s. (ii) \lsl{} at state $s_1$. (iii) \lsl{} at state $s_2$ (lost update problem).}
	\label{fig:nostm16}
\end{figure}
\noindent
In \cp{} two consecutive updates within same transaction having overlapping $preds$ and $currs$ may overwrite the previous \mth{} such that only effect of the later \mth{} is visible (\emph{lost update}). This happens because the previous \mth{} while updating, changes the \lsl{} causing the $preds$ \& $currs$ of the next \mth{} working on the consecutive key to become obsolete. \figref{nostm16} explains this lucidly. Suppose, $T_1$ is in update phase of \nptc{} at state $s$ where $ins_1(k_5)$ and $ins_1(k_7)$ are waiting to take effect over the \lsl{}. The \lsl{} at $s$ is as in \figref{nostm16}(i) also $ins_1(k_5)$ and $ins_1(k_7)$ have $preds \langle k_3, k_4 \rangle $ and $currs \langle k_8, k_8 \rangle $ as their location. Now, Lets say $ins_1(k_5)$ adds $k_5$ between $k_3$ and $k_8$ and changes \lsl{} (as in \figref{nostm16}(ii)) at state $s_1$ in \figref{nostm16}(iv). But, at $s_1$ $\bn$ $preds$ and $currs$ of $ins_1(k_7)$ are still $k_3$ and $k_8$ thus it wrongly adds $k_7$ between $k_3$ and $k_8$ overwriting $ins_1(k_5)$ as shown in \figref{nostm16}(iii) with dotted links. We correct this through $\nppoval{}$ which updates current \upmt{}'s $preds$ and $currs$ with the help of its $le$. We discuss lost update validation in detail in at \algoref{povalidation}.

In case the current \npdel{} is not the first \mth{} on $key$ then it touches the shared memory to identify the correct location over the \tab{} from \Lineref{delete20} to \Lineref{dete49}. \nplsls{} gives the correct location for the current \npdel{} to take effect over the \tab{} in form of $preds$ and $currs$ (\Lineref{delete21}) along with the validation status which reveals weather the \npdel{} will succeed or abort. If the $op\_status$ is Abort, the \mth{} simply aborts the transaction. Otherwise, \npdel{} updates the local log and the time stamps of the corresponding nodes in the \lsl{} of the \tab{} from line \Lineref{delete24} to \Lineref{delete49}. From \Lineref{delete25} to \Lineref{delete28}, \npdel{} observes that the node to be deleted is reachable from $\bn$ i.e. it is \bc{} thus it updates it's time-stamp field and returns $op\_status$ to $OK$ with the value of \bc{} (the update corresponding to this case takes place in \nptc{} as represented in \figref{1nostm9}). From \Lineref{delete29} to \Lineref{delete32}, \npdel{} observes that the node to be deleted is reachable by $\rn$ i.e. it is \rc{} thus it updates its time-stamp field and sets $op\_status$ to $FAIL$ (as the node is dead node or marked for deletion) and value returned is $NULL$. Otherwise, in \Lineref{delete33} to \Lineref{delete37} the node is not at all present in \lsl{}. Thus first \npdel{} adds a node in \rn{} and updates its time-stamp and returns the $value$ as $NULL$ and sets the $op\_status$ as $FAIL$ (\figref{0nostm} and \figref{1nostm1} represents the case). \Lineref{delete40}, \Lineref{delete41} and \Lineref{delete42} sets the $value$, location and $opn$ in local log respectively. At \Lineref{delete433} the locks acquired(in invoked \nplsls{}) to update shared memory time-stamps are released in order.
\begin{figure}[H]
	\centering
	\begin{minipage}[b]{0.49\textwidth}
		\scalebox{.47}{\input{figs/nostm.pdf_t}}
		\centering
		\captionsetup{justification=centering}
		\captionof{figure}{$k_{10}$ is not present in $\bn$ as well as $\rn$}
		\label{fig:0nostm}
	\end{minipage}
	\hfill
	\begin{minipage}[b]{0.49\textwidth}
		\centering
		\captionsetup{justification=centering}
		\scalebox{.47}{\input{figs/nostm1.pdf_t}}
		\captionof{figure}{Adding $k_{10}$ into $\rn$}
		\label{fig:1nostm1}
	\end{minipage}   
\end{figure}

\begin{algorithm}[H]
	\caption{\tabspace[0.2cm] commonLu\&Del($t\_id \downarrow, obj\_id \downarrow, key \downarrow, value \uparrow, op\_status \uparrow$ ) 
	}
	\label{algo:commonLu&Del}
		\begin{algorithmic}[1]
			\makeatletter\setcounter{ALG@line}{74}\makeatother
			\Function{commonLu\&Del}{} 
			
			
			\cmntwa{le$ \langle obj\_id,key \rangle$ is not in log, search correct location for the operation over lsl and lock the corresponding \preds and \currs.}
			\State \lsls{$RV \downarrow$} \label{lin:delete21};
			\If{$(op\_status$ $=$ \textup{ABORT}$)$} \label{lin:delete22}
			\cmntwa{release local memory in case lslSearch returns abort}
			\State \handlea{} \label{lin:delete23}; 
			\State return $\langle op\_status\rangle$;

			\Else \label{lin:delete24}
			\cmntwa{if node$ \langle obj\_id, key \rangle$ is present update its lookup timestamp as delete in rv phase behaves as lookup}
			\If{$(\texttt{read}($$\bc$.\textup{key}$)$ $=$ key$)$} \label{lin:delete25}
			\cmntwa{$node \langle obj\_id, key \rangle$ is part of blue list}
			\State $op\_status$ $\gets$ OK \label{lin:delete26};
			\State $\texttt{write}$($\bc$.max\_ts.lookup, TS($t\_id$)) \label{lin:delete27};
			\State $value$ $\gets$ $\bc.value$ \label{lin:delete28};
			
			\ElsIf{$(\texttt{read}($$\rc$.\textup{key}$)$ $=$ key$)$} \label{lin:delete29}
			\cmntwa{$node \langle obj\_id, key \rangle$ is part of red list}		
			
			\State $op\_status$ $\gets$ FAIL \label{lin:delete30};
			
			\State $\texttt{write}$($\rc$.max\_ts.lookup, TS($t\_id$)) \label{lin:delete31};
			\State $value$ $\gets$ NULL \label{lin:delete32};
			\Else \label{lin:delete33}
			\cmntwa{if node$ \langle obj\_id,key \rangle$ is neither in blue or red list add the node in red list and update timestamp}
			\State \lslins{\textcolor{red}{$RL$} $\downarrow$} \label{lin:delete34};
			\State $op\_status$ $\gets$ FAIL \label{lin:delete35};
			\algstore{testcont} 
		\end{algorithmic}
	\end{algorithm}
	
	\begin{algorithm}[H]
		\begin{algorithmic}[1]
			\algrestore{testcont} 
			
			\State $\texttt{write}$(\textcolor{red}{$sh\_node$}.max\_ts.lookup, TS($t\_id$)) \label{lin:delete36};
			
			\State $value$ $\gets$ NULL \label{lin:delete37};
			\EndIf \label{lin:delete38}
			\cmntwa{release all the locks}
			\State{releasePred\&CurrLocks($\preds \downarrow$, $\currs \downarrow$)}\label{lin:delete433};
			\cmntwa{create new log entry in log}
			
			\State le $\gets$ new le$\langle obj\_id \downarrow, key\downarrow\rangle$\label{lin:delete39};
			\State \llsval{$NULL \downarrow$} \label{lin:delete40};
			
			\State \llspc{} \label{lin:delete41};
			
			\EndIf\label{lin:delete47}

			\State return $\langle value, op\_status\rangle$
			\EndFunction 
		\end{algorithmic}
\end{algorithm}


\begin{algorithm}[H]
	\caption{\tabspace[0.2cm] STM\_tryC($t\_id \downarrow, tx\_status \uparrow$) 
	}
	\label{algo:trycommit}
		\begin{algorithmic}[1]
			\makeatletter\setcounter{ALG@line}{136}\makeatother
			\Function{STM\_tryC}{} \label{lin:tryc1}
			
			\cmntwa{get the txlog of the current transaction by t\_id}	
			\State $ll\_list$ $\gets$ \txgllist \label{lin:tryc3};
			\cmntwa{sort the local log in increasing order of keys and copy into ordered list}
			\State $ordered\_ll\_list$ $\gets$ \llsort{} \label{lin:tryc4};
			\cmntwa{identify the new preds and currs for all update methods of a tx and validate it}
			\While{$(\textbf{$le_i \gets \textup{next}(ordered\_ll\_list$}))$} \label{lin:tryc5}
			\State ($key, obj\_id$) $\gets$ \llgkeyobj{} \label{lin:tryc6};
			\cmntwa{search correct location for the operation over lsl and lock the corresponding \preds and \currs}
			\State \lsls{$TRYC \downarrow$} \label{lin:tryc7};
			\cmntwa{if lslSearch return op\_status as ABORT then method will return ABORT}
			\If{$(op\_status$ $=$ \textup{ABORT}$)$} \label{lin:tryc8}
			\cmntwa{release local memory in case lslSearch returns abort}
			\State \handlea{} \label{lin:tryc9};
			\State return $\langle op\_status\rangle$;
			
			\EndIf \label{lin:tryc11}
			\cmntwa{modify the log entry to help upcoming update method of same tx}
			\State \llspc{} \label{lin:tryc12};
			\EndWhile \label{lin:tryc13}
			\cmntwa{get each update method one by one and take the effect in underlying DS}
			\While{$(\textbf{$le_i \gets \textup{next}(ordered\_ll\_list$}))$} \label{lin:tryc14}
			\State ($key, obj\_id$) $\gets$ \llgkeyobj{} \label{lin:tryc15};
			\cmntwa{get the operation name to local log entry}
			\algstore{testcont} 
		\end{algorithmic}
	\end{algorithm}
	
	\begin{algorithm}[H]
		\begin{algorithmic}[1]
			\algrestore{testcont} 
			
			\State opn $\gets$ $le_i$.opn \label{lin:tryc16};
			\State intraTransValdation($le_i \downarrow, \preds \uparrow, \currs \uparrow$) \label{lin:tryc17};
			\cmntwa{if operation is insert then after successful completion of it node corresponding to the key should be part of \bn}
			\If{$($\textup{INSERT} $=$ \textup{opn}$)$} \label{lin:tryc18}
			\cmntwa{if node corresponding to the key is part of \bn}
			\If{$\texttt{read}(\bc.\textup{key}) = key)$} \label{lin:tryc19}
			\cmntwa{get the value from local log}
			
			\State $value$ $\gets$ \llgval{} \label{lin:tryc20};
			\cmntwa{update the value into underlying DS}	
			\State $\texttt{write}$($\bc$.value, $value$) \label{lin:tryc21};
			\cmntwa{update the max\_ts of insert for node corresponding to the key into underlying DS}
			\State $\texttt{write}$($\bc$.max\_ts.insert, TS($t\_id$)) \label{lin:tryc23}; 
			\cmntwa{if node corresponding to the key is part of \rn}
			\ElsIf{$(\texttt{read}($$\rc$.\textup{key}$)$ $=$ key$)$} \label{lin:tryc24}
			\cmntwa{connect the node corresponding to the key to \bn as well}
			\State \lslins{\textcolor{red}{$RL$}$\_$\textcolor{blue}{$BL$} $\downarrow$} \label{lin:tryc25};
			
			\cmntwa{update the max\_ts of insert for node corresponding to the key into underlying DS}
			\State $\texttt{write}$($\rc$.max\_ts.insert, TS($t\_id$)) \label{lin:tryc27};
			
			\Else \label{lin:tryc28}
			\cmntwa{if node corresponding to the key is not part of \bn as well as \rn then create the node with the help of lslIns() and add it into \bn}
			\State \lslins{\textcolor{blue}{$BL$} $\downarrow$} \label{lin:tryc29};
			\cmntwa{update the max\_ts of insert for node corresponding to the key into underlying DS}
			\State $\texttt{write}$(node.max\_ts.insert, TS($t\_id$)) \label{lin:tryc31};
			\cmntwa{need to update the node field of log so that it can be released finally}
			\State $le_{i}$.node $\gets$ $\bp.\bn$ 
			
			\EndIf \label{lin:tryc32}

			\cmntwa{if operation is delete then after successful completion of it node corresponding to the key should not be part of \bn}
			\ElsIf{$($\textup{DELETE} $=$ \textup{opn}$)$} \label{lin:tryc33}
			\cmntwa{if node corresponding to the key is part of \bn}
			
			\If{$(\texttt{read}($$\bc$.\textup{key}$)$ $=$ key$)$} \label{lin:tryc34}
			\cmntwa{delete the node corresponding to the key from the \bn with the help of lslDel()}
			\State \lsldel{} \label{lin:tryc35};
			\cmntwa{update the max\_ts of delete for node corresponding to the key into underlying DS}
			\State $\texttt{write}$($\bc$.max\_ts.delete, TS($t\_id$)) \label{lin:tryc37};
			
			
			
			
			\EndIf \label{lin:tryc41}
			
			\EndIf \label{lin:tryc42}
			\cmntwa{modify the preds and currs for the consecutive update methods which are working on overlapping zone in lazyskip-list}
			\EndWhile \label{lin:tryc43}
			\cmntwa{release all the locks}
			\State \rlsol{} \label{lin:tryc44};  
			\cmntwa{set the tx status as OK}
			\State $tx\_status$ $\gets$ OK \label{lin:tryc45};
			
			\State return $\langle tx\_status\rangle$\label{lin:tryc47};
			\EndFunction \label{lin:tryc48}
		\end{algorithmic}
\end{algorithm}


\begin{algorithm}[H]
	\caption{\tabspace[0.2cm] rblIns($\preds \downarrow, \currs \downarrow, list\_type \downarrow$) : Inserts or overwrites a node in underlying hash table at location corresponding to $preds$ \& $currs$. 
	}
	\label{algo:lslins}
		\begin{algorithmic}[1]
			\makeatletter\setcounter{ALG@line}{238}\makeatother
			\Function{rblIns}{} \label{lin:lslins1}
			\cmntwa{inserting the node which is red list to bluelist}
			\If{$((list\_type)$ $=$ $($\textcolor{red}{$RL$}$\_$\textcolor{blue}{$BL$}$))$} \label{lin:lslins2}
			
			\State $\texttt{write}$($\rc$.marked, false) \label{lin:lslins3}; 
			\State $\texttt{write}$($\rc$.\bn, $\bc$) \label{lin:lslins4};
			\State $\texttt{write}$($\bp$.\bn, $\rc$) \label{lin:lslins5};
			\cmntwa{inserting the node into red list only}
			\ElsIf{$((list\_type$) $=$ \textcolor{red}{$RL$}$)$} \label{lin:lslins6}
			\State node = Create new node()
			\label{lin:lslins7};
			\State $\texttt{write}$(node.marked, True) \label{lin:lslins8};
			\State $\texttt{write}$(node.\rn, $\rc$) \label{lin:lslins9};
			\State $\texttt{write}$($\rp$.\rn, node) \label{lin:lslins10};
			
			\Else \label{lin:lslins11}
			\cmntwa{inserting the node into red as well as blue list}
			\State node = new node() \label{lin:lslins12}; 
			\cmntwa{after creating the node acquiring the lock on it}
			\State node.lock();
			\State $\texttt{write}$(node.\rn, $\rc$) \label{lin:lslins13};
			\State $\texttt{write}$(node.\bn, $\bc$) \label{lin:lslins14};
			
			\State $\texttt{write}($$\rp$.\rn, node ) \label{lin:lslins15};
			
			\State $\texttt{write}$($\bp$.\bn, node) \label{lin:lslins16};
			\EndIf \label{lin:lslins17}
			\State return $\langle \rangle$;
			\EndFunction \label{lin:lslins18}
		\end{algorithmic}
\end{algorithm}

\begin{figure}[H]
	\centering
	\captionsetup{justification=centering}
	\centerline{\scalebox{0.55}{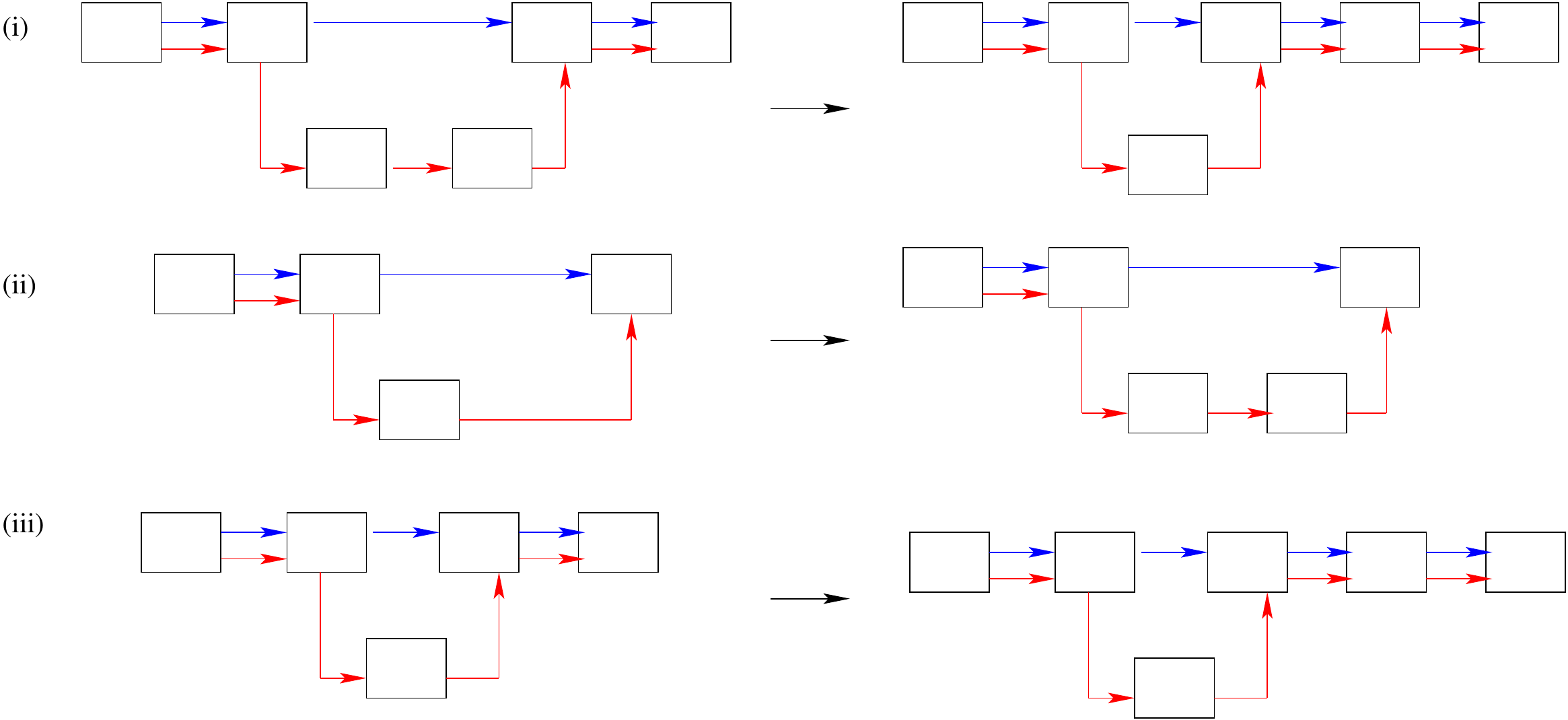}}
	\caption{Execution of \nplslins{}: (i) key $k_5$ is present in $\rn$ and adding it into $\bn$, (ii) key $k_5$ is not present in $\rn$ as well as $\bn$ and adding it into $\rn$, (iii) key $k_5$ is not present in $\rn$ as well as $\bn$ and adding it into $\rn$ as well as $\bn$}
	\label{fig:nostm31}
\end{figure}

\nplslins{} (\algoref{lslins}) adds a new node to the \lsl{} in the \tab{}. There can be following cases:
\textbf{if node is present in \rn{} and has to be inserted to \bn:} such a case implies that the \nplslins{} is invoked in \cp{} for the corresponding \npins{} in local log represented by the block from \Lineref{lslins2} to \Lineref{lslins5}. Here we first reset the \rc{} mark field and update the $\bn$ to the \bc{} and \bp{} $\bn$ to \rc{}. Thus the node is now reachable by $\bn$ also. \figref{nostm31}(i) represents the case.
\textbf{if node is meant to be inserted only in \rn:}
This implies that the node is not present at all in the \lsl{} and is to be inserted for the first time. Such a case can be invoked from \rvmt{} of \rvp{}, if \rvmt{} is the first \mth{} of its transaction. \Lineref{lslins6} to \Lineref{lslins10} depict such a case where a new $node$ is created and its $marked$ field is set, depicting that its a dead node meant to be reachable only via \rn{}. In \Lineref{lslins9} and \Lineref{lslins10} the \rn{} field of the $node$ is updated to \rc{} and \rn{} field of the \rp{} is modified to point to the $node$ respectively. \figref{nostm31}(ii) represents the case.
\textbf{if node is meant to be inserted in \bn:} In such a case it may happen that the node is already present in the \rn{} (already covered by \Lineref{lslins2} to \Lineref{lslins5}) or the node is not present at all. The later case is depicted in \Lineref{lslins11} to \Lineref{lslins16} which creates a new $node$ and add the node in both \rn{} and \bn{} note that order of insertion is important as the \lsl{} can be concurrently accessed by other transactions since traversal is lock free. \figref{nostm31}(iii) represents the case.


\begin{algorithm}[H]
	\caption{\tabspace[0.2cm] rblDel($\preds \downarrow, \currs \downarrow$) : Deletes a node from blue link in underlying hash table at location corresponding to $preds$ \& $currs$.
	}
	\label{algo:lsldelete}
		\begin{algorithmic}[1]
			\makeatletter\setcounter{ALG@line}{262}\makeatother
			\Function{rblDel}{} \label{lin:lsldel1}
			
			\cmntwa{mark the node$ \langle obj\_id, key \rangle$ for deletion}
			\State $\texttt{write}$($\bc$.marked, True) \label{lin:lsldel2};
			\cmntwa{set the update the blue links}
			\State $\texttt{write}$($\bp$.\bn, $\bc$.\bn) \label{lin:lsldel3};
			\State return $\langle \rangle$;
			\EndFunction \label{lin:lsldel4}
		\end{algorithmic}
		
\end{algorithm}

\nplsldel{} removes a node from \bn{}. It can be invoked from \cp{} for corresponding \npdel{} in $txlog$. It simply sets the marked field of the node to be deleted (\bc{}) and changes the \bn{} of \rp{} to \rc{} as shown in \Lineref{lsldel2} and \Lineref{lsldel3} of \algoref{lsldelete} respectively. \figref{1nostm9} shows the deletion of node corresponding to $k_5$.
\begin{figure}[H]
	\centering
	\captionsetup{justification=centering}
	\centerline{\scalebox{0.55}{\input{figs/nostm9.pdf_t}}}
	\caption{Execution of \nplsldel{}: (i) \lsl{} before $k_5$ is deleted, (ii) \lsl{} after $k_5$ is deleted from $\bn$}
	\label{fig:1nostm9}
\end{figure}
\rvmt{} and \upmt{} do the \emph{validation} in \rvp{} and \cp{} respectively. \emph{validation} invokes \npintv{} and then does the \nptov{} in the mentioned order. \npintv{} is the property of the method and \nptov{} is the property of the transaction. Thus validating the \mth{} before the transaction intuitively make sense.



\begin{algorithm}[H]
	\caption{\tabspace[0.2cm] validation($t\_id \downarrow, key \downarrow, \preds \downarrow, \currs \downarrow, val\_type \downarrow, op\_status \uparrow$)
	}
	\label{algo:validation}
		\begin{algorithmic}[1]
			\makeatletter\setcounter{ALG@line}{269}\makeatother
			\Function{validation}{} \label{lin:validation1}
			\cmntwa{validate against concurrent updates}
			\State $op\_status$ $\gets$ methodValidation($\preds \downarrow$, $\currs \downarrow$)\label{lin:validation2};
			\cmntwa{on succesfull method validation validate of transactional ordering to ensure opacity}
			\If{$( RETRY \neq op\_status)$} \label{lin:validation3}
			
			\State $op\_status$ $\gets$ \toval{} \label{lin:validation4};
			\EndIf \label{lin:validation5}
			\State return $\langle op\_status\rangle$ \label{lin:validation6}; 
			\EndFunction \label{lin:validation7}
		\end{algorithmic}
		
\end{algorithm}


\begin{algorithm}[H]
	\caption{\tabspace[0.2cm] methodValidation($\preds \downarrow, \currs \downarrow$) 
	}
	\label{algo:interferenceValidation}
	\begin{algorithmic}[1]
		\makeatletter\setcounter{ALG@line}{278}\makeatother
		\Function{methodValidation}{} \label{lin:iv1}
		\If{$(\texttt{read}(\bp.marked) || \texttt{read}(\bc.marked) ||\texttt{read}(\bp.\bn) \neq \bc || \texttt{read}(\rp.\rn) \neq {\rc})$} \label{lin:iv2}
		
		\State return $\langle RETRY\rangle$  \label{lin:iv3};
		
		\Else \label{lin:iv4}
		
		\State return $\langle OK\rangle$ \label{lin:iv5};
		
		\EndIf \label{lin:iv6}
		\EndFunction \label{lin:iv7}
	\end{algorithmic}
\end{algorithm}

In \nptov{} \rvmt{} always conflicts with the \upmt{} (as established in conflict notion \secref{conflicts}). If the node corresponding to the $key$ is present in the \lsl{} (\Lineref{tov5}) we compare with time-stamp of the transaction that last executed the conflicting \mth{} on same $key$. If the current \mth{} that invoked the \nptov{} is \rvmt{} then \Lineref{tov6} handles the case.
Otherwise, if the invoking \mth{} is \upmt{} then \Lineref{tov9} handles the case. \figref{1nostm25} and \figref{1nostm26} show the execution of \nptov{}. Here $Lu_1(ht, k_1)$ will return $Abort$ in  \figref{1nostm26} because $d_2((ht, k_1)$ of $T_2$ has already updated the time-stamp at the node corresponding to $k_1$. So, when $l_1(ht, k_1)$ does its \nptov{} at \Lineref{tov9}, $TS(t_1)$ $<$ $curr.max\_ts.delete(k)$ holds true (since, $T_1$ $<$ $T_2$) leading to $abort$ of $T_1$ at \Lineref{tov11}. This gives us a equivalent sequential schedule which can be shown \coop{}. \figref{1nostm25} shows the schedule where no sequential schedule is possible if \nptov{} is not applied as there is no way to recognize the time-order violation.

\begin{figure}[H]
	\centering
	\begin{minipage}[b]{0.49\textwidth}
		\scalebox{.5}{\input{figs/nostm25.pdf_t}}
		\centering
		\captionsetup{justification=centering}
		\captionof{figure}{non opaque history. Without time-stamp validation}
		\label{fig:1nostm25}
	\end{minipage}
	\hfill
	\begin{minipage}[b]{0.49\textwidth}
		\centering
		\captionsetup{justification=centering}
		\scalebox{.5}{\input{figs/nostm26.pdf_t}}
		\captionof{figure}{opaque history H1. With time-stamp validation}
		\label{fig:1nostm26}
	\end{minipage}   
\end{figure}


\begin{algorithm}[H]
	\caption{\tabspace[0.2cm] transValidation($t\_id \downarrow, key \downarrow, \currs \downarrow, val\_type \downarrow, op\_status \uparrow$) : Time-order validation for each transaction.
	}
	\label{algo:tovalidation}
	\begin{algorithmic}[1]
		\makeatletter\setcounter{ALG@line}{285}\makeatother
		\Function{transValidation}{} \label{lin:tov1}
		\cmntwa{by default setting the op\_status as RETRY}
		\State STATUS $op\_status$ $\gets$ OK \label{lin:tov3};
		\cmntwa{get the appropriate $sh\_curr$ (red or blue) correspondinjg to key}
		\State \llgaptc{$\currs \downarrow$, $key\downarrow$, $sh\_curr \uparrow$} \label{lin:tov4};
		
		\cmntwa{if $sh\_curr$ is not NULL and node corresponding to the key is equal to $sh\_curr$.key then check for TS}
		\If{$(($\textup{$sh\_curr$} $\neq$ \textup{NULL}$)\land(($\textup{$sh\_curr$.key}$)$ $=$ \textup{key}$))$} \label{lin:tov5}
		\algstore{testcont} 
	\end{algorithmic}
\end{algorithm}

\begin{algorithm}[H]
	\begin{algorithmic}[1]
		\algrestore{testcont} 
		
		\cmntwa{if val\_type is RV then transaction validation for \rvmt{}}
		\If{$ ((val\_type = RV) \land ($\textup{TS}$(t\_id)$ $<$ $(\texttt{read}($\textup{$sh\_curr$.max\_ts.insert}$($\textup{k}$)))$ $||$ \\\hspace{2cm} $($\textup{TS}($t\_id$) $<$ $(\texttt{read}($\textup{$sh\_curr$.max\_ts.delete}$($\textup{k}$)))))$} \label{lin:tov6}
		\State $op\_status$ $\gets$ ABORT \label{lin:tov8};
		\cmntwa{transaction validation for \upmt{}}
		\ElsIf{$(($\textup{TS}$(t\_id)$ $<$ $(\texttt{read}($\textup{$sh\_curr$.max\_ts.insert}$($\textup{k}$)))$ $||$ \textup{TS}$(t\_id)$ $<$ $(\texttt{read}($\textup{$sh\_curr$.max\_ts.delete}$($\textup{k}$)))$ $||$\\ \hspace{2cm} \textup{TS}$(t\_id)$ $<$ $(\texttt{read}($\textup{$sh\_curr$.max\_ts.lookup}$($\textup{k}$))))$} \label{lin:tov9}
		
		\State $op\_status$ $\gets$ ABORT \label{lin:tov11};
		\EndIf \label{lin:tov12}
		
		\EndIf \label{lin:tov13}
		
		\State return $\langle op\_status\rangle$ \label{lin:tov14}; 
		
		\EndFunction \label{lin:tov15}
	\end{algorithmic}
\end{algorithm}

\nppoval{} handles the case where two consecutive updates within same transaction having overlapping $preds$ and $currs$ may overwrite the previous \mth{} such that only effect of the later \mth{} is visible. This happens because the previous \mth{} while updating, changes the \lsl{} causing the $preds$ \& $currs$ of the next \mth{} working on the consecutive key to become obsolete. Thus, \nppoval{} corrects this by finding the new $preds$ and $currs$ of the current \mth{} on the consecutive key. There might be two cases (i) if previous \mth{} is \npins{} or (ii) previous \mth is \npdel{}. For case(i) we find the \bp{} (at \Lineref{threadv4} to \Lineref{threadv5} using previous log entry) and for case(ii) we find \bp{} using previous log entry's \bp{} (\Lineref{threadv7}) and finally find the new \rp{} and \rc{} between the new found \bp{} and \bc{} at \Lineref{threadv10} to \Lineref{threadv11}.


\begin{algorithm}[H]
	\caption{\tabspace[0.2cm] intraTransValidation($le \downarrow, \preds \uparrow, \currs \uparrow$)
	}
	\label{algo:povalidation}
		\begin{algorithmic}[1]
			\makeatletter\setcounter{ALG@line}{304}\makeatother
			\Function{intraTransValidation}{} \label{lin:threadv1}
			\State $le.getAllPreds\&Currs(le$ $\downarrow$, $\preds$ $\uparrow$, $\currs$ $\uparrow$) \label{lin:threadv2};	
			\cmntwa{if $\bp$ is marked or $\bc$ is not reachable from $\bp.\bn$ then modify the next consecutive \upmt{} $\bp$ based on previous \upmt{}}
			\If{$((\texttt{read}($$\bp$.\textup{marked}$)) ||$ $(\texttt{read}($ $\bp$.\textup{\bn}$)$ != $\bc$$))$} \label{lin:threadv3}
			\cmntwa{find $k$ < $i$; such that $le_k$ contains previous update method on same bucket }
			\If{$(($$le_{k}$.\textup{opn}$)$ $=$ INSERT$)$} \label{lin:threadv4}
			
			\State $le_{i}$.$\bp$.$\texttt{unlock()}$ \label{lin:thredv4-5};
			\State $\bp$ $\gets$ ($le_{k}.\bp.\bn)$ \label{lin:threadv5};
			\State $le_{i}$.$\bp$.$\texttt{lock()}$ \label{lin:thredv5-5};	
			\algstore{testcont} 
		\end{algorithmic}
	\end{algorithm}
	
	\begin{algorithm}[H]
		\begin{algorithmic}[1]
			\algrestore{testcont} 
				
			\Else \label{lin:threadv6}
			\cmntwa{\upmt{} method $\bp$ will be previous method $\bp$}
			\State $le_{i}$.$\bp$.$\texttt{unlock()}$ \label{lin:thredv6-5};
			\State $\bp$ $\gets$ ($le_{k}$.$\bp$) \label{lin:threadv7};
			
			\State $le_{i}$.$\bp$.$\texttt{lock()}$ \label{lin:thredv7-5};
			\EndIf \label{lin:threadv8}
			
			\EndIf \label{lin:threadv9}
			\cmntwa{if $\rc$ \& $\rp$ is modified by prev operation then update them also}
			\If{$(\texttt{read}($$\rp$.\textup{\rn}$)$ != $\rc$$)$} \label{lin:threadv10}
			\State $le_{i}$.$\rp$.$\texttt{unlock()}$
			\State $\rp$ $\gets$ ($le_{k}$.$\rp.\rn$)  \label{lin:threadv11}; 
			\State $le_{i}$.$\rp$.$\texttt{lock()}$
			\EndIf \label{lin:threadv12}
			\State return $\langle \preds, \currs\rangle$;
			\EndFunction \label{lin:threadv13}
		\end{algorithmic}
\end{algorithm}

\emph{findInLL} is an utility \mth{} that returns true to the \mth{} that has invoked it, if the calling \mth{} is not the first method of the transaction on the $key$. This is done by linearly traversing the log and finding an entry corresponding to the $key$. If the calling \mth{} is the first \mth{} of the transaction for the $key$ then \emph{findInLL} return false as it would not find any entry in the log of the transaction corresponding to the $key$. Since we consider that their can be multiple objects (\tab{}) so we need to find unique $\langle obj\_id, key \rangle$ pair (refer \Lineref{findll5}).

\begin{algorithm}[H]
	\caption{\tabspace[0.2cm] findInLL($t\_id \downarrow, obj\_id \downarrow, key \downarrow, le \uparrow$) : Checks whether any operation corresponding to $\left\langle obj\_id, key  \right\rangle$ is present in ll\_list.
	}
	\label{algo:findInLL}
	\begin{algorithmic}[1]
		\makeatletter\setcounter{ALG@line}{328}\makeatother
		\Function{findInLL}{} \label{lin:findll1}
		
		\State ll\_list $\gets$ \txgllist{} \label{lin:findll3}; 
		\cmntwa{every method first identify the node corresponding to the key into local log}
		\While{$(le_i \gets next$($ll\_list$$))$} \label{lin:findll4}
		\If{$((le_i.first = obj\_id) \& (le_i.first = key))$} \label{lin:findll5}
		
		\State return $\langle TRUE, le\rangle$  \label{lin:findll6};
		\EndIf \label{lin:findll7}
		\EndWhile \label{lin:findll8}
		\State return $\langle FALSE, le = NULL\rangle$ \label{lin:findll9};
		\EndFunction \label{lin:findll10}
	\end{algorithmic}
\end{algorithm}

While executing the \nptov{} the time-stamp field of the corresponding $node$ has to be updated. Such a node can be either the marked (dead or \rc) or the unmarked (live \bc). \emph{get\_aptcurr} is the utility \mth{} which returns the appropriate $node$ corresponding to the $key$.


\begin{algorithm}[H]
	\caption{\tabspace[0.2cm] get\_aptcurr($\currs \downarrow$, $key \downarrow, sh\_curr \uparrow$) : Returns a curr node from underlying DS which corresponds to the key of $le_i$.
	}
	\label{algo:getaptcurr}
		\begin{algorithmic}[1]
			\makeatletter\setcounter{ALG@line}{338}\makeatother
			\Function{get\_aptcurr}{} \label{lin:aptcurr1}
			\cmntwa{by default set curr to NULL}
			\State sh\_curr $\gets$ NULL;
			\cmntwa{if node corresponding to the key is part of $\bn$ then curr is $\bc$}
			\If{$($$\bc$.\textup{key} = \textup{key}$)$} \label{lin:aptcurr2}
			\State sh\_curr $\gets$ $\bc$ \label{lin:aptcurr3};
			\cmntwa{if node corresponding to the key is part of $\rn$ then curr is $\rc$}
			\ElsIf{$($$\rc$.\textup{key} = \textup{key}$)$} \label{lin:aptcurr4}
			\State sh\_curr $\gets$ $\rc$ \label{lin:aptcurr5};
			
			\EndIf \label{lin:aptcurr6}
			
			\State return $\langle sh\_curr\rangle$ \label{lin:aptcurr7}; 
			
			\EndFunction \label{lin:aptcurr8}
		\end{algorithmic}
\end{algorithm}

\emph{release\_ordered\_locks} is an utility \mth{} to release the locks in order.


\begin{algorithm}[H]
	\caption{\tabspace[0.2cm] release\_ordered\_locks($ordered\_ll\_list \downarrow$) : Release all locks taken during \nplsls{}.
	}
	\label{algo:releaseorderedlocks}
		\begin{algorithmic}[1]
			\makeatletter\setcounter{ALG@line}{350}\makeatother
			\Function{release\_ordered\_locks}{} \label{lin:rlock1}
			\cmntwa{releasing all the locks on preds, currs and node}
			\While{($\textbf{$le_i \gets next(ordered\_ll\_list$})$)} \label{lin:rlock2}
			
			\State $le_i$.$\bp$.$\texttt{unlock()}$ \label{lin:rlock3};//$\Phi_{lp}$ 
			\State $le_i$.$\rp$.$\texttt{unlock()}$ \label{lin:rlock4};
			\If{$le_i$.$node$}
			\State{$le_i$.$node$.$\texttt{unlock()}$}
			\EndIf
			\State $le_i$.$\rc$.$\texttt{unlock()}$ \label{lin:rlock5};
			\State $le_i$.$\bc$.$\texttt{unlock()}$ \label{lin:rlock6}; 
			\EndWhile \label{lin:rlock7}
			\State return $\langle \rangle$;
			\EndFunction \label{lin:rlock8}
		\end{algorithmic}
		
\end{algorithm}


\begin{algorithm}[H]
	\caption{\tabspace[0.2cm] acquirePred\&CurrLocks($ \preds \downarrow$, $ \currs \downarrow$) : acquire all locks taken during \nplsls{}.
	}
	\label{algo:acquirepreds&currs}
		\begin{algorithmic}[1]
			\makeatletter\setcounter{ALG@line}{360}\makeatother
			\Function{acquirePred\&CurrLocks}{}
			\State $\bp$.$\texttt{lock()}$;
			\State $\rp$.$\texttt{lock()}$;
			\State $\rc$.$\texttt{lock()}$;
			\State $\bc$.$\texttt{lock()}$;
			\State return $\langle \rangle$;        
			\EndFunction
		\end{algorithmic}
		
\end{algorithm}


\begin{algorithm}[H]
	\caption{\tabspace[0.2cm] releasePred\&CurrLocks($ \preds \downarrow$, $ \currs \downarrow$) : Release all locks taken during \nplsls{}.
	}
	\label{algo:releasepreds&currs}
		\begin{algorithmic}[1]
			\makeatletter\setcounter{ALG@line}{367}\makeatother
			\Function{releasePred\&CurrLocks}{}
			\State $\bp$.$\texttt{unlock()}$;//$\Phi_{lp}$ 
			\State $\rp$.$\texttt{unlock()}$;
			\State $\rc$.$\texttt{unlock()}$;
			\State $\bc$.$\texttt{unlock()}$;
			\State return $\langle \rangle$;        
			\EndFunction
		\end{algorithmic}
		
\end{algorithm}
}

\ignore{
\noindent Now, we have some theoretical results.
\begin{theorem}
	\label{thm:otm-acyclic}
	Consider a history $H$ generated by \otm. Then there exists a sequential \& legal history $H'$ equivalent to $H$ such that the conflict-graph of $H'$ is acyclic.
\end{theorem}
}
\subsection{\otm{} execution cycle} 
\begin{figure}[H]
	\centering
	\captionsetup{justification=centering}
	\centerline{\scalebox{0.48}{\input{figs/nostm24.pdf_t}}}
	\caption{Transaction lifecycle of \otm}
	\label{fig:nostm24}
\end{figure}

Through out its life an \otm{} transaction may execute \emph{STM\_begin()}, \npins{}, \npluk{}, \npdel{} and \nptc{} \mth{s} which are also exported to the user. A user can implement his/her applications using \otm{} which would provide efficient composability. Each transaction has a 1) \rvp: where \upmt{} \& \rvmt{} locally identify and logs the location to be worked upon and other meta information which would be needed for successful validation. Within \rvp{} \rvmt{s} do lock free traversal and then validate. And, \npins{} merely log its execution to be validated and updated during transaction commit. 2) \cp: where it validates the \upmt{} executed during its lifetime and validates whether the transaction will commit and finally make changes in \tab{} atomically or it will abort and flush its log. This phase is executed by \nptc{} \mth . \figref{nostm24} depicts the transaction life cycle.

\textbf{Pseudocode convention:} In each algorithm $\downarrow$ represents the input parameter and $\uparrow$ shows the output parameter (or return value) of the corresponding methods (such in and out variables are italicized). Instructions in \emph{read()} and \emph{write()} with in each \mth{} denote that they touch the shared memory. The variable prefixed with $sh\_$ are shared memory variables and can be accessed by multiple transactions concurrently, for instance $sh\_preds[]$. \bp{} \& \bc{} depict the blue nodes accessible by blue links and \rp{} \& \rc{} depict the red nodes accessed by red links respectively. 

\textbf{\rvmt{} execution phase:}
Initially, in \emph{\rvmt{} execution} phase each transaction invokes \emph{STM\_begin()} of \algoref{begin} for getting unique transaction id and \emph{local log}. Then transaction may encounter the \upmt{} or \rvmt{}. \npins{} of \algoref{insert}, first looks for the node corresponding to the $key$ into the ll\_list (\Lineref{insert2}). If $key$ is not found then it will create the $le$ and store the value, operation name and status (\Lineref{insert3} to \Lineref{insert7}) into it which would be validated and realized in shared memory in \nptc{}.

\nptc{} and \rvmt{} of \otm{} uses \nplsls{} to find the location at the \lsl{} (thus the name) in lock free manner. \Lineref{lslsearch3} to \Lineref{lslsearch8} and \Lineref{lslsearch11} to \Lineref{lslsearch15} of \algoref{lslsearch} find the location at \lsl{} for $\bn$ and $\rn$ respectively. This is motivated by the search in lazylist \cite[section 9.7]{Herlihy:ArtBook:2012}. The $preds$ and $currs$ thus identified are subjected to \npintv{} of \algoref{interferenceValidation} and \nptov{} of \algoref{tovalidation} after acquiring locks on the $preds$ and $currs$ (\Lineref{lslsearch17} of \algoref{lslsearch}). If the validation succeeds \nplsls{} returns the correct location to the operation which invoked it, otherwise \nplsls{} retries (if concurrent interference detected) or aborts (if time order violated) post releasing locks (\Lineref{lslsearch22}).  

Interference validation helps detecting the execution where underlying data structure has been changed by second concurrent transaction while first was under execution without it realizing. This can be illustrated with \figref{nostm8}. Consider the history in \figref{nostm8}(iii) where two conflicting transactions $T_1$ and $T_2$ are trying to access key $k_5$, here $s_1$, $s_2$ and $s_3$ represent the state of the \lsl{} at that instant. Let at $s_1$ both the methods record the same $preds \langle k_1, k_3 \rangle $ and $currs \langle k_5, k_5 \rangle$ with the help of $\nplsls{}$ for key $k_5$ (refer \figref{nostm8}(i)). Now, let $d_1(k_5)$ acquire the lock on the $preds$ and $currs$ before the $l_2(k_5)$ and delete the node corresponding to the key $k_5$ from $\bn$ leading to state $s_2$ (in \figref{nostm8}(iii)) and commit. \figref{nostm8}(ii) shows the state $s_2$ where key $k_5$ is the part of $\rn$. Now, \npintv{} (in \algoref{interferenceValidation}) will identify that location of $l_2(k_5)$ is no more valid due to ($\bp.\bn$ $\neq$ $\bc$) at \Lineref{iv2} of \algoref{interferenceValidation}. Thus, $\nplsls{}$ will retry to find the updated location for $l_2(k_5)$ at state $s_3$ (in \figref{nostm8}(iii)) and eventually $T_2$ will commit.
\begin{figure}[H]
	\centering
	\captionsetup{justification=centering}
	\centerline{\scalebox{0.5}{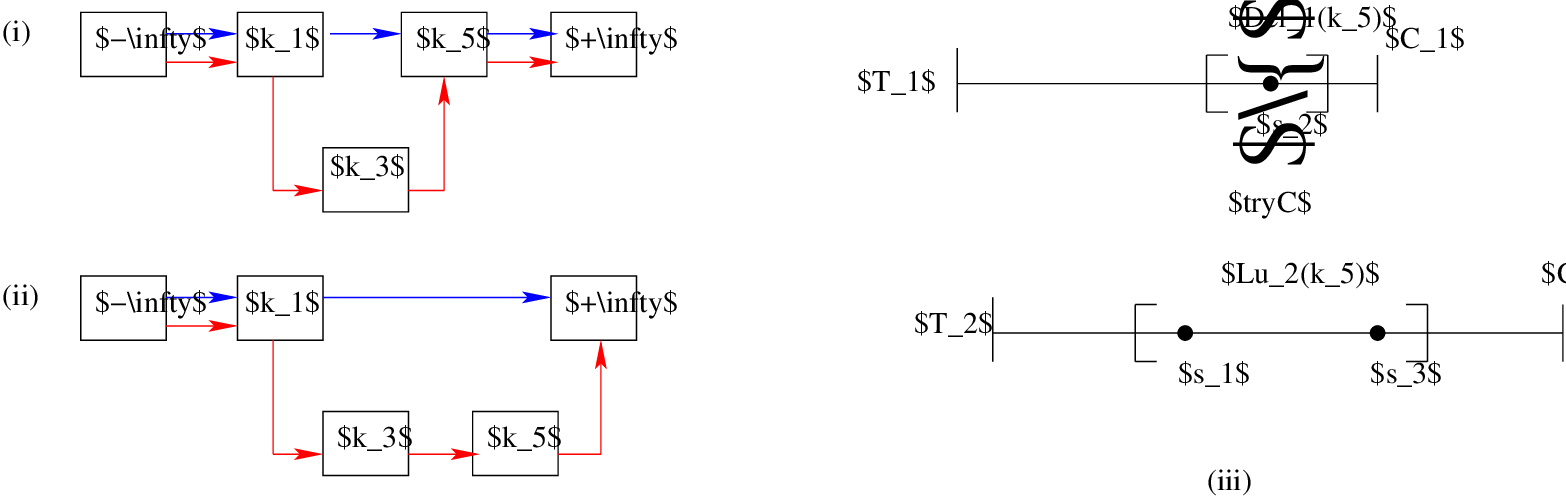}}
	\caption{Interference Validation for conflicting concurrent methods on key $k_5$}
	\label{fig:nostm8}
\end{figure}

\npluk{} \& \npdel{} behaves similarly during \emph{\rvmt{} execution} phase execept that \npdel{} is validated twice. First, in \emph{\rvmt{} execution} similar to \npluk{} and secondly in \emph{\upmt{} execution} (of \nptc{}) to ensure opacity\cite{GuerKap:2008:PPoPP}. We adopt lazy delete approach for \npdel{} \mth{}. Thus, nodes are marked for deletion and not physically deleted for \npdel{} \mth{}. In the current work we assume that a garbage collection mechanism is present and we donot worry about it.

\textbf{\upmt{} execution phase:} Finally a transaction after executing the designated operations reaches the \emph{\upmt{} execution} phase executed by the \nptc{} method. It starts with modifying the log to $ordered\_ll\_list$ which contains the log entries in sorted order of the keys (so that locks can be acquired in an order, refer \Lineref{tryc4} of \algoref{trycommit}) and contains only the \upmt{} (because we do not validate the lookup again for the reasons explained above for \figref{nostm19}). From \Lineref{tryc5} to \Lineref{tryc12} (in \algoref{trycommit}) we re-validate the modified log operation to ensure that the location for the operations has not changed since the point they were logged during \rvp{}. If the location for an operation has changed this block ensures that they are updated.

Now, \nptc{} enters the phase where it updates the shared memory using local data stored from \Lineref{tryc14} to \Lineref{tryc37} in \algoref{trycommit}. \figref{nostm10} \& \figref{nostm9} explain the execution of insert and delete in update phase of \nptc{} using \nplslins{} and \nplsldel{} respectively. \figref{nostm10}(i) represents the case when $k_5$ is neither present in $\bn$ and nor in $\rn$ (\Lineref{tryc28} to \Lineref{tryc31} in \algoref{trycommit}). It adds $k_5$ to \lsl{} at location $preds\langle k_3, k_4 \rangle $ and $currs\langle k_8, k_8 \rangle$. \figref{nostm10}(i)(a) is \lsl{} before addition of $k_5$ and \figref{nostm10}(i)(b) is \lsl{} state post addition. Similarly, \figref{nostm10}(ii) represents the case when $k_5$ is present in $\rn$ (\Lineref{tryc24} to \Lineref{tryc27} in \algoref{trycommit}). It adds $k_5$ to \lsl{} at location $pred \langle k_3, k_4 \rangle $ and $curr \langle k_5, k_8 \rangle$. \figref{nostm10}(i)(c) is \lsl{} before addition of $k_5$ into \bn{} and \figref{nostm10}(i)(d) is \lsl{} state post addition. In case of $d(k_5)$ from \lsl{} when $k_5$ is present in $\bn$ (\Lineref{tryc33} to \Lineref{tryc37} in \algoref{trycommit}) \figref{nostm9}(i) represent the \lsl{} state before $k_5$ is deleted at location $preds \langle k_1, k_3 \rangle $ and $currs \langle k_5, k_5 \rangle $ and \figref{nostm9}(ii) represents the \lsl{} state after deletion.       

\begin{figure}[H]
	\centering
	\captionsetup{justification=centering}
	\centerline{\scalebox{0.55}{\input{figs/nostm10.pdf_t}}}
	\caption{$i(k_5)$ using \nplslins{} in $\nptc{}$}
	\label{fig:nostm10}
\end{figure}

\begin{figure}[tbph]
	\centering
	\captionsetup{justification=centering}
	\centerline{\scalebox{0.55}{\input{figs/nostm9.pdf_t}}}
	\caption{$d(k_5)$ using \nplsldel{} in $\nptc{}$}
	\label{fig:nostm9}
\end{figure}

\begin{figure}[H]
	\centering
	\captionsetup{justification=centering}
	\scalebox{.5}{\input{figs/nostm23.pdf_t}}
	\captionof{figure}{Problem in execution without \nppoval{} ($i_1(k_5)$ and $i_1(k_7)$). (i) \lsl{} at state s. (ii) \lsl{} at state $s_1$. (iii) \lsl{} at state $s_2$.}
	\label{fig:nostm16}
\end{figure}

\noindent
In \cp{} two consecutive updates within same transaction having overlapping $preds$ and $currs$ may overwrite the previous \mth{} such that only effect of the later \mth{} is visible (\emph{lost update}). This happens because the previous \mth{} while updating, changes the \lsl{} causing the $preds$ \& $currs$ of the next \mth{} working on the consecutive key to become obsolete. \figref{nostm16} explains this lucidly. Suppose, $T_1$ is in update phase of \nptc{} at state $s$ where $i_1(k_5)$ and $i_1(k_7)$ are waiting to take effect over the \lsl{}. The \lsl{} at $s$ is as in \figref{nostm16}(i) also $i_1(k_5)$ and $i_1(k_7)$ have $preds \langle k_3, k_4 \rangle $ and $currs \langle k_8, k_8 \rangle $ as their location. Now, Lets say $i_1(k_5)$ adds $k_5$ between $k_3$ and $k_8$ and changes \lsl{} (as in \figref{nostm16}(ii)) at state $s_1$ in \figref{nostm16}(iv). But, at $s_1$ $\bn$ $preds$ and $currs$ of $i_1(k_7)$ are still $k_3$ and $k_8$ thus it wrongly adds $k_7$ between $k_3$ and $k_8$ overwriting $i_1(k_5)$ as shown in \figref{nostm16}(iii) with dotted links. We correct this through $\nppoval{}$ which updates current \upmt{}'s $preds$ and $currs$ with the help of its $le$. We discuss it in detail at \algoref{povalidation}. Next we elaborate each of the \mth{} exported by \otm{}.

\cmnt{
\textbf{\otm{} \mth{} execution flowchart}

\begin{figure}[H]
	\begin{multicols}{2}
		\includegraphics[scale=0.5]{figs/insertfc.png}\par\caption{flowchart for \npins}
		\includegraphics[scale=0.5]{figs/lookupfc.png}\par\caption{flowchart for \npluk}
	\end{multicols}	
			\includegraphics[scale=0.5]{figs/trycfc.png}\par\caption{flowchart for \nptc}
\end{figure}
}

	\section{\otm{} Pseudocode} \label{sec:pscode}
	\cmnt{
		\begin{figure}[H]
			\centering
			\captionsetup{justification=centering}
			\centerline{\scalebox{0.48}{\input{figs/nostm24.pdf_t}}}
			\caption{Transaction lifecycle of \otm}
			\label{fig:nostm24}
		\end{figure}
		\vspace{-.5cm}
	}
	We now describe the implementation internals of the \otm{}. As discussed in life cycle of each transaction that every \otm{} transcation executes in two phases \rvmt \& \upmt{}. The \mth{s} executed in theses phases are \npbegin{()}, \npluk{}, \npins{}, \npdel{}, \nptc{}. We one by one explain each of the \mths{} in the ensuing text.
	  	
	\vspace{1mm}
	\noindent
	\textbf{\npbegin{}.} is the first function a transaction executes in its life cycle. It initiates the $txlog$ (local log) for the transaction (\Lineref{begin2}) and provides an unique id to the transaction (\Lineref{begin3}).

	\begin{algorithm}[H]
		\caption{\tabspace[0.2cm]  STM\_begin($t\_id\uparrow$) : initiates local transaction log and return the transaction id.}     
		\label{algo:begin}
		\setlength{\multicolsep}{0pt}
		\begin{multicols}{2}
			
			\begin{algorithmic}[1]
				\scriptsize
				\makeatletter\setcounter{ALG@line}{0}\makeatother
				\Function{stm\_begin}{} 
				\label{lin:begin1}
				\cmntwa{ init the local log}
				\State txlog $\gets$ new txlog(); \label{lin:begin2}
				\cmntwa{atomic variable to assign transaction id i.e. TS initilized by OSTM as 0}
				\State t\_id$\gets$ $get\&inc(sh\_cntr \uparrow)$;//$\Phi_{lp}$ \label{lin:begin3}
				\State return $\langle t\_id\rangle$;\label{lin:begin4}
				\EndFunction 
			\end{algorithmic}
			
		\end{multicols}
	\end{algorithm}

		 \noindent
		 \textbf{\npluk{}} in \algoref{lookup}. If this is the subsequent operation by a transaction $T_i$ for a particular key $k$ on hash-table $ht$ i.e. an operation on $k$ has already been scheduled with in the same transaction $T_i$, then this \npluk{} return the value from the $txlog$ and does not access shared memory (\Lineref{lookup3} to \Lineref{lookup10} in \algoref{lookup}). If the last operation was an \npins{} (or \npluk{}) on same key then the subsequent \npluk{} of the same transaction returns the previous value (\Lineref{lookup6} in \algoref{lookup}) inserted (or observed) without accessing shared memory, and if the last operation was an \npdel{} then \npluk{} returns the value NULL (\Lineref{lookup9} in \algoref{lookup}) and is said to have failed. Thus in this process subsequent \mth{s} also have same conflicts as the first \mth{} on same key within the same transaction (\emph{conflict inheritance}) as indicated by LR1 in \subsecref{legal}. 
	
	\begin{algorithm}[H]
		\scriptsize
		\caption{\tabspace[0.2cm] STM\_lookup($t\_id \downarrow, obj\_id \downarrow, key \downarrow, value \uparrow, op\_status \uparrow$ ):If the transaction to which this operation belongs has locally done an operation on the same key then returns apt value and status(wrt the previous local operation). Else do the \nplsls{} to find the correct location of the key and validate it.}
		\label{algo:lookup}
		\setlength{\multicolsep}{0pt}
		\begin{multicols}{2}
			
			\begin{algorithmic}[1]
				\makeatletter\setcounter{ALG@line}{7}\makeatother
				\Function{STM\_lookup}{} \label{lin:lookup1}
				\State STATUS $op\_status$ $\gets$ RETRY \label{lin:lookup2};
				\State{\cmntwa{get the txlog of the current transaction by t\_id}};
				\State txlog $\gets$ getTxLog($t\_id \downarrow$);
				
				\cmntwa{If already in log update the le with the current operation}
				\If{$($\txlfind$)$} \label{lin:lookup3}
				\State opn $\gets$ \llgopn{} \label{lin:lookup4}; 
				
				\cmntwa{if previous operation is insert/lookup then current method would have value/op\_status same as previous log entry}
				\If{$(($\textup{INSERT} $=$ \textup{opn} $)||($ \textup{LOOKUP} $=$ \textup{opn}$))$} \label{lin:lookup5}
				
				\State $value$ $\gets$ \llgval{} \label{lin:lookup6};
				\State $op\_status$ $\gets$  $le.getOpStatus$($obj\_id \downarrow$, $key \downarrow$) \label{lin:lookup7};
				\cmntwa{if previous operation is delete then current method would have value as NULL and op\_status as FAIL}
				\ElsIf{$($\textup{DELETE} $=$ \textup{opn}$)$} \label{lin:lookup8}
				\State $value$ $\gets$ NULL \label{lin:lookup9}; 
				\State $op\_status$ $\gets$ FAIL \label{lin:lookup10}; 
				\EndIf \label{lin:lookup11}
				\Else \label{lin:lookup12}
				\cmntwa{common function for \rvmt{}, if node corresponding to the key is not the part of underlying DS}
				\State \cld{};\label{lin:lookupa20}
				\EndIf \label{lin:lookup39}
				\cmntwa{update the local log}
				\State \llsopn{$LOOKUP \downarrow$} \label{lin:lookup33};
				\State \llsopst{$op\_status \downarrow$} \label{lin:lookup40};
				\State return $\langle value, op\_status\rangle$\label{lin:lookup41}; 
				
				\EndFunction \label{lin:lookup42}
			\end{algorithmic}
			
		\end{multicols}
		
	\end{algorithm}
	
	If \npluk{} is the first operation on a particular key then it has to do a wait free traversal (\Lineref{delete21} in \algoref{commonLu&Del}) with the help of $\nplsls$ (\algoref{lslsearch}) to identify the target node ($preds$ and $currs$) to be logged in $txlog$. These logged $preds$ \& $currs$ are utilized for subsequent \mth{s} in \rvp{} (discussed above for the case where \npluk{} is the subsequent \mth{}). The commonLu\&Del() algorithm is invoked at \Lineref{lookupa20} of \algoref{lookup}. If the node is present as blue (or red) node then it updates the operation status as OK (or FAIL) and returns the value respectively (\Lineref{delete25} to \Lineref{delete32} in \algoref{commonLu&Del}). If node corresponding to the key is not found  then it inserts that node (\Lineref{delete33} to \Lineref{delete37} in \algoref{commonLu&Del}) corresponding to the key into  $\rn$ of \lsl{}. The inserted node can be accessed only via red links. Hence, it will not visible to any subsequent \npluk{}. The node is inserted to take care of situations as  illustrated in \figref{nostm25} \& \figref{nostm26} . Finally, it updates the meta information in $txlog$ and releases the locks acquired inside $\nplsls$ (\Lineref{delete433} to \Lineref{delete41}).   

	\begin{algorithm}
		\scriptsize
		\caption{\tabspace[0.2cm] commonLu\&Del($t\_id \downarrow, obj\_id \downarrow, key \downarrow, value \uparrow, op\_status \uparrow$ ) 
		}
		\label{algo:commonLu&Del}
		\setlength{\multicolsep}{0pt}
		\begin{multicols}{2}
			\begin{algorithmic}[1]
				\makeatletter\setcounter{ALG@line}{33}\makeatother
				\Function{commonLu\&Del}{} 
				
				
				\cmntwa{le$ \langle obj\_id,key \rangle$ is not in log, search correct location for the operation over lsl and lock the corresponding \preds and \currs.}
				\State \lsls{$RV \downarrow$} \label{lin:delete21};
				\If{$(op\_status$ $=$ \textup{ABORT}$)$} \label{lin:delete22}
				\cmntwa{release local memory in case lslSearch returns abort}
				\State \handlea{} \label{lin:delete23}; 
				\State return $\langle op\_status\rangle$;

				\Else \label{lin:delete24}
				\cmntwa{if node$ \langle obj\_id, key \rangle$ is present update its lookup timestamp as delete in rv phase behaves as lookup}
				\If{$(\texttt{read}($$\bc$.\textup{key}$)$ $=$ key$)$} \label{lin:delete25}
				\cmntwa{$node \langle obj\_id, key \rangle$ is part of blue list}
				\State $op\_status$ $\gets$ OK \label{lin:delete26};
				\State $\texttt{write}$($\bc$.max\_ts.lookup, TS($t\_id$)) \label{lin:delete27};
				\State $value$ $\gets$ $\bc.value$ \label{lin:delete28};
				
				\ElsIf{$(\texttt{read}($$\rc$.\textup{key}$)$ $=$ key$)$} \label{lin:delete29}
				\cmntwa{$node \langle obj\_id, key \rangle$ is part of red list}		
				
				\State $op\_status$ $\gets$ FAIL \label{lin:delete30};
				
				\State $\texttt{write}$($\rc$.max\_ts.lookup, TS($t\_id$)) \label{lin:delete31};
				\State $value$ $\gets$ NULL \label{lin:delete32};
				\Else \label{lin:delete33}
				\cmntwa{if node$ \langle obj\_id,key \rangle$ is neither in blue or red list add the node in red list and update timestamp}
				\State \lslins{\textcolor{black}{$RL$} $\downarrow$} \label{lin:delete34};
				\State $op\_status$ $\gets$ FAIL \label{lin:delete35};
				
				\State $\texttt{write}$(\textcolor{black}{$sh\_node$}.max\_ts.lookup, TS($t\_id$)) \label{lin:delete36};
				
				\State $value$ $\gets$ NULL \label{lin:delete37};
				\EndIf \label{lin:delete38}
				\cmntwa{release all the locks}
				\State{releasePred\&CurrLocks($\preds \downarrow$, $\currs \downarrow$)}\label{lin:delete433};
				\cmntwa{create new log entry in log}
				
				\State le $\gets$ new le$\langle obj\_id \downarrow, key\downarrow\rangle$\label{lin:delete39};
				\State \llsval{$NULL \downarrow$} \label{lin:delete40};
				
				\State \llspc{} \label{lin:delete41};
				
				\EndIf\label{lin:delete47}

				\State return $\langle value, op\_status\rangle$
				\EndFunction 
			\end{algorithmic}
		\end{multicols}	
	\end{algorithm}
	
	We prefer \npluk{} to be validated instantly and is never validated again in \nptc{} as the design choice to aid performance. Let's consider \otm{} history in \figref{nostm19}(i), if we would have validated $l(ht, k_1, v_0)$ again during $tryC$, $T_1$ would abort due to time order violation\cite{WeiVoss:2002:Morg}, but we can see that this history is acceptable where $T_1$ can be serialized before $T_2$ (\figref{nostm19}(ii)). Thus, \otm{} prevents such unnecessary aborts. Another advantage for this design choice is that $T_1$ doesn't have to wait for $tryC$ to know that the transaction is bound to abort as can be seen in \figref{nostm19}(iii). Here $l(ht, k_1, Abort)$ instantly aborts as soon as it realizes that time order is violated and schedule can no more be ensured to be correct saving significant computations of $T_1$.  This gain becomes significant if the application is lookup intensive where it would be inefficient to wait till \nptc{} to validate the \npluk{} only to know that transaction has to abort.

	\begin{figure}[H]
		\centering
		\captionsetup{justification=centering}
		{\scalebox{0.47}{\input{figs/nostm19.pdf_t}}}
		\caption{Advantages of encounter time lookup validation.}
		\label{fig:nostm19}
	\end{figure}

	 \noindent
	 \textbf{\npdel{}} (\algoref{delete}) in \rvp{} executes as similar to \rvmt{} and in \cp{} executes as \upmt{}. In \rvp{}, the \npdel{} first checks if their is already a previous \mth{} on same $key$ with the help of $txlog$. In case their is already a \mth{} that executed on same $key$, \npdel{} does not need to touch shared memory and sees the effect of the previous \mth{} and returns accordingly (\Lineref{delete3} to \Lineref{delete18}). For example if previous executed \mth{} is an insert then the current \npdel{} \mth will return $OK$ (\Lineref{delete5} to \Lineref{delete9}). If the previous executed \mth{} is an \npdel{} then the current \npdel{} should return $FAIL$ (\Lineref{delete10} to \Lineref{delete13}). In case previous \mth{} was \npluk{} then current \npdel{} returns the status same as that of the previous \npluk{} \mth{} also overwriting the log for the $value$ and $opn$.	
	\begin{algorithm}[H]
		\scriptsize
		\caption{\tabspace[0.2cm] STM\_delete($t\_id \downarrow, obj\_id \downarrow, key \downarrow, value \uparrow, op\_status \uparrow$ ) 
		}
		\label{algo:delete}
		\setlength{\multicolsep}{0pt}
		\begin{multicols}{2}
			\begin{algorithmic}[1]
				\makeatletter\setcounter{ALG@line}{68}\makeatother
				\Function{STM\_delete}{} \label{lin:delete1}
				
				\State STATUS $op\_status$ $\gets$ RETRY;
				\cmntwa{get the txlog of the current transaction by t\_id}
				\State txlog $\gets$ getTxLog($t\_id \downarrow$);
				\cmntwa{If le$\langle obj\_id, key \rangle$ already in log, update the le with the current operation}
				\If{$($\txlfind$)$} \label{lin:delete3}
				\State opn $\gets$ \llgopn{} \label{lin:delete4}; 
				\cmntwa{if previous local method is insert and current operation is delete then overall effect should be of delete, update log accordingly}
				\If{$($\textup{INSERT} $=$ \textup{opn}$)$} \label{lin:delete5}
				\State $value$ $\gets$ \llgval{} \label{lin:delete6};
				
				\State \llsval{$NULL \downarrow$} \label{lin:delete7};
				
				\State \llsopn{$DELETE \downarrow$} \label{lin:delete8};
				
				\State $op\_status$ $\gets$ OK \label{lin:delete9};
				\cmntwa{if previous local method is delete and current operation is delete then overall effect should be of delete, update log accordingly}
				\ElsIf{$($\textup{DELETE} $=$ \textup{opn}$)$} \label{lin:delete10}
				\State \llsval{$NULL \downarrow$} \label{lin:delete11};
				\State $value$ $\gets$ NULL \label{lin:delete12}; 
				\State $op\_status$ $\gets$ FAIL \label{lin:delete13}; 
				\Else \label{lin:delete14}
				\cmntwa{if previous local method is lookup and current operation is delete then overall effect should be of delete, update log accordingly}
				\State $value$ $\gets$ \llgval{} \label{lin:delete15}; 
				
				\State \llsval{$NULL \downarrow$} \label{lin:delete16};

				\State \llsopn{$DELETE \downarrow$} \label{lin:delete17};
				\State $op\_status$ $\gets$  $le.getOpStatus$($obj\_id \downarrow$, $key \downarrow$) \label{lin:delete18};
				
				\EndIf \label{lin:delete19}
				\Else \label{lin:delete20}
				\cmntwa{common function for \rvmt{}, if node corresponding to the key is not the part of underlying DS}
				\State \cld{}\label{lin:deletecld};
				\EndIf \label{lin:delete48}
				\cmntwa{update the local log}
				\State \llsopn{$DELETE \downarrow$} \label{lin:delete42};
				
				\State \llsopst{$op\_status \downarrow$} \label{lin:delete49};
				\State return $\langle value, op\_status\rangle$\label{lin:delete50};
				\EndFunction \label{lin:delete51}
			\end{algorithmic}
		\end{multicols}	
	\end{algorithm}	
	In case the current \npdel{} is not the first \mth{} on $key$ then it touches the shared memory to identify the correct location over the \tab{} from \Lineref{delete20} to \Lineref{delete49} (this refers to implementing the LR1 \& LR2). In order to do this \nplsls{} gives the correct location for the current \npdel{} to take effect over the \tab{} in form of $preds$ and $currs$ (\Lineref{delete21} in \algoref{commonLu&Del}) along with the validation status which reveals whether the \npdel{} will succeed or abort. If the $op\_status$ is Abort, the \mth{} simply aborts the transaction. Otherwise, \npdel{} updates the local log and the time stamps of the corresponding nodes in the \lsl{} of the \tab{} from line \Lineref{delete24} to \Lineref{delete49}. 
	
	From \Lineref{delete25} to \Lineref{delete28}, \npdel{} observes that the node to be deleted is reachable from $\bn$ i.e. it is \bc{} thus it updates it's time-stamp field and returns $op\_status$ to $OK$ with the value of \bc{} (the update corresponding to this case takes place in \nptc{} as represented in \figref{1nostm9}). From \Lineref{delete29} to \Lineref{delete32}, \npdel{} observes that the node to be deleted is reachable by $\rn$ i.e. it is \rc{} thus it updates its time-stamp field and sets $op\_status$ to $FAIL$ (as the node is dead node or marked for deletion) and value returned is $NULL$. Otherwise, in \Lineref{delete33} to \Lineref{delete37} the node is not at all present in \lsl{}. Thus first \npdel{} adds a node in \rn{} and updates its time-stamp and returns the $value$ as $NULL$ and sets the $op\_status$ as $FAIL$ (\figref{0nostm} and \figref{1nostm1} represents the case). \Lineref{delete40}, \Lineref{delete41} and \Lineref{delete42} sets the $value$, location and $opn$ in local log respectively. At \Lineref{delete433} the locks acquired(in invoked \nplsls{}) to update shared memory time-stamps are released in order.

	\begin{figure}[H]
		\centering
		\begin{minipage}[b]{0.49\textwidth}
			\scalebox{.47}{\input{figs/nostm.pdf_t}}
			\centering
			\captionsetup{justification=centering}
			\captionof{figure}{$k_{10}$ is not present in $\bn$ as well as $\rn$}
			\label{fig:0nostm}
		\end{minipage}
		\hfill
		\begin{minipage}[b]{0.49\textwidth}
			\centering
			\captionsetup{justification=centering}
			\scalebox{.47}{\input{figs/nostm1.pdf_t}}
			\captionof{figure}{Adding $k_{10}$ into $\rn$}
			\label{fig:1nostm1}
		\end{minipage}   
	\end{figure}
	
			 \noindent
			 \textbf{\npins{}} method in \rvp{} simply checks if their is a previous \mth{} that executed on the same $key$. If their is already a previous \mth{} that has executed within the same transaction it simply updates the new $value$, $opn$ as insert and $op\_status$ to $OK$ (\Lineref{insert5}, \Lineref{insert6} and \Lineref{insert7} respectively). In case the \npins{} is the first \mth{} on $key$ it creates a new log entry for the $ll\_list$ of $txlog$ at \Lineref{insert3}. Finally the \npins{} gets to modify the underlying \tab{} using \nplslins{} at the \cp{} in \nptc.

		
		\begin{algorithm}[H]
			\scriptsize
			\caption{\tabspace[0.2cm] STM\_insert ($t\_id \downarrow, obj\_id \downarrow, key\downarrow, value\downarrow, op\_status\uparrow$) : updates log entry and return op\_status locally.
			}
			\label{algo:insert}
			\setlength{\multicolsep}{0pt}
			\begin{multicols}{2}
				\begin{algorithmic}[1]
					\makeatletter\setcounter{ALG@line}{102}\makeatother
					\scriptsize
					\Function{STM\_insert}{} \label{lin:insert1}
					\State STATUS op\_status $\gets$ OK;
					\cmntwa{get the txlog of the current transaction by t\_id}
					\State txlog $\gets$ getTxLog($t\_id \downarrow$);
					\If{$(!$\txlfind$)$} \label{lin:insert2}
					\cmntwa{no le present for this $\langle obj\_id, key\rangle$, create one}
					
					\State le $\gets$ new le$\langle obj\_id \downarrow, key\downarrow\rangle$\label{lin:insert3}; 
					\EndIf \label{lin:insert4}
					
					\cmntwa{le present for $\langle obj\_id, key\rangle$, merely update the log}
					\State \llsval{$value \downarrow$} \label{lin:insert5}; //$\Phi_{lp}$
					\State \llsopn{$INSERT \downarrow$}
					\label{lin:insert6};
					
					\State \llsopst{$OK \downarrow$} \label{lin:insert7};
					\cmntwa{return op\_status to the transaction that invoked insert}
					\State return $\langle op\_status\rangle$;
					\EndFunction \label{lin:insert8}
					
				\end{algorithmic}
			\end{multicols}
		\end{algorithm}
		
		\cmnt{
			\begin{figure}[H]
				\centering
				\captionsetup{justification=centering}
				\includegraphics[width=5cm, height=4cm]{figs/insertfc.png}
				\caption{Flow of \npins{}}
				\label{fig:insertfc}
			\end{figure}
		}

	The \nptc{} in \algoref{trycommit} implements the \cp{}. When a transaction is ready to commit it executes \nptc{} and fetches its $txlog$ into $ll\_list$ (at \Lineref{tryc3}). Next, this list is sorted in increasing order of keys accessed by the transaction during its lifetime at \Lineref{tryc4}. This is done to ensure that locks are acquired in an order to ensure deadlock free execution. It may so happen that the $preds$ \& $currs$ recorded by the transaction may be obsolete, thus a need for recalculating the $preds$ \& $currs$ arises. This is done using \nplsls{} which recalculates the $preds$ \& $currs$ and decides the $op\_status$ (\Lineref{tryc5}-\Lineref{tryc12}). 
	
	Now, from \Lineref{tryc14} to \Lineref{tryc43} the shared memory data structure (underlying hash table) is changed. Each \upmt{} modifies the underlying \tab one by one. While the shared memory is updated the $preds$ \& $currs$ may get obsolete as explained in \figref{nostm16}. We handle this using \nppoval{} in \algoref{povalidation} invoked at \Lineref{tryc17}. The different cases for insert are handled at \Lineref{tryc18}- \Lineref{tryc32}. The case where the node to be inserted is already present (i.e. reachable by \bn) is handled in block from \Lineref{tryc19} to \Lineref{tryc23}. When key to be inserted is present in the \tab{} but corresponding node is marked (i.e. only reachable by the \bn), \Lineref{tryc24} to \Lineref{tryc27} insert it in \bn{} as well. The \Lineref{tryc28} to \Lineref{tryc32} take care of the case where node corresponding to the key to be inserted is not at all present in the \tab{}.
	
	When the \mth{} is \npdel{} and the node to be deleted is present in the \tab{} (i.e reachable by \bn{}), \Lineref{tryc33} to \Lineref{tryc41} set the node marked using \nplsldel{}. Finally, the acquired locks are released at \Lineref{tryc44} and the transaction status is returned.
	
	
	\begin{algorithm}[H]
		\scriptsize
		\caption{\tabspace[0.2cm] STM\_tryC($t\_id \downarrow, tx\_status \uparrow$) 
		}
		\label{algo:trycommit}
		\setlength{\multicolsep}{0pt}
		\begin{multicols}{2}
			\begin{algorithmic}[1]
				\makeatletter\setcounter{ALG@line}{117}\makeatother
				\Function{STM\_tryC}{} \label{lin:tryc1}
				
				\cmntwa{get the txlog of the current transaction by t\_id}	
				\State $ll\_list$ $\gets$ \txgllist \label{lin:tryc3};
				\cmntwa{sort the local log in increasing order of keys and copy into ordered list}
				\State $ordered\_ll\_list$ $\gets$ \llsort{} \label{lin:tryc4};
				\cmntwa{identify the new preds and currs for all update methods of a tx and validate it}
				\While{$(\textbf{$le_i \gets \textup{next}(ordered\_ll\_list$}))$} \label{lin:tryc5}
				\State ($key, obj\_id$) $\gets$ \llgkeyobj{} \label{lin:tryc6};
				\cmntwa{search correct location for the operation over lsl and lock the corresponding \preds and \currs}
				\State \lsls{$TRYC \downarrow$} \label{lin:tryc7};
				\cmntwa{if lslSearch return op\_status as ABORT then method will return ABORT}
				\If{$(op\_status$ $=$ \textup{ABORT}$)$} \label{lin:tryc8}
				\cmntwa{release local memory in case lslSearch returns abort}
				\State \handlea{} \label{lin:tryc9};
				\State return $\langle op\_status\rangle$;
				
				\EndIf \label{lin:tryc11}
				\cmntwa{modify the log entry to help upcoming update method of same tx}
				\State \llspc{} \label{lin:tryc12};
				\EndWhile \label{lin:tryc13}
				\cmntwa{get each update method one by one and take the effect in underlying DS}
				\While{$(\textbf{$le_i \gets \textup{next}(ordered\_ll\_list$}))$} \label{lin:tryc14}
				\State ($key, obj\_id$) $\gets$ \llgkeyobj{} \label{lin:tryc15};
				\cmntwa{get the operation name to local log entry}
				\State opn $\gets$ $le_i$.opn \label{lin:tryc16};
				\cmntwa{if operation is insert then after successful completion of it node corresponding to the key should be part of \bn}
				
				\cmntwa{modify the preds and currs for the consecutive update methods which are working on overlapping zone in lazyrb-list}
				\State intraTransValdation($le_i \downarrow, \preds \uparrow, \currs \uparrow$) \label{lin:tryc17};
				
				\If{$($\textup{INSERT} $=$ \textup{opn}$)$} \label{lin:tryc18}
				\cmntwa{if node corresponding to the key is part of \bn}
				\If{$\texttt{read}(\bc.\textup{key}) = key)$} \label{lin:tryc19}
				\cmntwa{get the value from local log}
				\State $value$ $\gets$ \llgval{} \label{lin:tryc20};
				\cmntwa{update the value into underlying DS}	
				\State $\texttt{write}$($\bc$.value, $value$) \label{lin:tryc21};
				\cmntwa{update the max\_ts of insert for node corresponding to the key into underlying DS}
				\State $\texttt{write}$($\bc$.max\_ts.insert, TS($t\_id$)) \label{lin:tryc23}; 
				\cmntwa{if node corresponding to the key is part of \rn}
				\ElsIf{$(\texttt{read}($$\rc$.\textup{key}$)$ $=$ key$)$} \label{lin:tryc24}
				\cmntwa{connect the node corresponding to the key to \bn as well}
				\State \lslins{\textcolor{black}{$RL$}$\_$\textcolor{black}{$BL$} $\downarrow$} \label{lin:tryc25};
				
				\cmntwa{update the max\_ts of insert for node corresponding to the key into underlying DS}
				\State $\texttt{write}$($\rc$.max\_ts.insert, TS($t\_id$)) \label{lin:tryc27};
				
				\Else \label{lin:tryc28}
				\cmntwa{if node corresponding to the key is not part of \bn as well as \rn then create the node with the help of lslIns() and add it into \bn}
				\State \lslins{\textcolor{black}{$BL$} $\downarrow$} \label{lin:tryc29};
				\cmntwa{update the max\_ts of insert for node corresponding to the key into underlying DS}
				\State $\texttt{write}$(node.max\_ts.insert, TS($t\_id$)) \label{lin:tryc31};
				\cmntwa{need to update the node field of log so that it can be released finally}
				\State $le_{i}$.node $\gets$ $\bp.\bn$ 
				
				\EndIf \label{lin:tryc32}

				\cmntwa{if operation is delete then after successful completion of it node corresponding to the key should not be part of \bn}
				\ElsIf{$($\textup{DELETE} $=$ \textup{opn}$)$} \label{lin:tryc33}
				\cmntwa{if node corresponding to the key is part of \bn}
				
				\If{$(\texttt{read}($$\bc$.\textup{key}$)$ $=$ key$)$} \label{lin:tryc34}
				\cmntwa{delete the node corresponding to the key from the \bn with the help of lslDel()}
				\State \lsldel{} \label{lin:tryc35};
				\cmntwa{update the max\_ts of delete for node corresponding to the key into underlying DS}
				\State $\texttt{write}$($\bc$.max\_ts.delete, TS($t\_id$)) \label{lin:tryc37};
				
				
				
				
				\EndIf \label{lin:tryc41}
				
				\EndIf \label{lin:tryc42}
				\EndWhile \label{lin:tryc43}
				\cmntwa{release all the locks}
				\State \rlsol{} \label{lin:tryc44};  
				\cmntwa{set the tx status as OK}
				\State $tx\_status$ $\gets$ OK \label{lin:tryc45};
				
				\State return $\langle tx\_status\rangle$\label{lin:tryc47};
				\EndFunction \label{lin:tryc48}
			\end{algorithmic}
		\end{multicols}
	\end{algorithm}

	\begin{algorithm}[H]
		\scriptsize
		\caption{\look($t\_id \downarrow, obj\_id \downarrow, key \downarrow, val\_type \downarrow, \preds \uparrow, \currs \uparrow, op\_status \uparrow$) : finds location (\preds \& \currs) for given $ \langle obj\_id, key \rangle$ and returns them in locked state else returns ABORT. 
		}\label{algo:lslsearch}
		\setlength{\multicolsep}{0pt}
		\begin{multicols}{2}
			\begin{algorithmic}[1]
				\makeatletter\setcounter{ALG@line}{184}\makeatother
				\Function{lslSearch}{} \label{lin:lslsearch1}
				
				\State STATUS $op\_status$ $\gets$ RETRY;
				\While{($op\_status$ = \textup{RETRY})} \label{lin:lslsearch2}
				\cmntwa{get the head of the bucket in hash-table}
				\State head $\gets$ \glslhead \label{lin:lslsearch3};
				\cmntwa{init $\bp$ to head}
				\State $\bp$ $\gets$ $\texttt{read}$($head$) \label{lin:lslsearch4}; 
				\cmntwa{init $\bc$ to $\bp.\bn$}
				\State $\bc$ $\gets$ $\texttt{read}$($\bp$.\bn) \label{lin:lslsearch5};
				\cmntwa{search node $ \langle obj\_id, key \rangle$ location in blue list}
				\While{$(\texttt{read}($$\bc$.\textup{key}$)$ $<$ $key)$} \label{lin:lslsearch6}
				\State $\bp$ $\gets$ $\bc$ \label{lin:lslsearch7};
				
				\State $\bc$ $\gets$ $\texttt{read}$($\bc$.\bn) \label{lin:lslsearch8};
				
				\EndWhile \label{lin:lslsearch9}
				\State    /*init $\rp$ to $\bp$*/
				\State $\rp$ $\gets$ $\bp$ \label{lin:lslsearch11};
				\State    /*init $\rc$ to $\bp.\rn$*/
				\State $\rc$ $\gets$ $\bp$.\rn \label{lin:lslsearch12};
				\State    /*search node $ \langle obj\_id, key \rangle$ location in red list between \bp \& \bc*/
				\While{$(\texttt{read}($$\rc$.\textup{key}$)$ $<$ $key)$} \label{lin:lslsearch13}
				
				\State $\rc$ $\gets$ $\rc$ \label{lin:lslsearch14};
				\State $\rc$ $\gets$ $\texttt{read}$($\rc$.\rn) \label{lin:lslsearch15};
				
				\EndWhile \label{lin:lslsearch16}
				\cmntwa{acquire the locks on increasing order of keys}
				\State acquirePred\&CurrLocks($ \preds \downarrow$, $ \currs \downarrow$)\label{lin:lslsearch17};
				
				
				\cmntwa{validate the location recorded in \preds \& \currs. Also verify if the transaction has to be aborted.}
				\State validation($t\_id \downarrow$, $key$ $\downarrow$, $\preds$ $\downarrow$, $\currs$ $\downarrow$, $val\_type$ $\downarrow$, $op\_status \uparrow$)\label{lin:lslsearch21};	
				\cmntwa{if validation returns op\_status as RETRY or ABORT then release all the locks}
				\If{(($op\_status$ = \textup{RETRY}) $\lor$ ($op\_status$ = \textup{ABORT}))} \label{lin:lslsearch22}
				\cmntwa{release all the locks}
				\State{releasePred\&CurrLocks($\preds \downarrow$, $\currs \downarrow$)}
				
				\EndIf \label{lin:lslsearch27}
				
				\EndWhile \label{lin:lslsearch28}
				
				\State return $\langle \preds, \currs, op\_status \rangle$ \label{lin:lslsearch29};
				
				\EndFunction \label{lin:lslsearch30}
			\end{algorithmic}
		\end{multicols}
	\end{algorithm}
	

	\begin{algorithm}[H]
		\scriptsize
		\caption{\tabspace[0.2cm] \ins($\preds \downarrow, \currs \downarrow, list\_type \downarrow$) : Inserts or overwrites a node in underlying hash table at location corresponding to $preds$ \& $currs$. 
		}
		\label{algo:lslins}
		\setlength{\multicolsep}{0pt}
		\begin{multicols}{2}
			\begin{algorithmic}[1]
				\makeatletter\setcounter{ALG@line}{219}\makeatother
				\Function{lslIns}{} \label{lin:lslins1}
				\cmntwa{inserting the node which is red list to bluelist}
				\If{$((list\_type)$ $=$ $($\textcolor{black}{$RL$}$\_$\textcolor{black}{$BL$}$))$} \label{lin:lslins2}
				
				\State $\texttt{write}$($\rc$.marked, false) \label{lin:lslins3}; 
				\State $\texttt{write}$($\rc$.\bn, $\bc$) \label{lin:lslins4};
				\State $\texttt{write}$($\bp$.\bn, $\rc$) \label{lin:lslins5};
				\cmntwa{inserting the node into red list only}
				\ElsIf{$((list\_type$) $=$ \textcolor{black}{$RL$}$)$} \label{lin:lslins6}
				\State node = Create new node()
				\label{lin:lslins7};
				\State $\texttt{write}$(node.marked, True) \label{lin:lslins8};
				\State $\texttt{write}$(node.\rn, $\rc$) \label{lin:lslins9};
				\State $\texttt{write}$($\rp$.\rn, node) \label{lin:lslins10};
				
				\Else \label{lin:lslins11}
				\cmntwa{inserting the node into red as well as blue list}
				\State node = new node() \label{lin:lslins12}; 
				\cmntwa{after creating the node acquiring the lock on it}
				\State node.lock();
				\State $\texttt{write}$(node.\rn, $\rc$) \label{lin:lslins13};
				\State $\texttt{write}$(node.\bn, $\bc$) \label{lin:lslins14};
				
				\State $\texttt{write}($$\rp$.\rn, node ) \label{lin:lslins15};
				
				\State $\texttt{write}$($\bp$.\bn, node) \label{lin:lslins16};
				\EndIf \label{lin:lslins17}
				\State return $\langle \rangle$;
				\EndFunction \label{lin:lslins18}
			\end{algorithmic}
		\end{multicols}
	\end{algorithm}

	\begin{figure}[H]
		\centering
		\captionsetup{justification=centering}
		\centerline{\scalebox{0.55}{\input{figs/nostm31.pdf_t}}}
		\caption{Execution of \nplslins{}: (i) key $k_5$ is present in $\rn$ and adding it into $\bn$, (ii) key $k_5$ is not present in $\rn$ as well as $\bn$ and adding it into $\rn$, (iii) key $k_5$ is not present in $\rn$ as well as $\bn$ and adding it into $\rn$ as well as $\bn$.}
		\label{fig:nostm31}
	\end{figure}
	\noindent
	\textbf{\nplslins{}} (\algoref{lslins}) adds a new node to the \lsl{} in the \tab{}. There can be following cases:
	\textbf{If node is present in \rn{} and has to be inserted to \bn:} such a case implies that the \nplslins{} is invoked in \cp{} for the corresponding \npins{} in local log represented by the block from \Lineref{lslins2} to \Lineref{lslins5}. Here we first reset the \rc{} mark field and update the $\bn$ to the \bc{} and \bp{} $\bn$ to \rc{}. Thus the node is now reachable by $\bn$ also. \figref{nostm31}(i) represents the case.
	\textbf{If node is meant to be inserted only in \rn:}
	This implies that the node is not present at all in the \lsl{} and is to be inserted for the first time. Such a case can be invoked from \rvmt{} of \rvp{}, if \rvmt{} is the first \mth{} of its transaction. \Lineref{lslins6} to \Lineref{lslins10} depict such a case where a new $node$ is created and its $marked$ field is set, depicting that its a dead node meant to be reachable only via \rn{}. In \Lineref{lslins9} and \Lineref{lslins10} the \rn{} field of the $node$ is updated to \rc{} and \rn{} field of the \rp{} is modified to point to the $node$ respectively. \figref{nostm31}(ii) represents the case.
	\textbf{If node is meant to be inserted in \bn:} In such a case it may happen that the node is already present in the \rn{} (already covered by \Lineref{lslins2} to \Lineref{lslins5}) or the node is not present at all. The later case is depicted in \Lineref{lslins11} to \Lineref{lslins16} which creates a new $node$ and add the node in both \rn{} and \bn{} note that order of insertion is important as the \lsl{} can be concurrently accessed by other transactions since traversal is lock free. \figref{nostm31}(iii) represents the case.
	
	
	\begin{algorithm}[H]
		\scriptsize
		\caption{\tabspace[0.2cm] \del($\preds \downarrow, \currs \downarrow$) : Deletes a node from blue link in underlying hash table at location corresponding to $preds$ \& $currs$.
		}
		\label{algo:lsldelete}
		\setlength{\multicolsep}{0pt}
		\begin{multicols}{2}
			\begin{algorithmic}[1]
				\makeatletter\setcounter{ALG@line}{243}\makeatother
				\Function{lslDel}{} \label{lin:lsldel1}
				
				\cmntwa{mark the node$ \langle obj\_id, key \rangle$ for deletion}
				\State $\texttt{write}$($\bc$.marked, True) \label{lin:lsldel2};
				\cmntwa{set the update the blue links}
				\State $\texttt{write}$($\bp$.\bn, $\bc$.\bn) \label{lin:lsldel3};
				\State return $\langle \rangle$;
				\EndFunction \label{lin:lsldel4}
			\end{algorithmic}
			
		\end{multicols}
	\end{algorithm}
	
	\noindent
	\textbf{\nplsldel{}} removes a node from \bn{}. It can be invoked from \cp{} for corresponding \npdel{} in $txlog$. It simply sets the marked field of the node to be deleted (\bc{}) and changes the \bn{} of \bp{} to \bc{} as shown in \Lineref{lsldel2} and \Lineref{lsldel3} of \algoref{lsldelete} respectively. \figref{1nostm9} shows the deletion of node corresponding to $k_5$.

	\begin{figure}[H]
		\centering
		\captionsetup{justification=centering}
		\centerline{\scalebox{0.55}{\input{figs/nostm9.pdf_t}}}
		\caption{Execution of \nplsldel{}: (i) \lsl{} before $k_5$ is deleted, (ii) \lsl{} after $k_5$ is deleted from $\bn$}
		\label{fig:1nostm9}
	\end{figure}
	
	\noindent
	\textbf{validation:}
	\rvmt{} and \upmt{} do the \emph{validation} in \rvp{} and \cp{} respectively. \emph{validation} invokes \npintv{} and then does the \nptov{} in the mentioned order. \npintv{} is the property of the method and \nptov{} is the property of the transaction. Thus validating the \mth{} before the transaction intuitively make sense.



	\begin{algorithm}[H]
		\scriptsize
		\caption{\tabspace[0.2cm] validation($t\_id \downarrow, key \downarrow, \preds \downarrow, \currs \downarrow, val\_type \downarrow, op\_status \uparrow$)
		}
		\label{algo:validation}
		\setlength{\multicolsep}{0pt}
		\begin{multicols}{2}
			\begin{algorithmic}[1]
				\makeatletter\setcounter{ALG@line}{250}\makeatother
				\Function{validation}{} \label{lin:validation1}
				\cmntwa{validate against concurrent updates}
				\State $op\_status$ $\gets$ methodValidation($\preds \downarrow$, $\currs \downarrow$)\label{lin:validation2};
				\cmntwa{on succesfull method validation validate of transactional ordering to ensure opacity}
				\If{$( RETRY \neq op\_status)$} \label{lin:validation3}
				
				\State $op\_status$ $\gets$ \toval{} \label{lin:validation4};
				\EndIf \label{lin:validation5}
				\State return $\langle op\_status\rangle$ \label{lin:validation6}; 
				\EndFunction \label{lin:validation7}
			\end{algorithmic}
			
		\end{multicols}
	\end{algorithm}
	
		\noindent
		In \textbf{\npintv{}} each transaction ensures that no other transaction has concurrently updated the same location in \lsl{} where it wants to perform the operation. This is done by checking that the \bp{} and \bc{} are not marked for deletion and next node of \bp{} and \rp{} is still the same as observed by lockfree traversal over the \lsl{}.
	
	\begin{algorithm}[H]
		\scriptsize
		\caption{\tabspace[0.2cm] methodValidation($\preds \downarrow, \currs \downarrow$) 
		}
		\label{algo:interferenceValidation}
		\begin{algorithmic}[1]
			\makeatletter\setcounter{ALG@line}{259}\makeatother
			\Function{methodValidation}{} \label{lin:iv1}
			\If{$(\texttt{read}(\bp.marked) || \texttt{read}(\bc.marked) ||\texttt{read}(\bp.\bn) \neq \bc || \texttt{read}(\rp.\rn) \neq {\rc})$} \label{lin:iv2}
			
			\State return $\langle RETRY\rangle$  \label{lin:iv3};
			
			\Else \label{lin:iv4}
			
			\State return $\langle OK\rangle$ \label{lin:iv5};
			
			\EndIf \label{lin:iv6}
			\EndFunction \label{lin:iv7}
		\end{algorithmic}
	\end{algorithm}
	
	\noindent
	In \textbf{\nptov{}} \rvmt{} always conflicts with the \upmt{} (as established in conflict notion \secref{conflicts}). If the node corresponding to the $key$ is present in the \lsl{} (\Lineref{tov5}) we compare with time-stamp of the transaction that last executed the conflicting \mth{} on same $key$. If the current \mth{} that invoked the \nptov{} is \rvmt{} then \Lineref{tov6} handles the case.
	Otherwise, if the invoking \mth{} is \upmt{} then \Lineref{tov9} handles the case. \figref{1nostm25} and \figref{1nostm26} show the execution of \nptov{}. Here $l_1(ht, k_1)$ will return $Abort$ in  \figref{1nostm26} because $d_2((ht, k_1)$ of $T_2$ has already updated the time-stamp at the node corresponding to $k_1$. So, when $l_1(ht, k_1)$ does its \nptov{} at \Lineref{tov9}, $TS(t_1)$ $<$ $curr.max\_ts.delete(k)$ holds true (since, $T_1$ $<$ $T_2$) leading to $abort$ of $T_1$ at \Lineref{tov11}. This gives us a equivalent sequential schedule which can be shown \coop{}. \figref{1nostm25} shows the schedule where no sequential schedule is possible if \nptov{} is not applied as there is no way to recognize the time-order violation.
	
	\begin{figure}[H]
		\centering
		\begin{minipage}[b]{0.49\textwidth}
			\scalebox{.47}{\input{figs/nostm25.pdf_t}}
			\centering
			\captionsetup{justification=centering}
			\captionof{figure}{Non opaque history. Without time-stamp validation in \nptov}
			\label{fig:1nostm25}
		\end{minipage}
		\hfill
		\begin{minipage}[b]{0.49\textwidth}
			\centering
			\captionsetup{justification=centering}
			\scalebox{.47}{\input{figs/nostm26.pdf_t}}
			\captionof{figure}{Opaque history H1. With time-stamp validation in \nptov}
			\label{fig:1nostm26}
		\end{minipage}   
	\end{figure}
	
	
	\begin{algorithm}[H]
		\scriptsize
		\caption{\tabspace[0.2cm] transValidation($t\_id \downarrow, key \downarrow, \currs \downarrow, val\_type \downarrow, op\_status \uparrow$) : Time-order validation for each transaction.
		}
		\label{algo:tovalidation}
		\begin{algorithmic}[1]
			\makeatletter\setcounter{ALG@line}{266}\makeatother
			\Function{transValidation}{} \label{lin:tov1}
			\cmntwa{by default setting the op\_status as RETRY}
			\State STATUS $op\_status$ $\gets$ OK \label{lin:tov3};
			\cmntwa{get the appropriate $sh\_curr$ (red or blue) correspondinjg to key}
			\State \llgaptc{$\currs \downarrow$, $key\downarrow$, $sh\_curr \uparrow$} \label{lin:tov4};
			\cmntwa{if $sh\_curr$ is not NULL and node corresponding to the key is equal to $sh\_curr$.key then check for TS}
			\If{$(($\textup{$sh\_curr$} $\neq$ \textup{NULL}$)\land(($\textup{$sh\_curr$.key}$)$ $=$ \textup{key}$))$} \label{lin:tov5}
			\cmntwa{if val\_type is RV then transaction validation for \rvmt{}}
			\If{$ ((val\_type = RV) \land ($\textup{TS}$(t\_id)$ $<$ $(\texttt{read}($\textup{$sh\_curr$.max\_ts.insert}$($\textup{k}$)))$ $||$ \\\hspace{2cm} $($\textup{TS}($t\_id$) $<$ $(\texttt{read}($\textup{$sh\_curr$.max\_ts.delete}$($\textup{k}$)))))$} \label{lin:tov6}
			\State $op\_status$ $\gets$ ABORT \label{lin:tov8};
			\cmntwa{transaction validation for \upmt{}}
			\ElsIf{$(($\textup{TS}$(t\_id)$ $<$ $(\texttt{read}($\textup{$sh\_curr$.max\_ts.insert}$($\textup{k}$)))$ $||$ \textup{TS}$(t\_id)$ $<$ $(\texttt{read}($\textup{$sh\_curr$.max\_ts.delete}$($\textup{k}$)))$ $||$\\ \hspace{2cm} \textup{TS}$(t\_id)$ $<$ $(\texttt{read}($\textup{$sh\_curr$.max\_ts.lookup}$($\textup{k}$))))$} \label{lin:tov9}
			
			\State $op\_status$ $\gets$ ABORT \label{lin:tov11};
			\EndIf \label{lin:tov12}
			
			\EndIf \label{lin:tov13}
			
			\State return $\langle op\_status\rangle$ \label{lin:tov14}; 
			
			\EndFunction \label{lin:tov15}
		\end{algorithmic}
	\end{algorithm}
	
	    \noindent
		\textbf{\nppoval{}} handles the case where two consecutive updates within same transaction having overlapping $preds$ and $currs$ may overwrite the previous \mth{} such that only effect of the later \mth{} is visible. This happens because the previous \mth{} while updating, changes the \lsl{} causing the $preds$ \& $currs$ of the next \mth{} working on the consecutive key to become obsolete. Thus, \nppoval{} corrects this by finding the new $preds$ and $currs$ of the current \mth{} on the consecutive key. There might be two cases (i) if previous \mth{} is \npins{} or (ii) previous \mth is \npdel{}. For case(i) we find the \bp{} (at \Lineref{threadv4} to \Lineref{threadv5} using previous log entry) and for case(ii) we find \bp{} using previous log entry's \bp{} (\Lineref{threadv7}) and finally find the new \rp{} and \rc{} between the new found \bp{} and \bc{} at \Lineref{threadv10} to \Lineref{threadv11}.

	
	\begin{algorithm}[H]
		\scriptsize
		\caption{\tabspace[0.2cm] intraTransValidation($le \downarrow, \preds \uparrow, \currs \uparrow$)
		}
		\label{algo:povalidation}
		\setlength{\multicolsep}{0pt}
		\begin{multicols}{2}
			\begin{algorithmic}[1]
				\makeatletter\setcounter{ALG@line}{285}\makeatother
				\Function{intraTransValidation}{} \label{lin:threadv1}
				\State $le.getAllPreds\&Currs(le$ $\downarrow$, $\preds$ $\uparrow$, $\currs$ $\uparrow$) \label{lin:threadv2};	
				\cmntwa{if $\bp$ is marked or $\bc$ is not reachable from $\bp.\bn$ then modify the next consecutive \upmt{} $\bp$ based on previous \upmt{}}
				\If{$((\texttt{read}($$\bp$.\textup{marked}$)) ||$ $(\texttt{read}($ $\bp$.\textup{\bn}$)$ != $\bc$$))$} \label{lin:threadv3}
				\cmntwa{find $k$ < $i$; such that $le_k$ contains previous update method on same bucket }
				\If{$(($$le_{k}$.\textup{opn}$)$ $=$ INSERT$)$} \label{lin:threadv4}
				
				\State $le_{i}$.$\bp$.$\texttt{unlock()}$ \label{lin:thredv4-5};
				\State $\bp$ $\gets$ ($le_{k}.\bp.\bn)$ \label{lin:threadv5};
				\State $le_{i}$.$\bp$.$\texttt{lock()}$ \label{lin:thredv5-5};		
				\Else \label{lin:threadv6}
				\cmntwa{\upmt{} method $\bp$ will be previous method $\bp$}
				\State $le_{i}$.$\bp$.$\texttt{unlock()}$ \label{lin:thredv6-5};
				\State $\bp$ $\gets$ ($le_{k}$.$\bp$) \label{lin:threadv7};
				
				\State $le_{i}$.$\bp$.$\texttt{lock()}$ \label{lin:thredv7-5};
				\EndIf \label{lin:threadv8}
				
				\EndIf \label{lin:threadv9}
				\cmntwa{if $\rc$ \& $\rp$ is modified by prev operation then update them also}
				\If{$(\texttt{read}($$\rp$.\textup{\rn}$)$ != $\rc$$)$} \label{lin:threadv10}
				\State $le_{i}$.$\rp$.$\texttt{unlock()}$
				\State $\rp$ $\gets$ ($le_{k}$.$\rp.\rn$)  \label{lin:threadv11}; 
				\State $le_{i}$.$\rp$.$\texttt{lock()}$
				\EndIf \label{lin:threadv12}
				\State return $\langle \preds, \currs\rangle$;
				\EndFunction \label{lin:threadv13}
			\end{algorithmic}
		\end{multicols}
	\end{algorithm}
	
	\noindent
	\textbf{\emph{findInLL()}} is an utility \mth{} that returns true to the \mth{} that has invoked it, if the calling \mth{} is not the first method of the transaction on the $key$. This is done by linearly traversing the log and finding an entry corresponding to the $key$. If the calling \mth{} is the first \mth{} of the transaction for the $key$ then \emph{findInLL()} return false as it would not find any entry in the log of the transaction corresponding to the $key$. Since we consider that their can be multiple objects (\tab{}) so we need to find unique $\langle obj\_id, key \rangle$ pair (refer \Lineref{findll5}).	
	
		While executing the \nptov{} the time-stamp field of the corresponding $node$ has to be updated. Such a node can be either the marked (dead or \rc) or the unmarked (live \bc).

	\begin{algorithm}[H]
		\scriptsize
		\caption{\tabspace[0.2cm] findInLL($t\_id \downarrow, obj\_id \downarrow, key \downarrow, le \uparrow$) : Checks whether any operation corresponding to $\left\langle obj\_id, key  \right\rangle$ is present in ll\_list.
		}
		\label{algo:findInLL}
		\begin{algorithmic}[1]
			\makeatletter\setcounter{ALG@line}{309}\makeatother
			\Function{findInLL}{} \label{lin:findll1}
			
			\State ll\_list $\gets$ \txgllist{} \label{lin:findll3}; 
			\cmntwa{every method first identify the node corresponding to the key into local log}
			\While{$(le_i \gets next$($ll\_list$$))$} \label{lin:findll4}
			\If{$((le_i.first = obj\_id) \& (le_i.first = key))$} \label{lin:findll5}
			
			\State return $\langle TRUE, le\rangle$  \label{lin:findll6};
			\EndIf \label{lin:findll7}
			\EndWhile \label{lin:findll8}
			\State return $\langle FALSE, le = NULL\rangle$ \label{lin:findll9};
			\EndFunction \label{lin:findll10}
		\end{algorithmic}
	\end{algorithm}
	
	\noindent
	\textbf{\emph{get\_aptcurr()}} in \algoref{getaptcurr} is the utility \mth{} which returns the appropriate $node$ corresponding to the $key$. 
	
	\begin{algorithm}[H]
		\scriptsize
		\caption{\tabspace[0.2cm] get\_aptcurr($\currs \downarrow$, $key \downarrow, sh\_curr \uparrow$) : Returns a curr node from underlying DS which corresponds to the key of $le_i$.
		}
		\label{algo:getaptcurr}
		\setlength{\multicolsep}{0pt}
		\begin{multicols}{2}
			\begin{algorithmic}[1]
				\makeatletter\setcounter{ALG@line}{319}\makeatother
				\Function{get\_aptcurr}{} \label{lin:aptcurr1}
				\cmntwa{by default set curr to NULL}
				\State sh\_curr $\gets$ NULL;
				\cmntwa{if node corresponding to the key is part of $\bn$ then curr is $\bc$}
				\If{$($$\bc$.\textup{key} = \textup{key}$)$} \label{lin:aptcurr2}
				\State sh\_curr $\gets$ $\bc$ \label{lin:aptcurr3};
				\cmntwa{if node corresponding to the key is part of $\rn$ then curr is $\rc$}
				\ElsIf{$($$\rc$.\textup{key} = \textup{key}$)$} \label{lin:aptcurr4}
				\State sh\_curr $\gets$ $\rc$ \label{lin:aptcurr5};
				
				\EndIf \label{lin:aptcurr6}
				
				\State return $\langle sh\_curr\rangle$ \label{lin:aptcurr7}; 
				
				\EndFunction \label{lin:aptcurr8}
			\end{algorithmic}
		\end{multicols}
	\end{algorithm}
	
	\noindent
	\textbf{\emph{release\_ordered\_locks()}} in \algoref{releaseorderedlocks} is an utility \mth{} to release the locks in order of the keys to avoid deadlock.
	
	
	\begin{algorithm}[H]
		\scriptsize
		\caption{\tabspace[0.2cm] release\_ordered\_locks($ordered\_ll\_list \downarrow$) : Release all locks taken during \nplsls{}.
		}
		\label{algo:releaseorderedlocks}
		\setlength{\multicolsep}{0pt}
		\begin{multicols}{2}
			\begin{algorithmic}[1]
				\makeatletter\setcounter{ALG@line}{331}\makeatother
				\Function{release\_ordered\_locks}{} \label{lin:rlock1}
				\cmntwa{releasing all the locks on preds, currs and node}
				\While{($\textbf{$le_i \gets next(ordered\_ll\_list$})$)} \label{lin:rlock2}
				
				\State $le_i$.$\bp$.$\texttt{unlock()}$ \label{lin:rlock3};//$\Phi_{lp}$ 
				\State $le_i$.$\rp$.$\texttt{unlock()}$ \label{lin:rlock4};
				\If{$le_i$.$node$}
				\State{$le_i$.$node$.$\texttt{unlock()}$}
				\EndIf
				\State $le_i$.$\rc$.$\texttt{unlock()}$ \label{lin:rlock5};
				\State $le_i$.$\bc$.$\texttt{unlock()}$ \label{lin:rlock6}; 
				\EndWhile \label{lin:rlock7}
				\State return $\langle \rangle$;
				\EndFunction \label{lin:rlock8}
			\end{algorithmic}
			
		\end{multicols}
	\end{algorithm}
	
	\noindent
	\textbf{\emph{acquirePred\&CurrLocks()}} in \algoref{acquirepreds&currs} \& \emph{releasePred\&CurrLocks} in \algoref{releasepreds&currs} do what their names denote. They are used as helping methods in \algoref{lslsearch}.
	
	\begin{algorithm}[H]
		\scriptsize
		\caption{\tabspace[0.2cm] acquirePred\&CurrLocks($ \preds \downarrow$, $ \currs \downarrow$) : acquire all locks taken during \nplsls{}.
		}
		\label{algo:acquirepreds&currs}
		\setlength{\multicolsep}{0pt}
		\begin{multicols}{2}
			\begin{algorithmic}[1]
				\makeatletter\setcounter{ALG@line}{344}\makeatother
				\Function{acquirePred\&CurrLocks}{}
				\State $\bp$.$\texttt{lock()}$;
				\State $\rp$.$\texttt{lock()}$;
				\State $\rc$.$\texttt{lock()}$;
				\State $\bc$.$\texttt{lock()}$;
				\State return $\langle \rangle$;        
				\EndFunction
			\end{algorithmic}
			
		\end{multicols}
	\end{algorithm}

	
	\begin{algorithm}[H]
		\scriptsize
		\caption{\tabspace[0.2cm] releasePred\&CurrLocks($ \preds \downarrow$, $ \currs \downarrow$) : Release all locks taken during \nplsls{}.
		}
		\label{algo:releasepreds&currs}
		\setlength{\multicolsep}{0pt}
		\begin{multicols}{2}
			\begin{algorithmic}[1]
				\makeatletter\setcounter{ALG@line}{351}\makeatother
				\Function{releasePred\&CurrLocks}{}
				\State $\bp$.$\texttt{unlock()}$\label{lin:rpandc};//$\Phi_{lp}$ 
				\State $\rp$.$\texttt{unlock()}$;
				\State $\rc$.$\texttt{unlock()}$;
				\State $\bc$.$\texttt{unlock()}$;
				\State return $\langle \rangle$;        
				\EndFunction
			\end{algorithmic}
			
		\end{multicols}
	\end{algorithm}
	
	\newpage
\section{Optimizations} 
\label{sec:optm}

In case a \npdel{} method returns FAIL then it would just behave as a \npluk{} because it does not modify the underlying data structure. Thus, we do not need to revalidate such failed \npdel{} method in \upmt{} phase inside \nptc{}. This helps in saving extra computation and time spent during \upmt{} phase leading to speedup of the transaction. 

Furthermore, twice validating the failed \npdel{} also may lead to unnecessary aborts as shown with an example in \figref{nostm32}. The \figref{nostm32}(i) shows the schedule where $T_1$ validates $del_{1}(k_1)$ two times. During \nptc{} it aborts realizing during its validation that $T_2$ has scheduled a conflicting insert operation on same node. On the other hand, if would not have validated this failed delete in \nptc{} the schedule can be accepted hence saving an unnecessary abort as shown in \figref{nostm32}(ii).

\begin{figure}[H]
	\centering
	\captionsetup{justification=centering}
	{\scalebox{0.5}{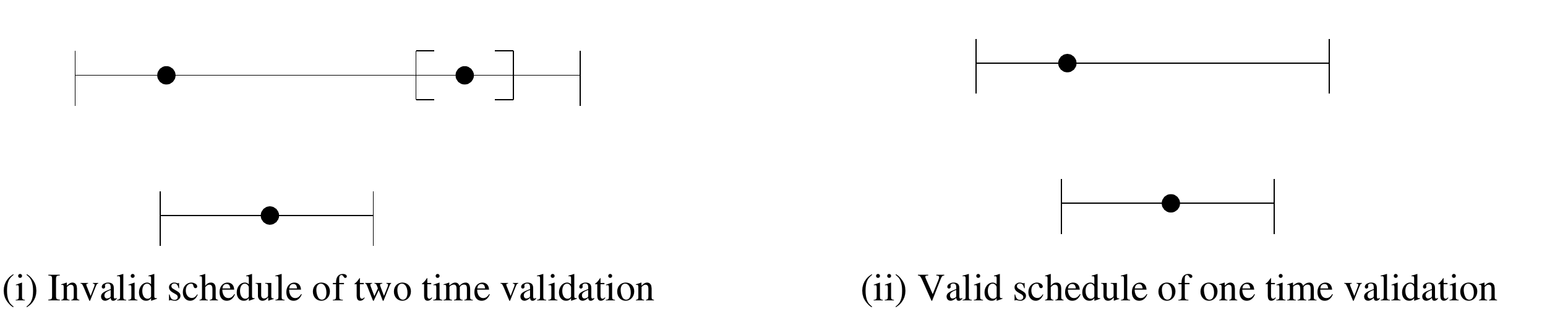}}
	\caption{Advantage of validating \npdel{} once, if its returning FAIL in \emph{\rvmt{} execution} phase}
	\label{fig:nostm32}
\end{figure}

Second optimization could be that during \nplsls{} if node corresponding to the node is part of the underlying data structure and the corresponding \npintv{} returns a retry (unsuccessful) then instead of retrying again we can do a \nptov{} so that in case the transaction is doomed to abort we would avoid unnecessary computation in retrying a transaction that is bound to abort.

\section{Proof Sketch of \otm{}}
\label{sec:pocapp}
\subsection{Method Level}
\label{sec:opnlevel}
For a global state, $S$, we denote $\evts{S}$ as all the events that has lead the system to global state $S$. We denote a state $S'$ to be in future of $S$ if $\evts{S} \subset \evts{S'}$. In this case, we denote $S \sqsubset S'$. We have the following definitions and lemmas:

\begin{definition}
	\label{def:glbls}
	\textit{PublicNodes:} Which is having a incoming $\rn$, except head node.
\end{definition}

\begin{definition}
	\label{def:Abs}
	\textit{Abstract List (Abs):} At any global abstract state S, $S.Abs$ can be defined as set of all public nodes that are accessible from head via red links union of set of all unmarked public nodes that are accessible from head via blue links. Formally, $\langle S.Abs = S.Abs.\rn \bigcup S.Abs.\bn \rangle$, where,\\  
	$S.Abs.\rn := \{\forall n| (n \in S.PublicNodes) \land (S.Head \rightarrow^*_{\rn} S.n) \}$. \\
	$S.Abs.\bn = \{\forall n| (n \in S.PublicNodes) \land (\neg S.n.marked) \land (S.Head \rightarrow^*_{\bn} S.n) \}$\\
	
\end{definition}

\begin{observation}
	\label{obs:node-forever}
	Consider a global state $S$ which has a node $n$. Then in any future state $S'$ of $S$, $n$ is a node in $S'$ as well. Formally, $\langle \forall S, S': (n \in S.nodes) \land (S \sqsubset S') \Rightarrow (n \in S'.nodes) \rangle$.
\end{observation}

\noindent With \obsref{node-forever} , we assume that nodes once created do not get deleted (ignoring garbage collection for now). 

\begin{observation}
	\label{obs:obs-nodekey}
	Consider a global state $S$ which has a node $n$, initialized with key $k$. Then in any future state $S'$ the key of $n$ does not change. Formally, $\langle \forall S, S': (n \in S.nodes) \land (S \sqsubset S') \Rightarrow (n \in S'.nodes) \land (S.n.key = S'.n.key) \rangle$. 
\end{observation}

\begin{observation}
	\label{obs:lslSearch}
	Consider a global state $S$ which is the post-state of return event of the function $\nplsls$ invoked in the $\npdel$ or $\nptc$ or $\npluk$ methods. Suppose the $\nplsls$ method returns (\textcolor{blue}{$preds[0]$}, \textcolor{red}{$preds[1]$}, \textcolor{red}{$currs[0]$}, \textcolor{blue}{$currs[1]$}). Then in the state $S$, we have, 
	\begin{enumerate}[label=\ref{obs:lslSearch}.\arabic*]
		\item \label{obs:lslSearch1}$ (\textcolor{blue}{preds[0]} \land \textcolor{red}{preds[1]} \land \textcolor{red}{currs[0]} \land \textcolor{blue}{currs[1]}) \in S.PublicNodes$ 
		\item \label{obs:lslSearch2}$(S.\textcolor{blue}{preds[0]}.locked)$ $\land$ $(S.\textcolor{red}{preds[1]}.locked)$ $\land$ $(S.\textcolor{red}{currs[0]}.locked)$ $\land$ $(S.\textcolor{blue}{currs[1]}.locked)$ 
		\item \label{obs:lslSearch3}$(\neg S.\textcolor{blue}{preds[0]}.marked) \land (\neg S.\textcolor{blue}{currs[1]}.marked)$ $\land$ $(S.\textcolor{blue}{preds[0]}.$\bn$ = S.\textcolor{blue}{currs[1]}$) $\land$
		$(S.\textcolor{red}{preds[1]}\\.$\rn$ = S.\textcolor{red}{currs[0]})$  
		
	\end{enumerate} 
\end{observation}

In \obsref{lslSearch}, $\nplsls$ method returns only if validation succeed at \Lineref{lslsearch21}.

\begin{lemma}
	\label{lem:lslSearch-ret}
	Consider a global state $S$ which is the post-state of return event of the function $\nplsls$ invoked in the $\npdel$ or $\nptc$ or $\npluk$ methods. Suppose the $\nplsls$ method returns (\textcolor{blue}{$preds[0]$}, \textcolor{red}{$preds[1]$}, \textcolor{red}{$currs[0]$}, \textcolor{blue}{$currs[1]$}). Then in the state $S$, we have, 
	\begin{enumerate}[label=\ref{lem:lslSearch-ret}.\arabic*]
		\item $((S.\textcolor{blue}{preds[0]}.key) < key \leq (S.\textcolor{blue}{currs[1]}.key))$. \label{lem:lslSearch-ret1}
		\item $((S.\textcolor{red}{preds[1]}.key) < key \leq (S.\textcolor{red}{currs[0]}.key))$. \label{lem:lslSearch-ret2} 
		
	\end{enumerate} 
\end{lemma}

\begin{proof}
	\begin{enumerate}[label=\ref{lem:lslSearch-ret}.\arabic*]
		
		\item \textbf{ ($S.\textcolor{blue}{preds[0]}.key < key \leq S.\textcolor{blue}{currs[1]}.key)$ :}\\
		Line \ref{lin:lslsearch4} of \nplsls{} method of \algoref{lslsearch} initializes $S.\textcolor{blue}{preds[0]}$ to point head node. Also, $(S.\textcolor{blue}{currs[1]}$ = $S.\textcolor{blue}{preds[0]}.\bn)$ by line \ref{lin:lslsearch5}. As in penultimate execution of line \ref{lin:lslsearch6} $(S.\textcolor{blue}{currs[1]}.key < key)$ and at line \ref{lin:lslsearch7} $(S.\textcolor{blue}{preds[0]} = S.\textcolor{blue}{currs[1]})$ this implies,
		\begin{equation}
		\label{eq:pbllesstthenkey}
		(S.\textcolor{blue}{preds[0]}.key < key)
		\end{equation}
		
		The node key doesn't change as known by \obsref{obs-nodekey}. So, before executing of line \ref{lin:lslsearch11}, we know that,
		\begin{equation}
		\label{eq:beforeline8}
		(key \leq S.\textcolor{blue}{currs[1]}.key)
		\end{equation}
		From \Eqref{pbllesstthenkey} and \Eqref{beforeline8}, we get,
		\begin{equation}
		\label{eq:bluenodeorder}
		(S.\textcolor{blue}{preds[0]}.key < key \leq S.\textcolor{blue}{currs[1]}.key)
		\end{equation}
		From \obsref{lslSearch2} and \obsref{lslSearch3} we know that these nodes are locked and from \obsref{obs-nodekey}, we have that key is not changed for a node, so the lemma holds even when \nplsls{} method of \algoref{lslsearch} returns.

		\item \textbf{($S.\textcolor{red}{preds[1]}.key < key \leq S.\textcolor{red}{currs[0]}.key)$ :}\\ 
		
		Line \ref{lin:lslsearch11} of \nplsls{} method of \algoref{lslsearch} initializes $S.\textcolor{red}{preds[1]}$ to point $S.\textcolor{blue}{preds[0]}$. Also, $(S.\textcolor{red}{currs[0]}$ = $S.\textcolor{blue}{preds[0]}.\rn)$ by line \ref{lin:lslsearch12}. As in penultimate execution of line \ref{lin:lslsearch13} $(S.\textcolor{red}{currs[0]}.key < key)$ and at line \ref{lin:lslsearch14} $(S.\textcolor{red}{preds[1]} = S.\textcolor{red}{currs[0]})$ this implies,
		\begin{equation}
		\label{eq:prllessthenkey}
		(S.\textcolor{red}{preds[1]}.key < key)
		\end{equation}
		
		The node key doesn't change as known by \obsref{obs-nodekey}. So, before executing of line \ref{lin:lslsearch17}, we know that 
		\begin{equation}
		\label{eq:beforeline14}
		(key \leq S.\textcolor{red}{currs[0]}.key)
		\end{equation}
		From \Eqref{prllessthenkey} and \Eqref{beforeline14}, we get,
		\begin{equation}
		\label{eq:rednodeorder}
		(S.\textcolor{red}{preds[1]}.key < key \leq S.\textcolor{red}{currs[0]}.key)
		\end{equation}
		From \obsref{lslSearch2} and \obsref{lslSearch3} we know that these nodes are locked and from \obsref{obs-nodekey}, we have that key is not changed for a node, so the lemma holds even when \nplsls{} method of \algoref{lslsearch} returns.
	\end{enumerate}
	
\end{proof}

\begin{lemma}
	\label{lem:key-change}
	For a node $n$ in any global state $S$, we have that,$\langle \forall n \in S.nodes: (S.n.key < S.n.\rn.key) \rangle$.
\end{lemma}
\begin{proof}
	We prove by Induction on events that change the $\rn$ field of the node (as these affect reachability), which are Line \ref{lin:lslins9}, \ref{lin:lslins10}, \ref{lin:lslins13} \& \ref{lin:lslins15} of \nplslins{} method of \algoref{lslins} . It can be seen by observing the code that \nplsldel{} method of \algoref{lsldelete} do not have any update events of $\rn$.\\
	\textbf{\texttt{Base condition:}} Initially, before the first event that changes the $\rn$ field, we know the underlying \lsl{} has immutable $S.head$ and $S.tail$ nodes with $(S.head.\bn = S.tail)$ and $(S.head.\rn = S.tail)$. The relation between their keys is $(S.head.key < S.tail.key)$ $\land$ $(head, tail) \in S.nodes$.\\
	\textbf{\texttt{Induction Hypothesis:}} Say, upto k events that change the $\rn$ field of any node, $(\forall n \in S.nodes:$ $S.n.key$ $<$ $S.n.\rn.key)$.\\[0.1cm]
	\textbf{\texttt{Induction Step:}} 
	So, as seen from the code, the $(k+1)^{th}$ event which can change the $\rn$ field be only one of the following:
	
	\begin{enumerate}
		\item \textbf{\texttt{Line \ref{lin:lslins9} of \nplslins{} method:}} By observing the code, we notice that Line \ref{lin:lslins9} ($\rn{}$ field changing event) can be executed only after the \nplsls{} method of \algoref{lslsearch} returns. Line \ref{lin:lslins7} of the \nplslins{} method creates a new node, $node$ with $key$ and at line \ref{lin:lslins8} set the $(S.node.marked$ = $true)$ (because inserting the node only into the redlink). Line \ref{lin:lslins9} then sets $(S.node.\rn$ = $S.\textcolor{red}{currs[0]})$. Since this event doest not change the $\rn{}$ field of any node reachable from the head of the list (because $node \notin S.PublicNodes$), the lemma is not violated.
		
		\item \textbf{\texttt{Line \ref{lin:lslins10} of \nplslins{} method:}} By observing the code, we notice that Line \ref{lin:lslins10} ($\rn{}$ field changing event) can be executed only after the \nplsls{} method of \algoref{lslsearch} returns. From \lemref{lslSearch-ret2}, we know that when \nplsls{} method of \algoref{lslsearch} returns then,
		\begin{equation}
		\label{eq:line81}
		(S.\textcolor{red}{preds[1]}.key) < key \leq (S.\textcolor{red}{currs[0]}.key)
		\end{equation}
		To reach line \ref{lin:lslins10} of \nplslins{} method, line \ref{lin:delete33} of \cldd{} method of \algoref{commonLu&Del} should ensure that,
		\begin{equation}
		\label{eq:line82}
		(S.\textcolor{red}{currs[0]}.key \neq key) \xRightarrow[]{\Eqref{line81}} (S.\textcolor{red}{preds[1]}.key) < key < (S.\textcolor{red}{currs[0]}.key)
		\end{equation}
		
		From \obsref{lslSearch3}, we know that,
		\begin{equation}
		\label{eq:line84}
		(S.\textcolor{red}{preds[1]}.\rn = S.\textcolor{red}{currs[0]})
		\end{equation}
		Also, the atomic event at line \ref{lin:lslins10} of \nplslins{} sets, 
		\begin{equation}
		\begin{split}
		\label{eq:line85}
		(S.\textcolor{red}{preds[1]}.\rn = node) \xRightarrow[]{\Eqref{line82}} (S.\rp.key < node.key)\\ \Longrightarrow (S.\textcolor{red}{preds[1]}.key < S.\textcolor{red}{preds[1]}.\rn.key)
		\end{split}
		\end{equation}
		
		Where $(S.node.key = key)$. Since $(\textcolor{red}{preds[1]}, node) \in S.nodes$ and hence, $(S.\textcolor{red}{preds[1]}.key < S.\textcolor{red}{preds[1]}.\rn.key)$.
		
		\item \textbf{\texttt{Line \ref{lin:lslins13} of \nplslins{} method:}} By observing the code, we notice that Line \ref{lin:lslins13} ($\rn{}$ field changing event) can be executed only after the \nplsls{} method of \algoref{lslsearch} returns. Line \ref{lin:lslins12} of the \nplslins{} method creates a new node, $node$ with $key$. Line \ref{lin:lslins13} then sets $(S.node.\rn$ = $S.\textcolor{red}{currs[0]})$. Since this event doest not change the $\rn{}$ field of any node reachable from the head of the list (because $node \notin S.PublicNodes$), the lemma is not violated.
		
		\item \textbf{\texttt{Line \ref{lin:lslins15} of \nplslins{} method:}} By observing the code, we notice that Line \ref{lin:lslins15} ($\rn{}$ field changing event) can be executed only after the \nplsls{} \algoref{lslsearch} method returns. From \lemref{lslSearch-ret2}, we know that when \nplsls{} method of \algoref{lslsearch} returns then,
		\begin{equation}
		\label{eq:line131}
		(S.\textcolor{red}{preds[1]}.key) < key \leq (S.\textcolor{red}{currs[0]}.key)
		\end{equation}
		To reach line \ref{lin:lslins15} of \nplslins{} method, line \ref{lin:tryc28} of \nptc{} method of \algoref{trycommit} should ensure that,
		\begin{equation}
		\label{eq:line132}
		(S.\textcolor{red}{currs[0]}.key \neq key) \xRightarrow[]{\Eqref{line131}} (S.\textcolor{red}{preds[1]}.key) < key < (S.\textcolor{red}{currs[0]}.key)
		\end{equation}
		
		From \obsref{lslSearch3}, we know that,
		\begin{equation}
		\label{eq:line134}
		(S.\textcolor{red}{preds[1]}.\rn = S.\textcolor{red}{currs[0]})
		\end{equation}
		Also, the atomic event at line \ref{lin:lslins15} of \nplslins{} sets,
		\begin{equation}
		\begin{split}
		\label{eq:line135}
		(S.\textcolor{red}{preds[1]}.\rn = node) \xRightarrow[]{\Eqref{line132}} (S.\rp.key < node.key)\\ \Longrightarrow (S.\textcolor{red}{preds[1]}.key < S.\textcolor{red}{preds[1]}.\rn.key)
		\end{split}
		\end{equation}
		
		where $(S.node.key = key)$. Since $(\textcolor{red}{preds[1]}, node) \in S.nodes$ and hence, $(S.\textcolor{red}{preds[1]}.key < S.\textcolor{red}{preds[1]}.\rn.key)$.
	\end{enumerate}
\end{proof}

\begin{lemma}
	\label{lem:reach}
	In a global state $S$, any public node $n$ is reachable from $Head$ via red links. Formally, $\langle \forall S, n: n \in S.PublicNodes  \implies S.Head \rightarrow^*_{\rn} S.n \rangle$. 
\end{lemma}
\begin{proof}
	We prove by Induction on events that change the $\rn$ field of the node (as these affect reachability), which are Line \ref{lin:lslins9}, \ref{lin:lslins10}, \ref{lin:lslins13} \& \ref{lin:lslins15} of \nplslins{} method of \algoref{lslins} . It can be seen by observing the code that \nplsldel{} method of \algoref{lsldelete} do not have any update events of $\rn$.\\
	\textbf{\texttt{Base condition:}} Initially, before the first event that changes the $\rn$ field of any node, we know that $(head, tail)$ $\in$ $S.PublicNodes$ $\land$ $\neg$($S.head.marked$) $\land$ $\neg$($S.tail.marked$) $\land$ $(S.head$ $\rightarrow^*_{\rn}$ $S.tail)$.\\
	\textbf{\texttt{Induction Hypothesis:}} Say, upto k events that change the next field of any node, $(\forall n \in S.PublicNodes$, $(S.head$ $\rightarrow^*_{\rn}$ $S.n))$. \\[0.1cm]
	\textbf{\texttt{Induction Step:}} 
	So, as seen from the code, the $(k+1)^{th}$ event which can change the $\rn$ field be only one of the following:\\
	\begin{enumerate}
		\item \textbf{\texttt{Line \ref{lin:lslins9} of \nplslins{} method:}} Line \ref{lin:lslins7} of the \nplslins{} method creates a new node, $node$ with $key$ and at line \ref{lin:lslins8} set the $(S.node.marked$ = $true)$ (because inserting the node only into the redlink).  Line \ref{lin:lslins9} then sets $(S.node.\rn$ = $S.\textcolor{red}{currs[0]})$. Since this event doest not change the $\rn{}$ field of any node reachable from the head of the list (because $node \notin S.PublicNodes$), the lemma is not violated.
		
		\item \textbf{\texttt{Line \ref{lin:lslins10} of \nplslins{} method:}} By observing the code, we notice that Line \ref{lin:lslins10} ($\rn{}$ field changing event) can be executed only after the \nplsls{} method of \algoref{lslsearch} returns. From line \ref{lin:lslins9} \& \ref{lin:lslins10} of \nplslins{} method, $(S.node.\rn =S.\rc) \land (S.\rp.\rn = S.node) \land (node \in S.PublicNodes) \land (S.node.marked = true)$ (because inserting the node only into the redlink). It is to be noted that (from \obsref{lslSearch2}), $(\bp, \rp,\\ \rc, \bc)$ are locked, hence no other thread can change marked field of $S.\rp$ and $S.\rc$ simultaneously. Also, from \obsref{obs-nodekey}, a node's key field does not change after initialization. Before executing line \ref{lin:lslins10}, $\rp$ is reachable from head by $\rn{}$ (from induction hypothesis). After line \ref{lin:lslins10}, we know that from $\rp$, public marked node, $node$ is also reachable. Thus, we know that $node$ is also reachable from head. Formally, $(S.Head \rightarrow^*_{\rn} S.\rp) \land (S.\rp \rightarrow^*_{\rn} S.node) \Rightarrow (S.Head \rightarrow^*_{\rn} S.node)$.
		
		\item \textbf{\texttt{Line \ref{lin:lslins13} of \nplslins{} method:}} Line \ref{lin:lslins12} of the \nplslins{} method creates a new node, $node$ with $key$. Line \ref{lin:lslins13} then sets $(S.node.\rn$ = $S.\textcolor{red}{currs[0]})$. Since this event doest not change the $\rn{}$ field of any node reachable from the head of the list (because $node \notin S.PublicNodes$), the lemma is not violated.
		
		\item \textbf{\texttt{Line \ref{lin:lslins15} of \nplslins{} method:}} By observing the code, we notice that Line \ref{lin:lslins15} ($\rn{}$ field changing event) can be executed only after the \nplsls{} method of \algoref{lslsearch} returns. From line \ref{lin:lslins13} \& \ref{lin:lslins15} of \nplslins{} method, $(S.node.\rn =S.\rc) \land (S.\rp.\rn = S.node) \land (node \in S.PublicNodes) \land (node.marked = false)$ (because new node is created by default with unmarked field). It is to be noted that (from \obsref{lslSearch2}), $(\bp, \rp,\\ \rc, \bc)$ are locked, hence no other thread can change marked field of $S.\rp$ and $S.\rc$ simultaneously. Also, from \obsref{obs-nodekey}, a node's key field does not change after initialization. Before executing line \ref{lin:lslins15}, $\rp$ is reachable from head by $\rn{}$ (from induction hypothesis). After line \ref{lin:lslins15}, we know that from $\rp$, public unmarked node, $node$ is also reachable. Thus, we know that $node$ is also reachable from head. Formally, $(S.Head \rightarrow^*_{\rn} S.\rp) \land (S.\rp \rightarrow^*_{\rn} S.node) \Rightarrow (S.Head \rightarrow^*_{\rn} S.node)$.
		
	\end{enumerate}
\end{proof}

\begin{corollary}
	\label{cor:uniquenodeKey}
	Each node is associated with an unique key, i.e. at any given state S, their cannot be two nodes with same key.
\end{corollary}
As every node is reachable by redlinks and has a strict ordering and from \obsref{node-forever} and \obsref{obs-nodekey} we get this.

\begin{corollary}
	\label{cor:key-abs}
	Consider the global state $S$ such that for any public node $n$, if there exists a key strictly greater than n.key and strictly smaller than $n.\rn.key$, then the node corresponding to the key does not belong to S.Abs. Formally, $\langle \forall S, n, key$ : $S.PublicNodes$ $\land$ $(S.n.key < key < S.n.\rn.key)$ $\implies$ $node(key)$ $\notin S.Abs \rangle$. 
\end{corollary}

\begin{observation}
	\label{obs:futurestate}
	Consider a global state $S$ which has a node $n$ is reachable from $head$ via $\rn$. Then in any future state $S'$ of $S$, node $n$ is also reachable from $head$ via $\rn$ in $S'$ as well. Formally, $\langle \forall S, S': (n \in S.nodes) \land (S \sqsubset S') \land\ (S.head \rightarrow^*_{\rn} S.n) \Rightarrow (n \in S'.nodes) \land\ (S'.head \rightarrow^*_{\rn} S'.n) \rangle$.
\end{observation}

\begin{proof}
	From \Obsref{node-forever}, we have that for any node n,  $n$ $\in$ $S.nodes \Rightarrow n \in S'.nodes$. Also, we have that in absence of garbage collection no node is deleted from memory and the redlinks are preserved during delete update events (refer \nplsldel{} method of \algoref{lsldelete}). 
\end{proof}

\begin{lemma}
	\label{lem:key-change-bl}
	For a node $n$ in any global state $S$, we have that,$\langle \forall n \in S.nodes :  (S.n.key < S.n.\bn.key) \rangle$.
\end{lemma}
\begin{proof}
	We prove by Induction on events that change the $\bn$ field of the node (as these affect reachability), which are Line \ref{lin:lslins4}, \ref{lin:lslins5}, \ref{lin:lslins14} \& \ref{lin:lslins16} of \nplslins{} method of \algoref{lslins} and Line \ref{lin:lsldel3} of \nplsldel{} method of \algoref{lsldelete} .\\
	\textbf{\texttt{Base condition:}} Initially, before the first event that changes the $\bn$ field, we know the underlying \lsl{} has immutable $S.head$ and $S.tail$ nodes with $(S.head.\bn = S.tail)$ and $(S.head.\rn = S.tail)$. The relation between their keys is $(S.head.key < S.tail.key)$ $\land$ $(head, tail) \in S.nodes$.\\ 
	\textbf{\texttt{Induction Hypothesis:}} Say, upto k events that change the $\bn{}$ field of any node, $(\forall n \in S.nodes : (S.n.key$ $<$ $S.n.\bn.key))$.\\
	\textbf{\texttt{Induction Step:}} 
	So, as seen from the code, the $(k+1)^{th}$ event which can change the $\bn$ field be only one of the following:
	\begin{enumerate}
		\item \textbf{\texttt{Line \ref{lin:lslins4} \& \ref{lin:lslins5} of \nplslins{} method:}} By observing the code, we notice that Line \ref{lin:lslins4} \& \ref{lin:lslins5} ($\bn{}$ field changing event) can be executed only after the \nplsls{} method of \algoref{lslsearch} returns. From \lemref{lslSearch-ret1} and \lemref{lslSearch-ret2}, we know that when \nplsls{} method of \algoref{lslsearch} returns then,
		\begin{equation}
		\begin{split}
		\label{eq:bline21}
		((S.\textcolor{blue}{preds[0]}.key) < key \leq (S.\textcolor{blue}{currs[1]}.key)) \land ((S.\textcolor{red}{preds[1]}.key) < key \leq (S.\textcolor{red}{currs[0]}.key)) 
		\end{split}
		\end{equation}
		To reach line \ref{lin:lslins4} of \nplslins{} method, line \ref{lin:tryc24} of \nptc{} method of \algoref{trycommit} should ensure that,
		\begin{equation}
		\begin{split}
		\label{eq:bline22}
		(S.\textcolor{blue}{currs[1]}.key \neq key) \land (S.\textcolor{red}{currs[0]}.key = key) \xRightarrow[]{\Eqref{bline21}}\\ ((S.\textcolor{blue}{preds[0]}.key) < key < (S.\textcolor{blue}{currs[1]}.key))\\ \land ((S.\textcolor{red}{preds[1]}.key) < (key = S.\textcolor{red}{currs[0]}.key))
		\end{split}
		\end{equation}
		
		From \obsref{lslSearch3}, we know that,
		\begin{equation}
		\label{eq:bline26}
		(S.\textcolor{blue}{preds[0]}.\bn = S.\textcolor{blue}{currs[1]}) \land (S.\textcolor{red}{preds[1]}.\rn = S.\textcolor{red}{currs[0]})
		\end{equation}
		The atomic event at line \ref{lin:lslins4} of \nplslins{} sets,
		\begin{equation}
		\begin{split}
		\label{eq:bline27}
		(S.\textcolor{red}{currs[0]}.\bn = S.\textcolor{blue}{currs[1]}) \xRightarrow[\lemref{key-change}]{\Eqref{bline22}, \lemref{reach}} (S.\textcolor{red}{currs[0]}.key) < (S.\textcolor{blue}{currs[1]}.key) \Longrightarrow\\ (S.\textcolor{red}{currs[0]}.key) < (S.\textcolor{red}{currs[0]}.\bn.key)
		\end{split}
		\end{equation}
		
		Also, the atomic event at line \ref{lin:lslins5} of \nplslins{} sets,
		\begin{equation}
		\begin{split}
		\label{eq:bline3ins7}
		(S.\textcolor{blue}{preds[0]}.\bn = S.\textcolor{red}{currs[0]}) \xRightarrow[]{\Eqref{bline22}} (S.\textcolor{blue}{preds[0]}.key) < (S.\textcolor{red}{currs[0]}.key) \Longrightarrow \\ (S.\textcolor{blue}{preds[0]}.key) < (S.\textcolor{blue}{preds[0]}.\bn.key).
		\end{split}
		\end{equation}
		
		Where $(S.\textcolor{red}{currs[0]}.key = key)$. Since $(\textcolor{blue}{preds[0]}, \rc) \in S.nodes$ and hence, $(S.\textcolor{blue}{preds[0]}.\\key < S.\textcolor{blue}{preds[0]}.\bn.key)$.
		
		\item \textbf{\texttt{Line \ref{lin:lslins14} of \nplslins{} method:}} By observing the code, we notice that Line \ref{lin:lslins14} ($\bn{}$ field changing event) can be executed only after the \nplsls{} method of \algoref{lslsearch} returns. Line \ref{lin:lslins12} of the \nplslins{} method creates a new node, $node$ with $key$. Line \ref{lin:lslins14} then sets $(S.node.\bn$ = $S.\textcolor{blue}{currs[1]})$. Since this event doest not change the $\bn{}$ field of any node reachable from the head of the list (because $node \notin S.PublicNodes$), the lemma is not violated.
		
		\item \textbf{\texttt{Line \ref{lin:lslins16} of \nplslins{} method:}} By observing the code, we notice that Line \ref{lin:lslins16} ($\bn{}$ field changing event) can be executed only after the \nplsls{} method of \algoref{lslsearch} returns. From \lemref{lslSearch-ret1} and \lemref{lslSearch-ret2}, we know that when \nplsls{} method of \algoref{lslsearch} returns then,
		\begin{equation}
		\begin{split}
		\label{eq:bline14ins1}
		(S.\textcolor{blue}{preds[0]}.key) < key \leq (S.\textcolor{blue}{currs[1]}.key) \land (S.\textcolor{red}{preds[1]}.key) < key \leq (S.\textcolor{red}{currs[0]}.key)
		\end{split}
		\end{equation}
		To reach line \ref{lin:lslins16} of \nplslins{} method, line \ref{lin:tryc28} of \nptc{} method of \algoref{trycommit} should ensure that,
		\begin{equation}
		\begin{split}
		\label{eq:bline14ins2}
		(S.\textcolor{red}{currs[0]}.key \neq key) \land (S.\textcolor{blue}{currs[1]}.key \neq key) \xRightarrow[]{\Eqref{bline14ins1}}\\ (S.\textcolor{blue}{preds[0]}.key) < key < (S.\textcolor{blue}{currs[1]}.key)\\ \land (S.\textcolor{red}{preds[1]}.key) < key < (S.\textcolor{red}{currs[0]}.key)
		\end{split}
		\end{equation}
		
		From \obsref{lslSearch3}, we know that,
		\begin{equation}
		\label{eq:bline14ins4}
		(S.\textcolor{blue}{preds[0]}.\bn = S.\textcolor{blue}{currs[1]})
		\end{equation}
		Also, the atomic event at line \ref{lin:lslins16} of \nplslins{} sets,
		\begin{equation}
		\begin{split}
		\label{eq:bline14ins5}
		(S.\textcolor{blue}{preds[0]}.\bn = S.node) \xRightarrow[]{\Eqref{bline14ins2}} (S.\textcolor{blue}{preds[0]}.key < S.node.key)\\ \Longrightarrow (S.\textcolor{blue}{preds[0]}.key < S.\textcolor{blue}{preds[0]}.\bn.key)
		\end{split}
		\end{equation}
		
		Where $(S.node.key = key)$. Since $(\textcolor{blue}{preds[0]}, node) \in S.nodes$ and hence, $(S.\textcolor{blue}{preds[0]}.key < S.\textcolor{blue}{preds[0]}.\bn.key)$.
		
		\item \textbf{\texttt{Line \ref{lin:lsldel3} of \nplsldel{} method:}} By observing the code, we notice that Line \ref{lin:lsldel3} ($\bn{}$ field changing event) can be executed only after the \nplsls{} method of \algoref{lslsearch} returns. From \lemref{lslSearch-ret1}, we know that when \nplsls{} method of \algoref{lslsearch} returns then,
		\begin{equation}
		\label{eq:bline2del1}
		(S.\textcolor{blue}{preds[0]}.key) < key \leq (S.\textcolor{blue}{currs[1]}.key)
		\end{equation}
		To reach line \ref{lin:lsldel3} of $\nplsldel$ method, line \ref{lin:tryc34} of $\nptc$ method of \algoref{trycommit} should ensure that,
		\begin{equation}
		\begin{split}
		\label{eq:bline2del2}
		(S.\textcolor{blue}{currs[1]}.key = key) \xRightarrow[]{\Eqref{bline2del1}} (S.\textcolor{blue}{preds[0]}.key) < (key = S.\textcolor{blue}{currs[1]}.key)
		\end{split}
		\end{equation}
		
		From \obsref{lslSearch3}, we know that,
		\begin{equation}
		\label{eq:bline2del4}
		(S.\textcolor{blue}{preds[0]}.\bn = S.\textcolor{blue}{currs[1]})
		\end{equation}
		We know from Induction hypothesis,
		\begin{equation}
		\label{eq:bline2del6}
		(\textcolor{blue}{currs[1]}.key < \textcolor{blue}{currs[1]}.\bn.key)
		\end{equation}
		Also, the atomic event at line \ref{lin:lsldel3} of \nplsldel{} sets,
		\begin{equation}
		\begin{split}
		\label{eq:bline2del5}
		(S.\textcolor{blue}{preds[0]}.\bn = S.\textcolor{blue}{currs[1]}.\bn) \xRightarrow[]{\Eqref{bline2del2}, \Eqref{bline2del6}} (S.\textcolor{blue}{preds[0]}.key < S.\textcolor{blue}{currs[1]}.\bn.key)\\ \Longrightarrow (S.\textcolor{blue}{preds[0]}.key < S.\textcolor{blue}{preds[0]}.\bn.key)
		\end{split}
		\end{equation}
		
		Where $(S.\textcolor{blue}{currs[1]}.key = key)$. Since $(\textcolor{blue}{preds[0]}, \textcolor{blue}{currs[1]}) \in S.nodes$ and hence, $(S.\textcolor{blue}{preds[0]}.key < S.\textcolor{blue}{preds[0]}.\bn.key)$ 
	\end{enumerate}
\end{proof}

\begin{lemma}
	\label{lem:reach-bl}
	In a global state $S$, any unmarked public node $n$ is reachable from $Head$ via blue links. Formally, $\langle \forall S, n: (S.PublicNodes) \land (\neg S.n.marked)  \implies (S.Head \rightarrow^*_{\bn} S.n) \rangle$. 
\end{lemma}
\begin{proof}
	We prove by Induction on events that change the $\bn$ field of the node (as these affect reachability), which are Line \ref{lin:lslins4}, \ref{lin:lslins5}, \ref{lin:lslins14} \& \ref{lin:lslins16} of \nplslins{} method of \algoref{lslins} and line \ref{lin:lsldel3} of \nplsldel{} method of \algoref{lsldelete}.\\
	\textbf{\texttt{Base condition:}} Initially, before the first event that changes the $\bn$ field of any node, we know that $(head, tail)$ $\in$ $S.PublicNodes$ $\land$ $\neg$($S.head.marked$) $\land$ $\neg$($S.tail.marked$) $\land$ $(S.head$ $\rightarrow^*_{\bn}$ $S.tail)$.\\
	\textbf{\texttt{Induction Hypothesis:}} Say, upto k events that change the next field of any node, $\forall n \in S.PublicNodes$, $(\neg S.n.marked)$ $\land$ $(S.head$ $\rightarrow^*_{\bn}$ $S.n)$. \\[0.1cm]
	\textbf{\texttt{Induction Step:}} 
	So, as seen from the code, the $(k+1)^{th}$ event which can change the $\bn$ field be only one of the following:\\
	\begin{enumerate}
		\item \textbf{\texttt{Line \ref{lin:lslins4} \& \ref{lin:lslins5} of \nplslins{} method:}} By observing the code, we notice that Line \ref{lin:lslins4} \& \ref{lin:lslins5} ($\bn{}$ field changing event) can be executed only after the \nplsls{} method of \algoref{lslsearch} returns. It is to be noted that (from \obsref{lslSearch2}), $(\bp, \rp, \rc, \bc)$ are locked, hence no other thread can change $S.\bp.marked$ and $S.\bc.marked$ simultaneously. Also, from \obsref{obs-nodekey}, a node's key field does not change after initialization. Before executing line \ref{lin:lslins4}, from \obsref{lslSearch3} ,
		\begin{equation}
		\label{eq:reachablefrombl1}
		(S.\bp.marked = false) \land (S.\bc.marked = false)
		\end{equation}
		And from \lemref{reach} and induction hypothesis,
		\begin{equation}
		\label{eq:reachablefrombl2}
		(S.Head \rightarrow^*_{\rn} S.\rc) \land (S.Head \rightarrow^*_{\bn} S.\bc)
		\end{equation}
		After line \ref{lin:lslins4}, we know that from $\rc$, public unmarked node, $\bc$ is also reachable, implies that,
		\begin{equation}
		\label{eq:reachablefrombl3}
		(S.\rc \rightarrow^*_{\bn} S.\bc)
		\end{equation}
		Also, before executing line \ref{lin:lslins5}, from induction hypothesis and \lemref{reach} ,
		\begin{equation}
		\label{eq:reachablefrombl4}
		(S.Head \rightarrow^*_{\bn} S.\bp) \land (S.Head \rightarrow^*_{\rn} S.\rc)
		\end{equation}
		After line \ref{lin:lslins5}, we know that from $\bp$, public unmarked node (from line \ref{lin:lslins3} of \nplslins{} method), $\rc$ is also reachable via $\bn$, implies that,
		\begin{equation}
		\label{eq:reachablefrombl5}
		(S.\bp \rightarrow^*_{\bn} S.\rc) \land (S.\rc.marked = false)
		\end{equation}
		From \Eqref{reachablefrombl3} and \Eqref{reachablefrombl5},
		\begin{equation}
		\begin{split}
		\label{eq:reachablefrombl6}
		(S.\bp \rightarrow^*_{\bn} S.\rc) \land (S.\rc \rightarrow^*_{\bn} S.\bc) \land \\(S.\rc.marked = false)
		\end{split}
		\end{equation}
		Since $(\bp, \rc) \in S.PublicNode$ and hence, $(S.Head \rightarrow^*_{\bn} S.\bp) \land (S.\bp\\ \rightarrow^*_{\bn} S.\rc) \land (S.\rc.marked = false) \Rightarrow (S.Head \rightarrow^*_{\bn} S.\rc)$.
		
		\item \textbf{\texttt{Line \ref{lin:lslins14} of \nplslins{} method:}} Line \ref{lin:lslins12} of the \nplslins{} method creates a new node, $node$ with $key$. Line \ref{lin:lslins14} then sets $(S.node.\bn$ = $S.\textcolor{blue}{currs[1]})$. Since this event doest not change the $\bn{}$ field of any node reachable from the head of the list (because $node \notin S.PublicNodes$), the lemma is not violated.
		
		\item \textbf{\texttt{Line \ref{lin:lslins16} of \nplslins{} method:}} By observing the code, we notice that Line \ref{lin:lslins16} ($\bn{}$ field changing event) can be executed only after the \nplsls{} method of \algoref{lslsearch} returns. It is to be noted that (from \obsref{lslSearch2}), $(\bp, \rp, \rc, \bc)$ are locked, hence no other thread can change $S.\bp.marked$ and $S.\bc.marked$ simultaneously. Also, from \obsref{obs-nodekey}, a node's key field does not change after initialization. Before executing line \ref{lin:lslins14}, from \obsref{lslSearch3} ,
		\begin{equation}
		\label{eq:reachablefrombl11}
		(S.\bp.marked = false) \land (S.\bc.marked = false)
		\end{equation}
		And from induction hypothesis,
		\begin{equation}
		\label{eq:reachablefrombl21}
		(S.Head \rightarrow^*_{\bn} S.\bc)
		\end{equation}
		After line \ref{lin:lslins14}, we know that from $node$, public unmarked node, $\bc$ is also reachable via $\bn$, implies that,
		\begin{equation}
		\label{eq:reachablefrombl31}
		(S.node \rightarrow^*_{\bn} S.\bc)
		\end{equation}
		Also, before executing line \ref{lin:lslins16}, from induction hypothesis,
		\begin{equation}
		\label{eq:reachablefrombl41}
		(S.Head \rightarrow^*_{\bn} S.\bp)
		\end{equation}
		After line \ref{lin:lslins16}, we know that from $\bp$, public unmarked node (because new node is created by default with unmarked field), $node$ is also reachable via $\bn$, implies that,
		\begin{equation}
		\label{eq:reachablefrombl51}
		(S.\bp \rightarrow^*_{\bn} S.node) \land (S.node.marked = false)
		\end{equation}
		From \Eqref{reachablefrombl31} and \Eqref{reachablefrombl51},
		\begin{equation}
		\begin{split}
		\label{eq:reachablefrombl61}
		(S.\bp \rightarrow^*_{\bn} S.node) \land (S.node \rightarrow^*_{\bn} S.\bc) \land (S.node.marked = false)
		\end{split}
		\end{equation}
		Since $(\bp, node) \in S.PublicNode$ and hence, $(S.Head \rightarrow^*_{\bn} S.\bp) \land (S.\bp \rightarrow^*_{\bn} S.node) \land (S.node.marked = false) \Rightarrow (S.Head \rightarrow^*_{\bn} S.node)$.
		
	\end{enumerate}
\end{proof}

\begin{corollary}
	\label{cor:BlsubsetRL}
	All public node $n$, is reachable from $head$ via bluelist is $subset$ of all public node $n$, is reachable from $head$ via redlist. Formally, $\langle \forall S, n: (n \in S.nodes) \land (S.head \rightarrow^*_{\bn} S.n) \subseteq (S.head \rightarrow^*_{\rn} S.n) \rangle$. 
	
\end{corollary}

\begin{proof}
	From \lemref{reach} , we know that all public nodes either marked or unmarked are reachable from head by $\rn{}$, also from \lemref{reach-bl} we have that all unmarked public nodes are reachable by $\bn{}$. Unmarked public nodes are subset of all public nodes thus the corollary.
	
\end{proof}

\begin{lemma}
	\label{lem:con-spec-lsl}
	Consider a concurrent history, $E^H$, for any successful method which is call by transaction $T_i$, after the post-state of $LP$ event of the method, node corresponding to the key should be part of $\rn$ and $max\_ts$ of that node should be equal to method transaction time-stamp. Formally, $\langle (node(key) \in ([E^H.Post(\lin{m_{i}})].Abs.\rn)) \land (node.max\_ts = TS(T_i)) \rangle$. 
\end{lemma}

\begin{proof}
	\begin{enumerate}
		\item \textbf{\texttt{For \rvmt{} method:}} By observing the code, each $\rvmt{}$ first invokes \nplsls{} method of \algoref{lslsearch} (line \ref{lin:delete21} of \cldd{} method of \algoref{commonLu&Del}). From \lemref{key-change} \& \lemref{key-change-bl} we have that the nodes in the underlying data-structure are in increasing order of their keys, thus the key on which the method is working has a unique location in underlying data-structure from \corref{uniquenodeKey} . So, when the \nplsls{} is invoked from a \mth{}, it returns correct location $(\bp, \rp,$ $\rc, \bc)$ of corresponding $key$ as observed from \obsref{lslSearch} \& \lemref{lslSearch-ret} and all are locked, hence no other thread can change simultaneously (from \obsref{lslSearch2}). 
		
		In the pre-state of $LP$ event of $\rvmt{}$ , if $(node.key \in S.Abs.\rn)$, means $key$ is already there in $\rn$ and time-stamp of that node is less then the $\rvmt{}$ transactions time-stamp, from \nptov{} method of \algoref{tovalidation} , then in the post-state of $LP$ event of $\rvmt{}$, $node.key$ should be the part of $\rn$ from \obsref{futurestate} and $key$ can't be change from \obsref{obs-nodekey} and it just update the $max\_ts$ field for corresponding node $key$ by method transaction time-stamp else abort. 
		
		In the pre-state of $LP$ event of $\rvmt{}$ , if $(node.key \notin S.Abs.\rn)$, means $key$ is not there in $\rn$ then, in the post-state of $LP$ event of $\rvmt{}$, insert the $node$ corresponding to the $key$ into $\rn$ by using \nplslins{} method of \algoref{lslins} and update the $max\_ts$ field for corresponding node $key$ by method transaction time-stamp. Since, $node.key$ should be the part of $\rn$ from \obsref{futurestate} and $key$ can't be change from \obsref{obs-nodekey} , in post-state of $LP$ event of $\rvmt{}$.

		\item \textbf{\texttt{For \upmt{} method:}} By observing the code, each $\upmt{}$ also first invokes \nplsls{} method of \algoref{lslsearch} (line \ref{lin:tryc7} of \nptc{} method of \algoref{trycommit} ).
		From \lemref{key-change} \& \lemref{key-change-bl} we have that the nodes in the underlying data-structure are in increasing order of their keys, thus the key on which the method is working has a unique location in underlying data-structure from \corref{uniquenodeKey} . So, when the \nplsls{} is invoked from a \mth{}, it returns correct location $(\bp, \rp,$ $\rc, \bc)$ of corresponding $key$ as observed from \obsref{lslSearch} \& \lemref{lslSearch-ret} and all are locked, hence no other thread can change simultaneously (from \obsref{lslSearch2}). 
		
		\begin{enumerate}
			\item \textbf{\texttt{If \upmt{} is insert:}} In the pre-state of $LP$ event of $\upmt{}$, if $(node.key \in S.Abs.\rn)$, means $key$ is already there in $\rn$ and time-stamp of that node is less then the $\upmt{}$ transactions time-stamp, from \nptov{} method of \algoref{tovalidation} , then in the post-state of $LP$ event of $\upmt{}$, $node.key$ should be the part of $\rn$ and it just update the $max\_ts$ field for corresponding node $key$ by method transaction time-stamp else abort.
			
			In the pre-state of $LP$ event of $\upmt{}$, if $(node.key \notin S.Abs.\rn)$, means $key$ is not there in $\rn$ then in the post-state of $LP$ event of $\upmt{}$, it will insert the $node$ corresponding to the $key$ into the $\rn$ as well as $\bn$, from $\nplslins{}$ method of \algoref{lslins} at line \ref{lin:tryc33} of $\nptc{}$ method of \algoref{trycommit} and update the $max\_ts$ field for corresponding node $key$ by method transaction time-stamp.
			Once a node is created it will never get deleted from \obsref{futurestate} and node corresponding to a key can't be modified from \obsref{obs-nodekey}.
			
			\item \textbf{\texttt{If \upmt{} is delete:}} In the pre-state of $LP$ event of $\upmt{}$, if $(node.key \in S.Abs.\rn)$, means $key$ is already there in $\rn$ and time-stamp of that node is less then the $\upmt{}$ transactions time-stamp, from \nptov{} method of \algoref{tovalidation} , then in the post-state of $LP$ event of $\upmt{}$, $node.key$ should be the part of $\rn$, from $\nplsldel{}$ method of \algoref{lsldelete} at line \ref{lin:tryc37} of $\nptc{}$ method of \algoref{trycommit} and it just update the $max\_ts$ field for corresponding node $key$ by method transaction time-stamp else abort.
			
			In the pre-state of $LP$ event of $\upmt{}$, $(node.key \notin S.Abs.\rn)$ this should not be happen because execution of \npdel{} method of \algoref{delete} must have already inserted a node in the underlying data-structure prior to $\nptc$ method of \algoref{trycommit} . Thus, $(node.key \in S.Abs.\rn)$ and update the $max\_ts$ field for corresponding node $key$ by method transaction time-stamp else abort. 
			
		\end{enumerate}
	\end{enumerate}
\end{proof}
In \otm{} we have a \emph{\upmt{} execution} phase where all buffered $\upmt{}$ take effect together after successful validation of each of them.
Following problem may arise if two $\upmt{}$ within same transaction have at least one shared node amongst its recorded $(\bp,\rp,\\ \rc, \bc )$, in this case the previous $\upmt{}$ effect might be overwritten if the next $\upmt{}$ preds and currs are not updated according to the updates done by the previous $\upmt{}$. Thus program order might get violated. Thus to solve this we have intra trans validation after each $\upmt{}$ in $\nptc$, during \emph{\upmt{} execution} phase.

\begin{lemma}
	\label{lem:poval}
	\nppoval{} preserve the program order within a transaction.
\end{lemma}
\begin{proof}
	We are taking contradiction that $\nppoval{}$ is not preserving program order means two consecutive $\upmt{}$ of same transaction which are having at least one shared node amongst its recorded($\bp, \rp, \rc, \bc$) then effect of first $\upmt{}$ will be overwritten by the next $\upmt{}$.
	
	By observing the code at line \ref{lin:tryc17} of $\nptc{}$ method of \algoref{trycommit}, current $\upmt{}$ will go for $\nppoval{}$ and at line \ref{lin:threadv3} of $\nppoval{}$ method of \algoref{povalidation} , current $\upmt{}$ will validate its $(\bp.marked)$ and $(\bp.\bn != \bc)$. If any condition is true then, at line \ref{lin:threadv4} of $\nppoval{}$ method of \algoref{povalidation}, will check for previous $\upmt{}$. If the previous $\upmt{}$ is insert then the current $\upmt{}$ update its $\bp$ to previous $\upmt{}$, $node.key$ else set current $\upmt{}$ $\bp$ to previous $\upmt{}$ $\bp$.
	
	After that at line \ref{lin:threadv10} of $\nppoval{}$ method of \algoref{povalidation} , current $\upmt{}$ validate its $(\rp.\rn != \rc)$. If condition is true then current $\upmt{}$ set its $\rp$ to previous $\upmt{}$, $node.key$.    
	
	If we will not update the current method preds and currs using $\nppoval{}$ then effect of first $\upmt{}$ will be overwritten by the next $\upmt{}$.
	
\end{proof}

\begin{observation}
	\label{obs:poorder-validation}
	For any global state S, the \nppoval{} in \nptc{} preserves the properties of \nplsls{} as proved in \obsref{lslSearch} \& \lemref{lslSearch-ret} .
\end{observation}

\begin{lemma}
	\label{lem:tryc-con-lsl}
	Consider a concurrent history, $E^H$, after the post-state of $LP$ event of successful $\nptc$ method, where each key belonging to the last $\upmt{}$ of that transaction, then,
	\begin{enumerate}[label=\ref{lem:tryc-con-lsl}.\arabic*]
		\item \label{lem:tryc-con-lsl-in} If $\upmt{}$ is insert, then node corresponding to the key should be part of $\bn$ and node.val should be equal to v. Formally, $\langle (node(key) \in ([E^H.Post(\lin{m_{i}})].Abs.$\bn$) \land (node.val = v)\rangle$. 				
		
		\item \label{lem:tryc-con-lsl-dl} If $\upmt{}$ is delete, then node corresponding to the key should not be part of $\bn$. Formally, $\langle (node(key) \notin ([E^H.Post(\lin{m_{i}})].Abs.$\bn$) \rangle$.
	\end{enumerate}
	
\end{lemma}

\begin{proof}
	By observing the code, each $\upmt{}$ also first invokes \nplsls{} method of \algoref{lslsearch} (line \ref{lin:tryc7} of \nptc{} method of \algoref{trycommit} ).
	From \lemref{key-change} \& \lemref{key-change-bl} we have that the nodes in the underlying data-structure are in increasing order of their keys, thus the key on which the method is working has a unique location in underlying data-structure from \corref{uniquenodeKey} . So, when the \nplsls{} is invoked from a \mth{}, it returns correct location $(\bp, \rp,$ $\rc, \bc)$ of corresponding $key$ as observed from \obsref{lslSearch} \& \lemref{lslSearch-ret} and all are locked, hence no other thread can change simultaneously (from \obsref{lslSearch2}). 
	\begin{enumerate}[label=\ref{lem:tryc-con-lsl}.\arabic*]
		\item \textbf{\texttt{If \upmt{} is insert:}} In the pre-state of $LP$ event of $\upmt{}$ at Line \ref{lin:tryc19}, \ref{lin:tryc24} of $\nptc{}$ method of \algoref{trycommit}, if $(node.key \in S.Abs.\rn)$, means $key$ is already there in $\rn$ and time-stamp of that node is less then the $\upmt{}$ transactions time-stamp, from \nptov{} method of \algoref{tovalidation}, then in the post-state of $LP$ event of $\upmt{}$, $node.key$ should be the part of $\bn$ and it will update the $value$ as $v$.
		
		In the pre-state of $LP$ event of $\upmt{}$ at Line \ref{lin:tryc28} of $\nptc{}$ method of \algoref{trycommit} , if $(node.key \notin S.Abs.\rn)$, means $key$ is not there in $\rn$ then in the post-state of $LP$ event of $\upmt{}$, it will insert the $node$ corresponding to the $key$ into the $\bn$, from $\nplslins{}$ method of \algoref{lslins} at line \ref{lin:tryc29} of $\nptc{}$ method of \algoref{trycommit} and update the $value$ as $v$.
		Once a node is created it will never get deleted from \obsref{futurestate} and node corresponding to a key can't be modified from \obsref{obs-nodekey}.
		
		\item \textbf{\texttt{If \upmt{} is delete:}} In the pre-state of $LP$ event of $\upmt{}$ at Line \ref{lin:tryc34} of $\nptc{}$ method of \algoref{trycommit} , if $(node.key \in S.Abs.\bn)$, means $key$ is already there in $\bn$ and time-stamp of that node is less then the $\upmt{}$ transactions time-stamp, from \nptov{} method of \algoref{tovalidation} , then in the post-state of $LP$ event of $\upmt{}$, $node.key$ should not be the part of $\bn$, from $\nplsldel{}$ method of \algoref{lsldelete} at line \ref{lin:tryc34} of $\nptc{}$ method of \algoref{trycommit} .
		
		In the pre-state of $LP$ event of $\upmt{}$, $(node.key \notin S.Abs.\rn)$ this should not be happen because execution of \npdel{} method of \algoref{delete} must have already inserted a node in the underlying data-structure prior to $\nptc$ method of \algoref{trycommit} . 
		
	\end{enumerate}
	
\end{proof}

\begin{lemma}
	\label{lem:lookupT-conc}
	Consider a concurrent history, $E^H$, where S be the pre-state of $LP$ event of successful $rvm$ method, in that, if node corresponding to the key is the part of $\bn$ and node.val is equal to v then, $\rvmt{}$ return $OK$ and value v. Formally, $\langle (node(key) \in ([E^H.Pre(\lin{m_{i}})].Abs.\bn)) \land (S.node.val = v ) \Longrightarrow \rvm(key, OK, v) \rangle$. 
\end{lemma}
\begin{proof}
	Let the $\rvmt{}$ is \npluk{} method of \algoref{lookup} and it is the first key method of the transaction, we ignore the abort case for simplicity.\\
	From line \ref{lin:delete21} of \cldd{} method of \algoref{commonLu&Del} , when \nplsls{} method of \algoref{lslsearch} returns
	we have $(\textcolor{blue}{preds[0]}, \textcolor{red}{preds[1]}, \textcolor{red}{currs[0]}, \textcolor{blue}{currs[1]}$ $\in$ $S.PublicNodes)$ and are locked(from \obsref{lslSearch1} \& \obsref{lslSearch2}) until \npluk{} method of \algoref{lookup} return. Also, from \lemref{lslSearch-ret1} ,
	\begin{equation}
	\label{eq:rvmdel1}
	(S.\textcolor{blue}{preds[0]}.key < key \leq S.\textcolor{blue}{currs[1]}.key)
	\end{equation}
	To return OK, $S.\textcolor{blue}{currs[1]}$ should be reachable from the head via bluelist from \defref{Abs} , in the pre-state of $LP$ of $\rvmt{}$. And after observing code, at line \ref{lin:delete25} of \cldd{} method of \algoref{commonLu&Del},
	\begin{equation}
	\label{eq:rvmdel2}
	(S.\textcolor{blue}{currs[1]}.key = key) \xRightarrow[]{\Eqref{rvmdel1}} (S.\textcolor{blue}{preds[0]}.key < (key = S.\textcolor{blue}{currs[1]}.key))
	\end{equation}
	Also, from \obsref{lslSearch3} ,
	\begin{equation}
	\label{eq:rvmdel4}
	(S.\textcolor{blue}{preds[0]}.\bn = S.\textcolor{blue}{currs[1]})
	\end{equation}
	And $(\textcolor{blue}{currs[1]} \in S.nodes)$, we know $(\textcolor{blue}{currs[1]} \in S.Abs.\bn)$ where S is the pre-state of the LP event of the method. 
	From \lemref{tryc-con-lsl-in} , there should be a prior $\upmt{}$ which have to be $insert$ and $\bc.val$ is equal to $v$. Since \obsref{obs-nodekey} tells, no node changes its $key$ value after initialization. Hence 
	$(node(key) \in ([E^H.Pre(\lin{m_{i}})].Abs.\bn) \land (S.node.val = v))$.
	\\
	\\\texttt{*Same argument can be extended to \npdel{} \mth{}.}
	
\end{proof}

\begin{lemma}
	\label{lem:lookupF-conc}
	Consider a concurrent history, $E^H$, where S be the pre-state of $LP$ event of successful $\rvmt{}$, in that, if node corresponding to the key is not the part of $\bn$ then, $\rvmt{}$ return $FAIL$. Formally, $\langle (node(key) \notin ([E^H.Pre(\lin{m_{i}})].Abs.\bn)) \Longrightarrow \rvm(key, FAIL) \rangle$. 
\end{lemma}
\begin{proof}
	Let the $\rvmt{}$ is \npluk{} method of \algoref{lookup} and it is the first key method of the transaction, we ignore the abort case for simplicity.
	\begin{enumerate}
		\item From line \ref{lin:delete21} of \cldd{} method of \algoref{commonLu&Del}, when \nplsls{} method of \algoref{lslsearch} returns we have $(\textcolor{blue}{preds[0]}, \textcolor{red}{preds[1]}, \textcolor{red}{currs[0]}, \textcolor{blue}{currs[1]}$ $\in$ $S.PublicNodes)$ and are locked(from \obsref{lslSearch1} \& \obsref{lslSearch2}) until \npluk{} method of \algoref{lookup} return. Also, from \lemref{lslSearch-ret2} ,
		\begin{equation}
		\label{eq:rvmdel11}
		(S.\textcolor{red}{preds[1]}.key < key \leq S.\textcolor{red}{currs[0]}.key)
		\end{equation}
		To return FAIL, $S.\textcolor{red}{currs[0]}$ should not be reachable from the head via bluelist from \defref{Abs} , in the pre-state of $LP$ of $\rvmt{}$. And after observing code, at line \ref{lin:delete29} of \cldd{} method of \algoref{commonLu&Del} ,
		\begin{equation}
		\label{eq:rvmdel12}
		(S.\textcolor{red}{currs[0]}.key = key) \xRightarrow[]{\Eqref{rvmdel11}} (S.\textcolor{red}{preds[1]}.key < (key = S.\textcolor{red}{currs[0]}.key))
		\end{equation}
		Also, from \obsref{lslSearch3} ,
		\begin{equation}
		\label{eq:rvmdel14}
		(S.\textcolor{red}{preds[1]}.\rn = S.\textcolor{red}{currs[0]})
		\end{equation}
		And $(\textcolor{red}{currs[0]} \in S.nodes)$, we know $(\textcolor{red}{currs[0]} \in S.Abs.\rn)$ where S is the pre-state of the LP event of the method and $(S.\rc.marked = true)$. Thus, $(\rc \notin S.Abs.\bn)$ from \defref{Abs} . 
		Hence $(node(key) \notin ([E^H.Pre(\lin{m_{i}})].Abs.\bn)$
		
		\item From line \ref{lin:delete21} of \cldd{} method of \algoref{commonLu&Del}, when \nplsls{} method of \algoref{lslsearch} returns we have $(\textcolor{blue}{preds[0]}, \textcolor{red}{preds[1]}, \textcolor{red}{currs[0]}, \textcolor{blue}{currs[1]}$ $\in$ $S.PublicNodes)$ and are locked(from \obsref{lslSearch1} \& \obsref{lslSearch2}) until \npluk{} method of \algoref{lookup} return. Also, from \lemref{lslSearch-ret2} ,
		\begin{equation}
		\label{eq:rvmdel111}
		(S.\textcolor{red}{preds[1]}.key < key \leq S.\textcolor{red}{currs[0]}.key)
		\end{equation}
		And after observing code, at line \ref{lin:delete33} of \cldd{} method of \algoref{commonLu&Del} ,
		\begin{equation}
		\begin{split}
		\label{eq:rvmdel121}
		(S.\textcolor{blue}{currs[1]}.key \neq key) \land (S.\textcolor{red}{currs[0]}.key \neq key) \xRightarrow[]{\Eqref{rvmdel111}} \\(S.\textcolor{red}{preds[1]}.key < key < S.\textcolor{red}{currs[0]}.key)
		\end{split}
		\end{equation}
		Also, from \obsref{lslSearch3} ,
		\begin{equation}
		\label{eq:rvmdel141}
		(S.\textcolor{red}{preds[1]}.\rn = S.\textcolor{red}{currs[0]})
		\end{equation}
		From \Eqref{rvmdel121}, we can say that, $(node(key) \notin S.Abs)$ and from \Corref{key-abs}, we conclude that $node(key)$ not in the state after $\nplsls$ returns.  Since \obsref{obs-nodekey} tells, no node changes its key value after initialization. Hence 
		$(node(key) \notin ([E^H.Pre(\lin{m_{i}})].Abs.\bn))$.
		\\
		\\\texttt{*Same argument can be extended to \npdel{} \mth{}.}
	\end{enumerate}

\end{proof}

\begin{observation}
	\label{obs:tryupbl}	
	Only the successful $\nptc{}$ method working on the key k can update the Abs.$\bn$.
\end{observation}
By observing the code, only the successful $\nptc{}$ method of \algoref{trycommit} is changing the $\bn$. There is no line which is changing the $\bn$ in $\npdel{}$ method of \algoref{delete} and $\npluk{}$ method of \algoref{lookup} . Such that $\rvmt{}$ is not changing the $\bn$.
\begin{observation}
	\label{obs:lockacq}
	If $\nptc{}$ and $\rvmt{}$ wants to update $Abs$ on the key k, then first it has to acquire the lock on the node corresponding to the key k.
\end{observation}

If node corresponding to the key $k$ is not the part of $Abs$ then $\nptc{}$ and $\rvmt{}$ have to create the node corresponding to the key $k$ and before adding it into the shared memory($Abs$), it has to acquire the lock on the particular node corresponding to the key $k$.  

\begin{definition}
	\label{def:fLP}
	First unlocking point of each successful method is the $LP$.
\end{definition}
\textbf{Linearization Points:} Here, we list the linearization points (LPs) of each method. Note that each method of the list can return either $OK$, $FAIL$ or $ABORT$. So, we define the LP for all the methods:

\begin{enumerate}
	\item \emph{STM\_begin():} $get\&inc(sh\_cntr\uparrow)$ at \Lineref{begin3} of \emph{STM\_begin()}.
	\item \emph{STM\_insert(ht, k, OK/FAIL/ABORT):} \llsval{$value \downarrow$} at \Lineref{insert5} of \emph{STM\_insert()}.
	\item \emph{STM\_delete(ht, k, OK/FAIL/ABORT):} $\bp$.$\texttt{unlock()}$ at \Lineref{rpandc} of \emph{releasePred\&CurrLocks()} (\algoref{releasepreds&currs}) is the LP of \npdel{}. Which is called from \Lineref{delete433} of \cldd{} (\algoref{commonLu&Del}) at \Lineref{deletecld} of \npdel{}. 
	\item \emph{STM\_lookup(ht, k, OK/FAIL/ABORT):} $\bp$.$\texttt{unlock()}$ at \Lineref{rpandc} of \emph{releasePred\&CurrLocks()} (\algoref{releasepreds&currs}) is the LP of \npluk{}. Which is called from \Lineref{delete433} of \cldd{} (\algoref{commonLu&Del}) at \Lineref{lookupa20} of \npluk{}. 
	\item \emph{STM\_tryC(ht, k, OK/FAIL/ABORT):} $le_i$.$\bp$.$\texttt{unlock()}$ at \Lineref{rlock3} of \emph{releaseOrderedLocks()} (\algoref{releaseorderedlocks}). Which is called at \Lineref{tryc44} of \nptc{}.
	
\end{enumerate}

\begin{observation}
	\label{obs:locksimu}
	Two concurrent conflicting methods of different transaction can't acquire the lock on the same node corresponding to the key $k$ simultaneously.
\end{observation}

\begin{observation}
	\label{obs:LPorder}
	Consider two concurrent conflicting method of different transactions say $m_i$ of $T_i$ and $m_j$ of $T_j$ working on the same key k, then, if $ul(m_i(k))$ happen before the $l(m_j(k))$ then $LP(m_i)$ happen before $LP(m_j)$. Formally, $\langle (ul(m_i(k)) \prec l(m_j(k))) \Rightarrow (LP(m_i) \prec LP(m_j)) \rangle$
\end{observation}

If two concurrent conflicting methods are working on the same key k and want to update $Abs$ then they have to acquire the lock on the node corresponding to the key $k$ from \obsref{lockacq} and one of them succeed from \obsref{locksimu} . If $ul(m_i(k))$ happen before the $l(m_j(k))$ then from \defref{fLP} , $LP(m_i)$ happen before the $LP(m_j)$.

\begin{lemma}
	\label{lem:intertryC}
	Consider two state, $S_1$, $S_2$ s.t. $S_1$ $\sqsubset$ $S_2$ and $S_1.\bn.value(k)$ $\neq$ $S_2.\bn.value(k)$ then there exist $S'$ s.t. $S'$ $\sqsubset$ $S_2$ and $S'$ contain the $\nptc{}$ method on the same key k.
	Formally, $\langle (S_1.\bn.value(k) \neq (S_2.\bn.value(k)) \Rightarrow \exists (S' s.t., S_1.\bn \prec S'.LP(\tryc) \prec S_2.\bn) \rangle$.
	Where $S_1$ is the post-state of LP event of $\nptc{}$ method and $S_2$ is the pre-state of LP event of $\rvmt{}$. 
\end{lemma}

\begin{proof}
	In the state $S_1$ and $S_2$, if the $value$ corresponding to the key $k$ is not same then from \obsref{tryupbl} , we know that only the successful $\nptc{}$ method working on the same key $k$ can update the Abs.$\bn$. For updating the Abs on the key $k$ it has to acquire the lock on the node corresponding to the key $k$ from \obsref{lockacq}. Such that, $l(\tryc(k))$ happen before the $l(S_2(k))$ from \obsref{locksimu} , then, $ul(\tryc(k))$ happen before the $l(S_2(k))$ then $LP(\tryc)$ happen before the $LP(S_2)$ from \obsref{LPorder} .
\end{proof}

\begin{lemma}
	\label{lem:provingLegalityy}
	Consider a concurrent history, $E^H$, let there be a successfull $\nptc{}$ method of a transaction $T_i$ which last updated the node corresponding to $k$. Now, Consider a successful $\rvmt{}$ of a transaction $T_j$ on key $k$ then,
	\begin{enumerate}[label=\ref{lem:provingLegalityy}.\arabic*]
		\item \label{lem:provingLegalityy-in} If in the the pre-state of $LP$ event of the $\rvmt{}$ , node corresponding to the key $k$ is part of $\bn$ and value is $v$. Then the last \upmt{} of $\nptc{}$ would be insert on same key $k$ and value $v$ and it should be the previous closest to the \rvmt{}.
		
		\item \label{lem:provingLegalityy-del} If in the the pre-state of $LP$ event of the $\rvmt{}$ , node corresponding to the key $k$ is not part of the $\bn$. Then the last \upmt{} in $\nptc{}$ would be delete on same key $k$ and it should be the previous closest to the \rvmt{}.
	\end{enumerate}
\end{lemma}
\begin{proof}
	
	\begin{enumerate}[label=\ref{lem:provingLegalityy}.\arabic*]
		\item For proving this we are taking a contradiction that in the pre-state of $\rvmt{}$, node corresponding to the key $k$ is the part of $\bn$ and value as $v$, for that, there exist a previous closest successful $\tryc{}$ method should having the last $\upmt{}$ as insert on the same key $k$ from \corref{uniquenodeKey} , node corresponding to the key $k$ is unique and value is $v'$. If the $value$ of the node corresponding to the key $k$ is different for both the methods then from \lemref{intertryC} , there should be some other transaction $\tryc{}$ method working on the same key $k$ and its $LP$ should lies in between these two methods $LP$. Therefore that intermediate $\tryc{}$ should be the previous closest method for the $\rvmt{}$ and it will return the same value as previous closest method inserted.
		
		\item For proving this we are taking contradiction that previous closest successful $\tryc{}$ method should having the last $\upmt{}$ as insert on the same key $k$. If the last $\upmt{}$ is insert on the same key $k$ then after the post-state of successful $\tryc{}$ method, node corresponding to the key $k$ should be the part of $\bn$ from \lemref{tryc-con-lsl-in} . But we know that in the pre-state of $\rvmt{}$, node corresponding to the key $k$ is not the part of $\bn$. Such that previous closest successful $\tryc{}$ method should not having last $\upmt{}$ as insert on the same key $k$.
		Hence contradiction. 	
	\end{enumerate}
\end{proof}

\begin{theorem}
	The sequential history generated by \otm{} at method level is legal.
\end{theorem}

\begin{theorem}
	The legal sequential history generated by \otm{} at method level is Linearizable.
\end{theorem}

\textbf{Construction of sequential history} based on the $LP$ of concurrent methods of a concurrent history, $E^H$, and execute them in their $LP$ order for returning the same $return$ $value$.

\begin{lemma}
	\label{lem:provingLegalityysec}
	Let there be a successfull $\nptc{}$ method of a transaction $T_i$ which last updated the node corresponding to $k$. Now, consider a successful $\rvmt{}$ of a transaction $T_j$ on key $k$ then,
	\begin{enumerate}[label=\ref{lem:provingLegalityysec}.\arabic*]
		\item \label{lem:provingLegalityysec-in} If in the the pre-state of $\rvmt{}$ , node corresponding to the key $k$ is part of $\bn$ and value is $v$. Then the last \upmt{} of $\nptc{}$ would be insert on same key $k$ and value $v$ and it should be the previous closest to the \rvmt{}.
		
		\item \label{lem:provingLegalityysec-del} If in the the pre-state of $\rvmt{}$ , node corresponding to the key $k$ is not part of the $\bn$. Then the last \upmt{} in $\nptc{}$ would be delete on same key $k$ and it should be the previous closest to the \rvmt{}.
	\end{enumerate}
\end{lemma}

\begin{proof}
	
	\begin{enumerate}[label=\ref{lem:provingLegalityysec}.\arabic*]
		\item For proving this we are taking a contradiction that in the pre-state of $\rvmt{}$, node corresponding to the key $k$ is the part of $\bn$ and value as $v$, for that, there exist a previous closest successful $\tryc{}$ method should having the last $\upmt{}$ as insert on the same key $k$ from \corref{uniquenodeKey} , node corresponding to the key $k$ is unique and value is $v'$. If the $value$ of the node corresponding to the key $k$ is different for both the methods then from \lemref{intertryC} , there should be some other transaction $\tryc{}$ method working on the same key $k$ and its $LP$ should lies in between these two methods $LP$. Therefore that intermediate $\tryc{}$ should be the previous closest method for the $\rvmt{}$ and it will return the same value as previous closest method inserted.
		
		\item For proving this we are taking contradiction that previous closest successful $\tryc{}$ method should having the last $\upmt{}$ as insert on the same key $k$. If the last $\upmt{}$ is insert on the same key $k$ then after the post-state of successful $\tryc{}$ method, node corresponding to the key $k$ should be the part of $\bn$ from \lemref{tryc-con-lsl-in} . But we know that in the pre-state of $\rvmt{}$, node corresponding to the key $k$ is not the part of $\bn$. Such that previous closest successful $\tryc{}$ method should not having last $\upmt{}$ as insert on the same key $k$.
		Hence contradiction. 	
	\end{enumerate}
\end{proof}

\begin{lemma}
	\label{lem:seq-spec-lsl}
	Consider a sequential history, $E^S$, for any successful method which is call by transaction $T_i$, after the post-state of the method, node corresponding to the key should be part of $\rn$ and $max\_ts$ of that node should be equal to method transaction time-stamp. Formally, $\langle (node(key) \in (P.Abs.\rn)) \land (P.node.max\_ts = TS(T_i)) \rangle$. Where P is the post-state of the method.
\end{lemma}

\begin{proof}
	\begin{enumerate}
		\item \textbf{\texttt{For \rvmt{} method:}} By observing the code, each $\rvmt{}$ first invokes \nplsls{} method of \algoref{lslsearch} (line \ref{lin:delete21} of \cldd{} method of \algoref{commonLu&Del}). From \lemref{key-change} \& \lemref{key-change-bl} we have that the nodes in the underlying data-structure are in increasing order of their keys, thus the key on which the method is working has a unique location in underlying data-structure from \corref{uniquenodeKey} . So, when the \nplsls{} is invoked from a \mth{}, it returns correct location $(\bp, \rp,$ $\rc, \bc)$ of corresponding $key$ as observed from \obsref{lslSearch} \& \lemref{lslSearch-ret} and all are locked, hence no other thread can change simultaneously (from \obsref{lslSearch2}). 
		
		In the pre-state of $\rvmt{}$ , if $(node.key \in S.Abs.\rn)$, means $key$ is already there in $\rn$ and time-stamp of that node is less then the $\rvmt{}$ transactions time-stamp, from \nptov{} method of \algoref{tovalidation} , then in the post-state of $\rvmt{}$, $node.key$ should be the part of $\rn$ from \obsref{futurestate} and $key$ can't be change from \obsref{obs-nodekey} and it just update the $max\_ts$ field for corresponding node $key$ by method transaction time-stamp else abort. 
		
		In the pre-state of $\rvmt{}$ , if $(node.key \notin S.Abs.\rn)$, means $key$ is not there in $\rn$ then, in the post-state of $\rvmt{}$, insert the $node$ corresponding to the $key$ into $\rn$ by using \nplslins{} method of \algoref{lslins} and update the $max\_ts$ field for corresponding node $key$ by method transaction time-stamp. Since, $node.key$ should be the part of $\rn$ from \obsref{futurestate} and $key$ can't be change from \obsref{obs-nodekey} , in post-state of $\rvmt{}$.

		\item \textbf{\texttt{For \upmt{} method:}} By observing the code, each $\upmt{}$ also first invokes \nplsls{} method of \algoref{lslsearch} (line \ref{lin:tryc7} of \nptc{} method of \algoref{trycommit} ).
		From \lemref{key-change} \& \lemref{key-change-bl} we have that the nodes in the underlying data-structure are in increasing order of their keys, thus the key on which the method is working has a unique location in underlying data-structure from \corref{uniquenodeKey} . So, when the \nplsls{} is invoked from a \mth{}, it returns correct location $(\bp, \rp,$ $\rc, \bc)$ of corresponding $key$ as observed from \obsref{lslSearch} \& \lemref{lslSearch-ret} and all are locked, hence no other thread can change simultaneously (from \obsref{lslSearch2}). 
		
		\begin{enumerate}
			\item \textbf{\texttt{If \upmt{} is insert:}} In the pre-state of $\upmt{}$, if $(node.key \in S.Abs.\rn)$, means $key$ is already there in $\rn$ and time-stamp of that node is less then the $\upmt{}$ transactions time-stamp, from \nptov{} method of \algoref{tovalidation} , then in the post-state of $\upmt{}$, $node.key$ should be the part of $\rn$ and it just update the $max\_ts$ field for corresponding node $key$ by method transaction time-stamp else abort.
			
			In the pre-state of $\upmt{}$, if $(node.key \notin S.Abs.\rn)$, means $key$ is not there in $\rn$ then in the post-state of $\upmt{}$, it will insert the $node$ corresponding to the $key$ into the $\rn$ as well as $\bn$, from $\nplslins{}$ method of \algoref{lslins} at line \ref{lin:tryc31} of $\nptc{}$ method of \algoref{trycommit} and update the $max\_ts$ field for corresponding node $key$ by method transaction time-stamp.
			Once a node is created it will never get deleted from \obsref{futurestate} and node corresponding to a key can't be modified from \obsref{obs-nodekey}.
			
			\item \textbf{\texttt{If \upmt{} is delete:}} In the pre-state of $\upmt{}$, if $(node.key \in S.Abs.\rn)$, means $key$ is already there in $\rn$ and time-stamp of that node is less then the $\upmt{}$ transactions time-stamp, from \nptov{} method of \algoref{tovalidation} , then in the post-state of $\upmt{}$, $node.key$ should be the part of $\rn$, from $\nplsldel{}$ method of \algoref{lsldelete} at line \ref{lin:tryc37} of $\nptc{}$ method of \algoref{trycommit} and it just update the $max\_ts$ field for corresponding node $key$ by method transaction time-stamp else abort.
			
			In the pre-state of $\upmt{}$, $(node.key \notin S.Abs.\rn)$ this should not be happen because execution of \npdel{} method of \algoref{delete} must have already inserted a node in the underlying data-structure prior to $\nptc$ method of \algoref{trycommit} . Thus, $(node.key \in S.Abs.\rn)$ and update the $max\_ts$ field for corresponding node $key$ by method transaction time-stamp else abort. 
			
		\end{enumerate}

	\end{enumerate}
\end{proof}

\begin{corollary}
	\label{cor:nodeinabs}
	After the post-state of any successful method on a key ensures that underlying \rn{} contains a unique node corresponding to the key and $max\_ts$ field is updated by methods transactions time-stamp. 
\end{corollary}


\subsection{Transactional Level}
\label{sec:txlevel}

From \Secref{opnlevel} we are guaranteed to have a \seq{} history or in other terms we have a linearizable history. Now we shall prove that such linearizable history obtained from \otm{} is \opq.

\begin{observation}{H}
	\label{obs:defH}
	is a sequential history obtained from \otm{}, as shown at method level using LP.
\end{observation}

\begin{definition}{$CG(H)$}
	is a conflict graph of H.
\end{definition}

\begin{lemma}
	\label{lem:serialIsAcyclic}
	Conflict graph of a serial history is acyclic.
\end{lemma}
\begin{proof}
	If conflict graph of serial history contains an conflict edge  ( $T_1$, $T_2$ ), then $\levt{T_1} \prec_H \fevt{T_2}$. Now, assume that conflict graph of a serial history is cyclic, then their exist a cycle path in the form ($T_1$, $T_2$ $\cdots$ $T_k$, $T_1$), (k $\geq$ 1).  So, transitively,
	\begin{equation}
	\begin{split}
	\label{eq:serialIsAcycliceq}
	((\levt{T_1} \prec_H \fevt{T_k}) \land (\levt{T_k} \prec_H \fevt{T_1})) \Rightarrow \\(\levt{T_1} \prec_H \fevt{T_1})
	\end{split}
	\end{equation}
	
	This contradict our assumption as \Eqref{serialIsAcycliceq} is impossible, from definition of program order of a transaction. Thus, cycle is not possible in serial history.
	
\end{proof}

\begin{observation}
	\label{obs:topsort}
	$H_2$ is an history generated by applying topological sort on $CG(H_1)$. \end{observation}
\begin{observation}
	\label{obs:TSorder}
	Topological sort maintains conflict-order and real-time order of the original history $H_1$. 
\end{observation}

\begin{definition}{\conf{H}}
	\label{def:confh}
	is a set of ordered pair ($T_i$, $T_j$), such that their exists conflicting methods $m_i$, $m_j$ in $T_i$ \& $T_j$ respectively, such that $m_i$ $\prec_{H}^{\mr}$ $m_j$. And it is represented as  $\prec$$^{CO}_{H}$.
\end{definition}

\begin{lemma}
	\label{lem:propConflictPreserv}
	$H_1$ is legal \& $CG(H_1)$ is acyclic. then,
	\begin{enumerate}[label=\ref{lem:propConflictPreserv}.\arabic*]
		\item $H_1$ is equivalent to $H_2$ $\Rightarrow$ ($\met{H1} = \met{H2})$. \label{lem:propConflictPreserv1}
		\item $\prec$$^{CO}_{H1}$ $\subseteq$ $\prec$$^{CO}_{H2}$. i.e. $H_1$ preserves the conflicts of $H_2$\label{lem:propConflictPreserv2}
	\end{enumerate}
\end{lemma}
\begin{proof}
	\lemref{propConflictPreserv2}\\
	We should show that $\forall$( $T_i$, $T_j$ ), such that ( ( $T_i$, $T_j$ ) $\in$ $\prec$$^{CO}_{H1}$  $\Rightarrow$ ( ( $T_i$, $T_j$ ) $\in$ $\prec$$^{CO}_{H2}$ ).\\
	\\
	Lets assume that their exists a conflict $(T_i, T_j)$ in $\prec$$^{CO}_{H1}$ but not in $\prec$$^{CO}_{H2}$. But, from \obsref{topsort} \& \obsref{TSorder} we know that $(T_i, T_j)$ $\in$ $\prec$$^{CO}_{H2}$. Thus, $\prec$$^{CO}_{H1}$ $\subseteq$ $\prec$$^{CO}_{H2}$.\\
	\\
	The relation is of improper subset because topological sort may introduce new real-time orders in $H_2$ which might not be present in $H_1$.
\end{proof}

\begin{lemma}
	\label{lem:genericLegalityLemma}
	Let $H_1$ and $H_2$ be \emph{equivalent} histories such that $\prec$$^{CO}_{H_1}$ $\subseteq$ $\prec$$^{CO}_{H_2}$. Then, $H_1$ is legal $\Longrightarrow$ $H_2$ is legal.
\end{lemma}

\begin{proof}
	We know $H_1$ is legal, \emph{wlog} let us say $(\rv_{j}(ht, k, v)$ $\in$ \met{$H_1$}$)$, such that $(up_{p}(ht, k, v_p) = \lupdt{\rv_{j}(ht, k, v)}{H_1})$ where,
	$(v = v_p \neq nill)$, if $(up_{p}(ht, k, v_p) = \tins_{p}(ht, k, v_p))$ or\\ 
	$(v=nill)$, if $(up_{p}(ht, k, v_p) = \tdel_{p}(ht, k, v_p))$.
	From the \emph{conflict-notion} \conf{$H_1$} has,
	\begin{equation}
	\label{eq:conflict1}
	\up_{p}(ht, k, v_p) \prec_{H_1}^{MR} \rv_{j}(ht, k, v)
	\end{equation}
	Let us assume $H_2$ is not \legal. Since, $H_1$ is equivalent to $H_2$ from \lemref{propConflictPreserv1} such that ($\rv_{j}(ht, k, v)$ $\in$ \met{$H_2$}). Since $H_2$ is not \legal{}, there exist a ($\up_{r}(ht, k, v_r)$ $\in$ \met{$H_2$}) such that ($up_{r}(ht, k, v_r) = \lupdt{\rv_{j}(ht, k, v)}{H_2}$). So \conf{$H_2$} has,
	\begin{equation}
	\label{eq:conflict2}
	\up_{r}(ht, k, v_r) \prec_{H_2}^{MR} \rv_{j}(ht, k, v)
	\end{equation}
	We know, ($\prec$$^{CO}_{H_1}$ $\subseteq$ $\prec$$^{CO}_{H_2}$) so,
	\begin{equation}
	\label{eq:conflict3}
	\up_{p}(ht, k, v_p) \prec_{H_2}^{MR} \rv_{j}(ht, k, v)
	\end{equation}
	From \lemref{propConflictPreserv1} ($\up_{r}(ht, k, v_r)$ $\in$ \met{$H_1$}). Since $H_1$ is \legal{} $\up_{r}(ht, k, v_r)$ can occur only in one of following \emph{conflicts},
	\begin{equation}
	\label{eq:conflict4}
	\up_{r}(ht, k, v_r) \prec_{H_1}^{MR} up_{p}(ht, k, v_p)
	\end{equation}
	\begin{center}
		or
	\end{center}
	\begin{equation}
	\label{eq:conflict5}
	\rv_{j}(ht, k, v) \prec_{H_1}^{MR} \up_{r}(ht, k, v_r) 
	\end{equation}
	In $H_1$ \Eqref{conflict5} is not possible, because if (\Eqref{conflict5} $\in$ \conf{$H_1$}) implies (\Eqref{conflict5} $\in$ \conf{$H_2$}) from ($\prec$$^{CO}_{H_1}$ $\subseteq$ $\prec$$^{CO}_{H_2}$) and in $H_2$ \Eqref{conflict2} and \Eqref{conflict5} cannot occur together. Thus only possible way $\up_{r}(ht, k, v_r)$ can occur in $H_1$ is via \Eqref{conflict4}. From \Eqref{conflict4} we have,
	\begin{equation}
	\label{eq:conflict6}
	\up_{r}(ht, k, v_r) \prec_{H_2}^{MR} up_{p}(ht, k, v_p)
	\end{equation}
	From \Eqref{conflict2}, \Eqref{conflict3} and \Eqref{conflict6} we have,
	\begin{center}
		$\up_{r}(ht, k, v_r) \prec_{H_2}^{MR} up_{p}(ht, k, v_p) \prec_{H_2}^{MR} \rv_{j}(ht, k, v)$
	\end{center}
	This contradicts that $H_2$ is not legal. Thus if $H_1$ is legal $\longrightarrow$ $H_2$ is legal.
\end{proof}

\begin{observation}
	\label{obs:txTS}
	Each transaction is assigned a unique time-stamp in \npbegin{()} method using a shared counter which always increases atomically.
\end{observation}

\begin{observation}
	\label{obs:opTS}
	Each successful method of a transaction is assigned the time-stamp of its own transaction.
\end{observation}

\begin{lemma}
	\label{lem:maxtsbehav}
	Consider a global state $S$ which has a node $n$, initialized with $max\_ts$. Then in any future state $S'$ the $max\_ts$ of $n$ should be greater then or equal to $S$. Formally, $\langle \forall S, S': (n \in S.Abs) \land (S \sqsubset S') \Rightarrow (n \in S'.Abs) \land (S.n.max\_ts \leq S'.n.max\_ts) \rangle$. 
\end{lemma}
\begin{proof}
	We prove by Induction on events that change the $max\_ts$ field of a node associated with a key, which are Line \ref{lin:delete27}, \ref{lin:delete31} \& \ref{lin:delete36} of \cldd{} method of \algoref{commonLu&Del} and Line \ref{lin:tryc23}, \ref{lin:tryc27}, \ref{lin:tryc31} \& \ref{lin:tryc37} of \nptc{} method of \algoref{trycommit}.\\
	\textbf{\texttt{Base condition:}} Initially, before the first event that changes the $max\_ts$ field of a node associated with a key, we know the underlying \lsl{} has immutable $S.head$ and $S.tail$ nodes with $(S.head.\bn = S.tail)$ and $(S.head.\rn = S.tail)$. 
	
	Lets assume, a node corresponding to the key is already the part of underlying $\rn$ which is having a time-stamp of $m_1$ as $T_1$ from \obsref{opTS} . Let say $m_{2}$ of $T_{2}$ wants to perform on that node, by observing the code at line 6 of $\nptov$ method of \algoref{tovalidation} , if TS($T_{2}$) $<$ curr.max\_ts.$m_1$(), $T_{2}$ will return abort, else to  succeed, TS($T_{2}$) $>$ curr.max\_ts.$m_1$() should evaluate to true. Thus, for successful completion of $m_{2}$ of $T_{2}$, TS($T_{2}$) should be greater then the TS($T_{1}$). Hence, node corresponding to the key, $max\_ts$ field should be updated in increasing order of TS values.\\
	\textbf{\texttt{Induction Hypothesis:}} Say, upto k events that change the $max\_ts$ field of a node associated with a key always in increasing TS value.\\
	\textbf{\texttt{Induction Step:}} 
	So, as seen from the code, the $(k+1)^{th}$ event which can change the $max\_ts$ field be only one of the following:
	\begin{enumerate}
		\item \textbf{\texttt{Line \ref{lin:delete27}, \ref{lin:delete31} \& \ref{lin:delete36} of \cldd{} method of \algoref{commonLu&Del} :}} By observing the code, line \ref{lin:delete18} of $\cldd{}$ method of \algoref{commonLu&Del} first invokes \nplsls{} method of \algoref{lslsearch} for finding the node corresponding to the key. Inside the $\nplsls{}$ method of \algoref{lslsearch} , it will do the $\nptov$ method of \algoref{tovalidation} , if $(curr.key = key)$.
		
		From induction hypothesis, node corresponding to the key is already the part of underlying $\rn$ which is having a time-stamp of $m_k$ of $T_k$ from \obsref{opTS}. Let say $m_{k+1}$ of $T_{k+1}$ wants to perform on that node, by observing the code at line 6 of $\nptov$ method of \algoref{tovalidation} , if TS($T_{k+1}$) $<$ curr.max\_ts.$m_k$(), $T_{k+1}$ will return abort, else to  succeed, TS($T_{k+1}$) $>$ curr.max\_ts.$m_k$() should evaluate to true. Thus, for successful completion of $m_{k+1}$ of $T_{k+1}$, TS($T_{k+1}$) should be greater then the TS($T_{k}$). Hence, node corresponding to the key, $max\_ts$ field should be updated in increasing order of TS values.
		
		
		
		\item \textbf{\texttt{Line \ref{lin:tryc23}, \ref{lin:tryc27}, \ref{lin:tryc31} \& \ref{lin:tryc37} of \nptc{} method of \algoref{trycommit} :}} By observing the code, line \ref{lin:tryc7} of $\nptc{}$ method of \algoref{trycommit} first invokes \nplsls{} method of \algoref{lslsearch} for finding the node corresponding to the key. Inside the $\nplsls{}$ method of \algoref{lslsearch} , it will do the $\nptov$ method of \algoref{tovalidation} , if $(curr.key = key)$.
		
		From induction hypothesis, node corresponding to the key is already the part of underlying $\rn$ which is having a time-stamp of $m_k$ as $T_k$ from \obsref{opTS} . Let say $m_{k+1}$ of $T_{k+1}$ wants to perform on that node, by observing the code at line 6 of $\nptov$ method of \algoref{tovalidation} , if TS($T_{k+1}$) $<$ curr.max\_ts.$m_k$(), $T_{k+1}$ will return abort, else to  succeed, TS($T_{k+1}$) $>$ curr.max\_ts.$m_k$() should evaluate to true. Thus, for successful completion of $m_{k+1}$ of $T_{k+1}$, TS($T_{k+1}$) should be greater then the TS($T_{k}$). Hence, node corresponding to the key, $max\_ts$ field should be updated in increasing order of TS values.
	\end{enumerate}
	
\end{proof}

\begin{corollary}
	\label{cor:increasingmaxTS}
	Every successful methods update the $max\_ts$ field of a $node$ associated with a $key$ always in increasing TS values. 
\end{corollary}

\begin{lemma}
	\label{lem:increasingTS}
	If \npbegin$(T_{i})$ occurs before \npbegin$(T_{j})$ then $TS(T_{i})$ preceds $TS(T_{j})$. Formally, $\langle \forall T \in H: (\npbegin(T_{i}) \prec \npbegin(T_{j})) \Leftrightarrow (TS(T_{i}) < TS(T_{j})) \rangle$. 
\end{lemma}

\begin{proof}
	$(Only$ $if)$ If $(\npbegin(T_{i}) \prec \npbegin(T_{j}))$ then $(TS(T_{i}) < TS(T_{j}))$. Lets assume $(TS(T_{j}) < TS(T_{i})$. From \obsref{txTS} ,
	\begin{equation}
	\label{eq:increasingTS1}
	\npbegin(T_{j}) \prec_H \npbegin(T_{i})
	\end{equation}
	but we know that,
	\begin{equation}
	\label{eq:increasingTS2}
	\npbegin(T_{j}) \succ_H \npbegin(T_{i})
	\end{equation}
	Which is a contradiction thus, $(TS(T_{i}) < TS(T_{j}))$.\\
	\\
	$(if)$ If $(TS(T_{i}) < TS(T_{j}))$ then $(\npbegin(T_{i}) \prec \npbegin(T_{j}))$. Let us assume $(\npbegin(T_{j}) \prec \npbegin(T_{i}))$. From \obsref{txTS} ,
	\begin{equation}
	\label{eq:increasingTS3}
	TS(T_{j}) < TS(T_{i})
	\end{equation}
	but we know that,
	\begin{equation}
	\label{eq:increasingTS4}
	TS(T_{j}) > TS(T_{i})
	\end{equation}
	Again, a contradiction.
\end{proof}

\begin{lemma}
	\label{lem:conflictTO}
	If ($T_i$, $T_j$) $\in$ \conf{H} $\Rightarrow$ TS($T_i$) $<$  TS($T_j$).
\end{lemma}
\begin{proof}($T_i$, $T_j$) can have two kinds of conflicts from our conflict notion.
	\begin{enumerate}
		\item \textbf{\texttt{If ($T_i$, $T_j$) is an real-time edge:}} 
		Since, $T_i$ \& $T_j$ are real time ordered. Therefore,
		
		\begin{equation}
		\label{eq:conflictTOR}
		\levt{T_i} \prec_H \fevt{T_j}
		\end{equation}
		And from program order of $T_i$,
		\begin{equation}
		\label{eq:conflictTOR1}
		\fevt{T_i} \prec_H \levt{T_i} \Rightarrow \npbegin{(T_i)} \prec_H \levt{T_i}
		\end{equation}
		From \Eqref{conflictTOR} and \Eqref{conflictTOR1} implies that,
		\begin{equation}
		\begin{split}
		\label{eq:conflictTOR2}
		\fevt{T_i} \prec_H \fevt{T_j} \Rightarrow \npbegin{(T_i)} \prec_H \npbegin{(T_j)} \\\xRightarrow[]{\lemref{increasingTS}} TS(T_i) < TS(T_j)
		\end{split}
		\end{equation}
		
		\item \textbf{\texttt{If ($T_i$, $T_j$) is a conflict edge:}}
		We prove this case by contradiction, lets assume ($T_i$, $T_j$) $\in$ \conf{H} \& TS($T_j$) $<$  TS($T_i$). Given that ($T_i$, $T_j$) $\in$ \conf{H} and from \defref{confh} we get, $m_i$ $\prec_{H}^{\mr}$ $m_j$.
		
		$m_i$ can be $\rvmt{s}$ or $\upmt{s}$ (which are taking the effects in \nptc{} method of \algoref{trycommit} ) and we know that after the $LP$ of $m_i$ of $T_i$, $node$ corresponding to the $key$ should be there in $\rn{}$ (from \corref{nodeinabs} \& \defref{Abs} ) and the time-stamp of that $node$ corresponding to $key$ should be equal to time-stamp of this method transaction time-stamp from \corref{nodeinabs} \& \obsref{opTS} . 
		
		From \lemref{key-change} \& \lemref{key-change-bl} we have that the nodes in the underlying data-structure are in increasing order of their keys, thus the key on which the operation is working has a unique location in underlying data-structure from \corref{uniquenodeKey} . So, when the \nplsls{} is invoked from a \mth{} $m_j$ of $T_j$, it returns correct location $(\bp, \rp, \rc, \bc)$ of corresponding $key$ as observed from \obsref{lslSearch} \& \lemref{lslSearch-ret} .
		
		
		Now, $m_j$ similar to $m_i$ take effect on the same node represented by key $k$ (from \obsref{obs-nodekey} \& \corref{uniquenodeKey} ) \& from \obsref{futurestate} we know that the $node$ corresponding to the key $k$ is still reachable via $\rn$. Thus, we know that $T_i$ \& $T_j$ will work on same node with key $k$.
		
		By observing the code at line 6 \& 9 of $\nptov$ method of \algoref{tovalidation} , we know since, TS($T_j$) $<$ curr.max\_ts.$m_i$(), $T_{j}$ will return abort from \corref{increasingmaxTS} . 
		In \algoref{tovalidation} for \nptov{} to succeed, TS($T_j$) $>$ curr.max\_ts.$m_i$() should evaluate to true from \corref{increasingmaxTS} . 
		Thus, TS($T_j$) $<$  TS($T_i$), a contradiction. Hence, If ($T_i$, $T_j$) $\in$ \conf{H} $\Rightarrow$ TS($T_i$) $<$  TS($T_j$).
	\end{enumerate}
\end{proof}

\begin{lemma}
	If $($ $T_1$, $T_2$ $\cdots$ $T_n$ $)$ is a path in $CG(H)$, this implies that $($TS$(T_1)$ $<$  TS$(T_2)$ $<$ $\cdots$  $<$ TS$(T_n))$.
	\label{lem:increasingTO}
\end{lemma}

\begin{proof}
	The proof goes by induction on length of a path in $CG(H)$.\\
	\\
	\textbf{\texttt{Base Step:}} Assume $($ $T_1$, $T_2$ $)$ be a path of length 1. Then, from \lemref{conflictTO} $($TS$(T_1)$ $<$  TS$(T_2))$.\\
	\\
	\textbf{\texttt{Induction Hypothesis:}} The claim holds for a path of length $(n-1)$. That is,
	\begin{equation}
	\label{eq:increasingTO1}
	TS(T_1) < TS(T_2) < \cdots < TS(T_{n-1})
	\end{equation}
	\textbf{\texttt{Induction Step:}} 
	Let $T_{n}$ is a transaction in a path of length $n$. Then, ($T_{n-1}$, $T_{n}$) is path in $CG(H)$. Thus, it follows from \lemref{conflictTO} that,
	\begin{equation}
	\label{eq:increasingTO2}
	TS(T_{n-1}) < TS(T_{n}) \xRightarrow[]{\Eqref{increasingTO1}} (TS(T_1) <  TS(T_2) < \cdots  < TS(T_n))
	\end{equation}
	Hence, the lemma.
\end{proof}

\begin{theorem}
	\label{thm:CG}
	Consider a history $H$ generated by \otm. Then there exists a sequential \& legal history $H'$ equivalent to $H$ such that the conflict-graph of $H'(CG(H'))$ is acyclic.
\end{theorem}

\begin{proof}
	Assume that $CG(H')$ is cyclic, then their exist a cycle say of form $($ $T_1$,  $T_2$ $\cdots$ $T_n$, $T_1$ $)$, for all (n $\geq$ 1). From \lemref{increasingTO} ,
	\begin{equation}
	\label{eq:CG1}
	TS(T_1) < TS(T_2) \cdots < TS(T_n) < TS(T_1) \xRightarrow[]{} TS(T_1) < TS(T_1)
	\end{equation}
	
	But, this is impossible as each transaction has unique time-stamp, refer \obsref{txTS} . Hence the theorem. 
\end{proof}

\begin{theorem}
	\label{thm:coacyclic}
	A legal \otm{} history $H$ is \coop{} iff CG(H) is acyclic.
\end{theorem}
\begin{proof}
	(Only if) If H is co-opaque and legal, then CG(H) is acyclic: Since H is co-opaque, there exists a legal t-sequential history S equivalent to $\bar{H}$ and S respects $\prec$$^{RT}_H$ and $\prec$$^{CO}_H$ (from co-opacity\cite{KuzPer:NI:TCS:2017}). Thus from the conflict graph construction we have that (CG($\bar{H}$)=CG(H)) is a sub graph of CG(S). Since S is sequential, it can be inferred that CG(S) is acyclic using \lemref{serialIsAcyclic}. Any sub graph of an acyclic graph is also acyclic. Hence CG(H) is also acyclic.

	(if) If H is legal and CG(H) is acyclic then H is co-opaque:
	Suppose that CG(H) = CG($\bar{H}$) is acyclic. Thus we can
	perform a topological sort on the vertices of the graph and
	obtain a sequential order. Using this order, we can obtain a
	sequential schedule S that is equivalent to $\bar{H}$. Moreover, by construction, S respects $\prec$$^{RT}_H$ = $\prec$$^{RT}_{\bar{H}}$ and $\prec$$^{CO}_H$ = $\prec$$^{CO}_{\bar{H}}$. 
	
	Since every two operations related by the conflict relation in S are also related by $\prec$$^{CO}_{\bar{H}}$, we obtain $\prec$$^{CO}_{\bar{H}}$ $\subseteq$ $\prec$$^{CO}_S$. Since H is legal, $\bar{H}$ is also legal. Combining this with \lemref{genericLegalityLemma}, We get that S is also legal. This satisfies all the conditions necessary for H to be co-opaque.
\end{proof}

\section{Evaluation}
\label{sec:results}

We performed all the experiments on Intel(R) Xeon(R) CPU E5-2690 v4 @ 2.60GHz machine with 56 CPUs and 32K L1 data cache and 32 GB memory. Each thread spawns 10 transactions each of which randomly generate up to 5 \mths of \otm. We assume that the \tab{} of \otm has 5 buckets and each of the bucket (or list in case of \ltm{}) can have maximum size of 1K keys. We ran the experiments to calculate two parameters: (1) time taken for a transaction to commit. Upon abort, a transaction is retried until it commits. (2) Number of aborts incurred until all the transactions commit. 

We compare \otm{} with the ESTM\cite{Felber:2017:jpdc} based \tab{} and the transactional \emph{hash-table} application built using RWSTM\cite{Singh2017PerformanceCO} which is synchronised by basic time stamp ordering protocol\cite[Chap 4]{WeiVoss:2002:Morg}. Further, we evaluate \emph{list-OSTM} with the state of the art lock-free transactional list (LFT)\cite{Zhang:2016:LTW:2935764.2935780}, NOrec STM list (NTM)\cite{conf/ppopp/DalessandroSS10} and boosting list (BST)\cite{herlihy2008transactional}. All these implementations are directly taken from the TLDS framework\footnote{https://ucf-cs.github.io/tlds/}. The experiments were performed under two kinds of workloads. Update intensive(lookup:50\%, insert:25\%, delete:25\%) and lookup intensive(lookup:70\%, insert:10\%, delete:20\%). In case of lookup intensive workloads, with higher percentage of reads, ESTM was performing better. Here, we have shown lookup intensive workload in which \otm performs better. The evaluation is done by varying threads from 2 to 64 in power of 2. Before each application is run there is a initialization phase where the data structure is populated randomly with nodes of half its maximum size.  

\vspace{1mm}
\noindent
\textbf{\otm.\footnote{lib source code link: https://github.com/PDCRL/ostm}}
 \figref{htostm-th} shows that w.r.t. time taken \otm{} outperforms
 ESTM\cite{Felber:2017:jpdc} and RWSTM on an average by 3 times for lookup intensive workload. Plus, for update intensive workload \otm{} on average is 6 times better than ESTM \& RWSTM. Similarly, in terms of aborts, \otm{} has 3 \& 2 times lesser aborts than  ESTM and RWSTM for lookup intensive workload, respectively. Also for update intensive load \otm{} has 7 and 8 times lesser aborts with ESTM and RWSTM respectively, as can be seen in \figref{htostm-abort}.
\vspace{-.9cm}
\begin{figure}[H]
	\centering
	\subfloat[\otm{} time in second(s) \label{fig:htostm-th}]
	{\includegraphics[width=0.5\textwidth]{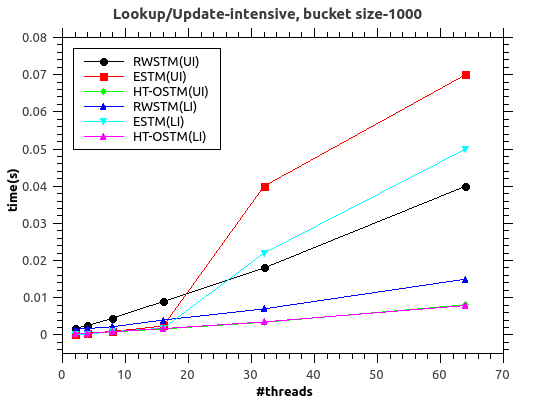}}
	\subfloat[\otm{} aborts \label{fig:htostm-abort}]
	{\includegraphics[width=0.5\textwidth]{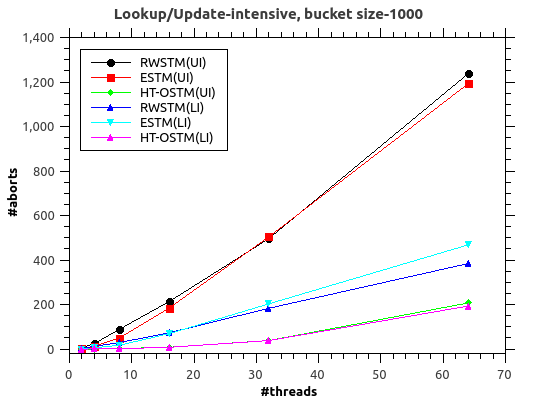}}
	\hfill\hfill
	\subfloat[\emph{list-OSTM} time in second(s) \label{fig:lsostm-th}]
	{\includegraphics[width=0.5\textwidth]{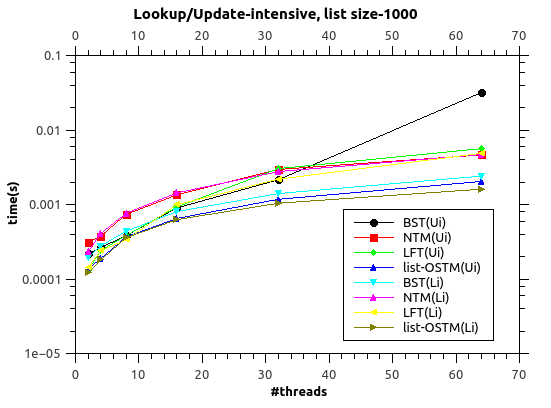}}
	\subfloat[\emph{list-OSTM} aborts \label{fig:lsostm-abort}]
	{\includegraphics[width=0.5\textwidth]{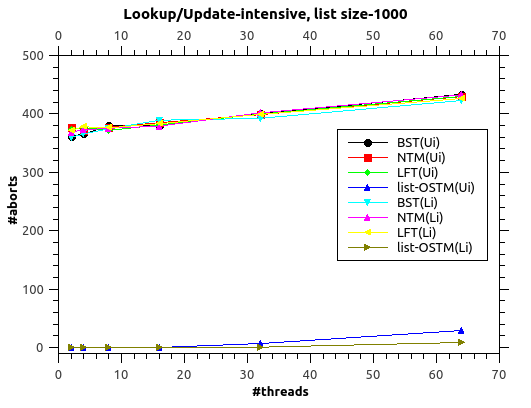}}
	\caption[The average and standard deviation of critical parameters]
	{\otm{} and \emph{list-OSTM} evaluation. Each curve is named as technique name(workload type). LI/UI denotes lookup intensive/ update intensive.}
	\label{fig:eval}
\end{figure}
\noindent
\textbf{\emph{list-OSTM.}}
The average aborts for \emph{list-OSTM} never go beyond 30 in magnitude while that of other techniques (in \figref{lsostm-abort}) are of 388 in the magnitude for both types of workloads. While time taken is 76\%, 89\% and 33\% (with lookup intensive) and 77\%, 77\% and 154\% (with update intensive) better than LFT, NTM and BST respectively (as shown in \figref{lsostm-th}).

For better understanding, we have done the various experimental analysis while varying the workloads for \otm{} and \emph{list\_OSTM} are below:
\newpage
\noindent
\textbf{\emph{HT-OSTM evaluation for lookup intensive:}} \otm{} witness lowest time, lowest aborts and highest throughput in comparison to the ESTM and BTO based RWSTM as shown in \figref{htostm-li-time}, \ref{fig:htostm-li-abort} and \ref{fig:htostm-li-th}. 

\begin{figure}[H]
	\centering
	\subfloat[\otm{} time \label{fig:htostm-li-time}]
	{\includegraphics[width=0.46\textwidth]{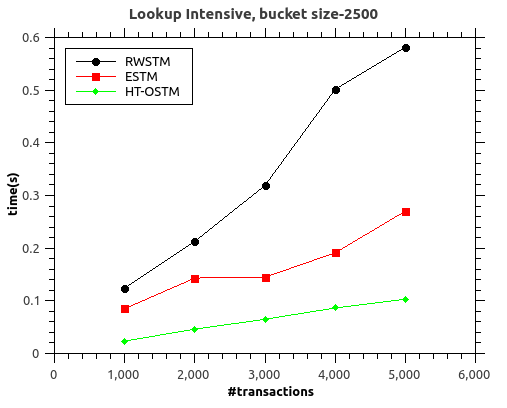}}
	\subfloat[\otm{} aborts \label{fig:htostm-li-abort}]
	{\includegraphics[width=0.5\textwidth]{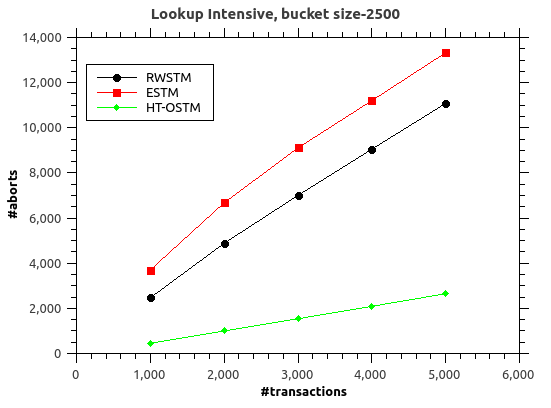}}
	\hfill
	\subfloat[\otm{} throughput \label{fig:htostm-li-th}]
	{\includegraphics[width=0.5\textwidth]{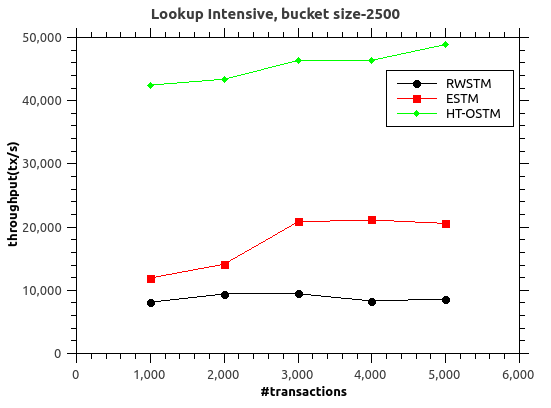}}
	\caption[The average and standard deviation of critical parameters]
	{\otm{}:Lookup Intensive(lookup:80\%, insert:15\%, delete:5\%). Number of operations/transaction are 10.}
	\label{fig:eval1}
\end{figure}

\noindent
\textbf{\emph{HT-OSTM evaluation for mid intensive:}} \otm{} witness lowest time, lowest aborts and highest throughput in comparison to the ESTM and BTO based RWSTM as shown in \figref{htostm-li-time}, \ref{fig:htostm-li-abort} and \ref{fig:htostm-li-th}. 

\begin{figure}[H]
	\centering
	\subfloat[\otm{} time \label{fig:htostm-mi-time}]
	{\includegraphics[width=0.46\textwidth]{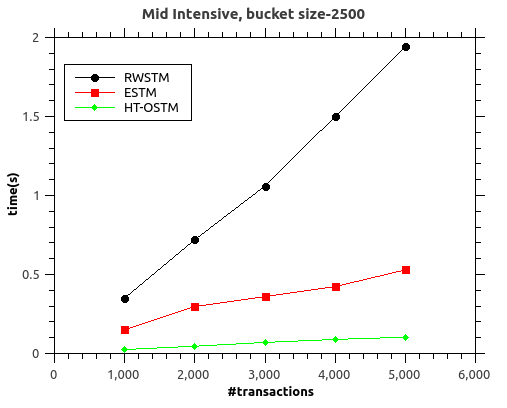}}
	\subfloat[\otm{} aborts \label{fig:htostm-mi-abort}]
	{\includegraphics[width=0.5\textwidth]{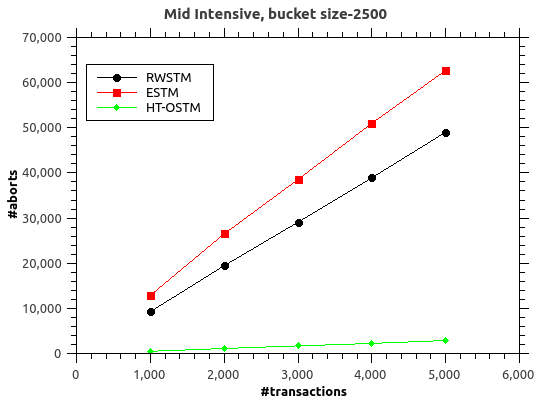}}
	\hfill
	\subfloat[\otm{} throughput \label{fig:htostm-mi-th}]
	{\includegraphics[width=0.5\textwidth]{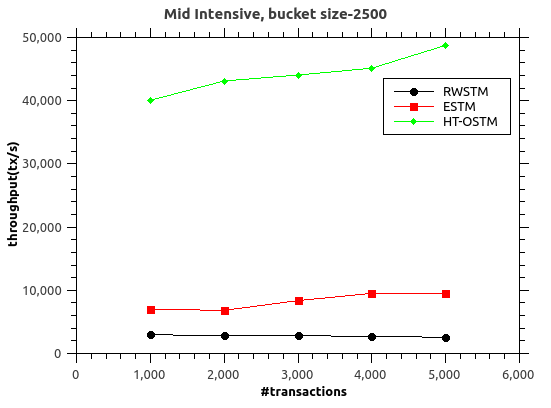}}
	\caption[The average and standard deviation of critical parameters]
	{\otm{}:Mid Intensive(lookup:50\%, insert:25\%, delete:25\%). Number of operations/transaction are 10.}
	\label{fig:eval2}
\end{figure}

\noindent
\textbf{\emph{HT-OSTM evaluation for update intensive:}} \otm{} witness lowest time, lowest aborts and highest throughput in comparison to the ESTM and BTO based RWSTM. The experiments in \figref{htostm-ui-time}, \ref{fig:htostm-ui-abort} and \ref{fig:htostm-ui-th} have the bucket size of 2500 (range of keys allowed) which implies that contention is low. \figref{htostm-ui-hc-time}, \ref{fig:htostm-ui-hc-abort} and \ref{fig:htostm-ui-hc-th} show the experiments for high contention with bucket size of 30.

\begin{figure}[H]
	\centering
	\subfloat[\otm{} time \label{fig:htostm-ui-time}]
	{\includegraphics[width=0.46\textwidth]{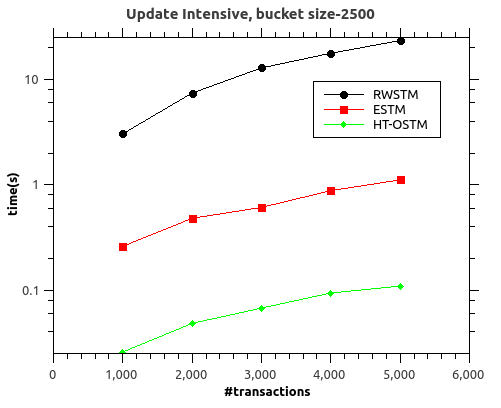}}
	\subfloat[\otm{} aborts \label{fig:htostm-ui-abort}]
	{\includegraphics[width=0.5\textwidth]{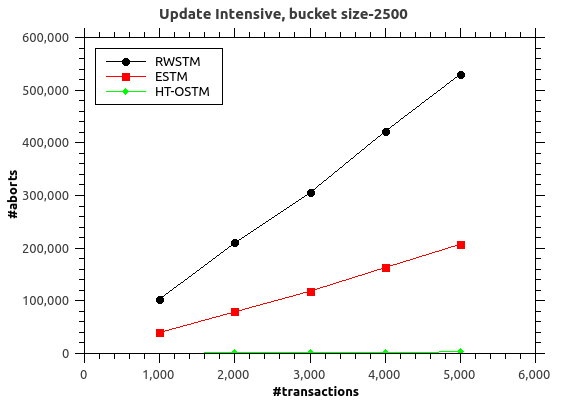}}
	\hfill
	\subfloat[\otm{} throughput \label{fig:htostm-ui-th}]
	{\includegraphics[width=0.5\textwidth]{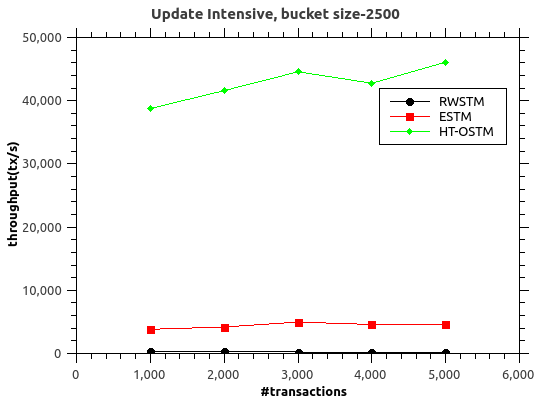}}
	\caption[The average and standard deviation of critical parameters]
	{\otm{}:Update Intensive(lookup:10\%, insert:45\%, delete:45\%). Number of operations/transaction are 10.}
	\label{fig:eval3}
\end{figure}

\begin{figure}[H]
	\centering
	\subfloat[\otm{} time \label{fig:htostm-ui-hc-time}]
	{\includegraphics[width=0.46\textwidth]{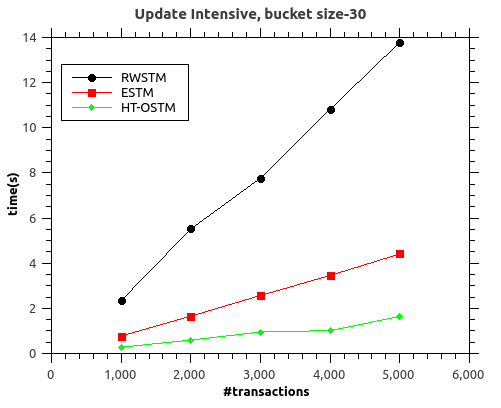}}
	\subfloat[\otm{} aborts \label{fig:htostm-ui-hc-abort}]
	{\includegraphics[width=0.5\textwidth]{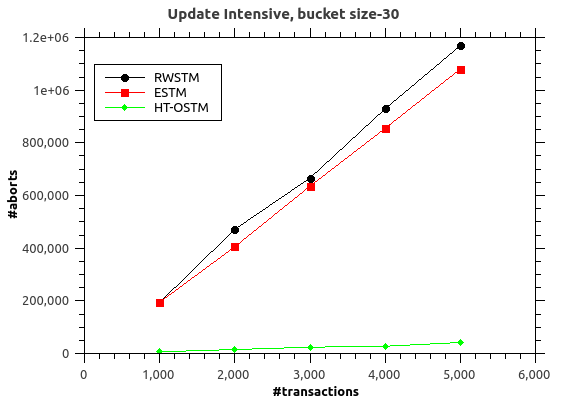}}
	\hfill
	\subfloat[\otm{} throughput \label{fig:htostm-ui-hc-th}]
	{\includegraphics[width=0.5\textwidth]{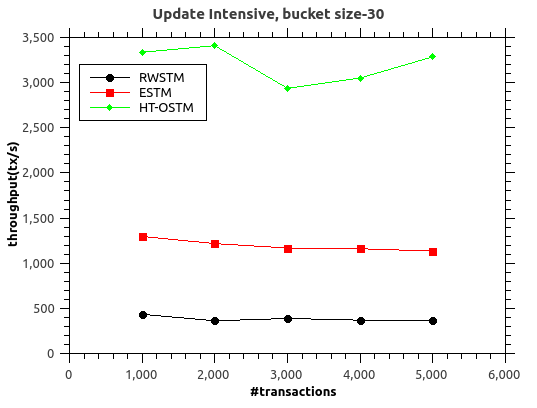}}
	\caption[The average and standard deviation of critical parameters]
	{\otm{}:Update Intensive(lookup:10\%, insert:45\%, delete:45\%) with high contention. Number of operations/transaction are 10.}
	\label{fig:eval4}
\end{figure}

\newpage
\noindent
\textbf{\emph{list-OSTM evaluation for lookup intensive:}} \emph{list-OSTM} witness lowest time, lowest aborts  in comparison to the LFT, NTM and BST as shown in \figref{lsostm-li-time} and \ref{fig:lsostm-li-abort}. 

\begin{figure}[H]
	\centering
	\subfloat[\emph{list-OSTM} time \label{fig:lsostm-li-time}]
	{\includegraphics[width=0.5\textwidth]{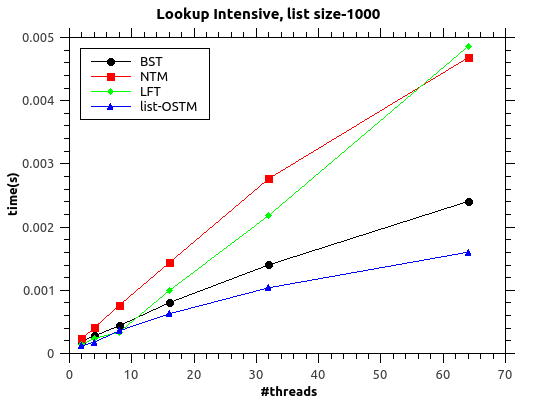}}
	\subfloat[\emph{list-OSTM} aborts \label{fig:lsostm-li-abort}]
	{\includegraphics[width=0.5\textwidth]{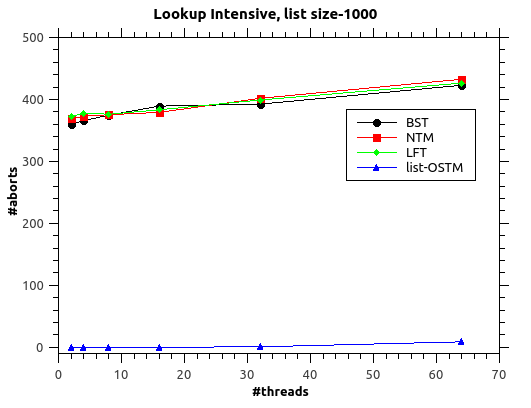}}
	\caption[The average and standard deviation of critical parameters]
	{\emph{list-OSTM}: Lookup Intensive(lookup:80\%, insert:15\%, delete:5\%). Number of operations/transaction are 10.}
	\label{fig:eval5}
\end{figure}

\noindent
\textbf{\emph{list-OSTM evaluation for mid intensive:}} \emph{list-OSTM} witness lowest time, lowest aborts  in comparison to the LFT, NTM and BST as shown in \figref{lsostm-mid-time} and \ref{fig:lsostm-mid-abort}. 
However, BST for lower number of threads takes similar time to list-OSTM.
\begin{figure}[H]
	\centering
	\subfloat[\emph{list-OSTM} time \label{fig:lsostm-mid-time}]
	{\includegraphics[width=0.5\textwidth]{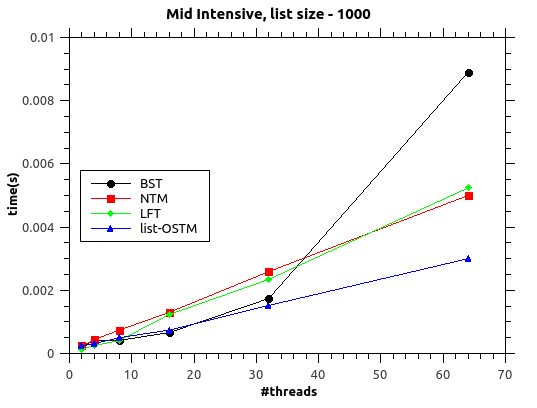}}
	\subfloat[\emph{list-OSTM} aborts \label{fig:lsostm-mid-abort}]
	{\includegraphics[width=0.5\textwidth]{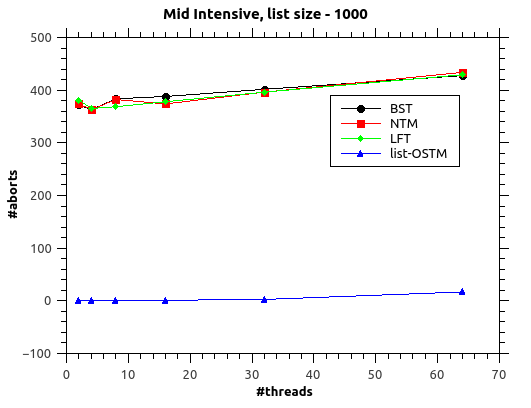}}
	\caption[The average and standard deviation of critical parameters]
	{\emph{list-OSTM}: Mid Intensive(lookup:50\%, insert:25\%, delete:25\%). Number of operations/transaction are 10.}
	\label{fig:eval7}
\end{figure}

\noindent
\textbf{\emph{list-OSTM evaluation for update intensive:}} \emph{list-OSTM} witness lowest time, lowest aborts  in comparison to the LFT, NTM and BST as shown in \figref{lsostm-ui-time} and \ref{fig:lsostm-ui-abort}.
\begin{figure}[H]
	\centering
	\subfloat[\emph{list-OSTM} time \label{fig:lsostm-ui-time}]
	{\includegraphics[width=0.5\textwidth]{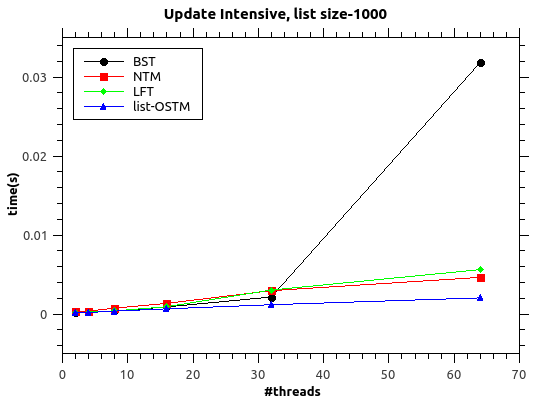}}
	\subfloat[\emph{list-OSTM} aborts \label{fig:lsostm-ui-abort}]
	{\includegraphics[width=0.5\textwidth]{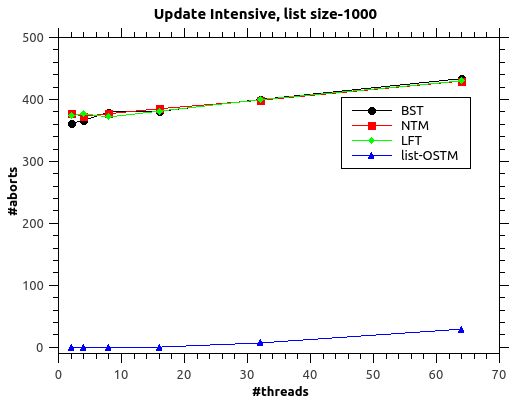}}
	\caption[The average and standard deviation of critical parameters]
	{\emph{list-OSTM}: Update Intensive(lookup:10\%, insert:45\%, delete:45\%). Number of operations/transaction are 10.}
	\label{fig:eval6}
\end{figure}


\vspace{-.2cm}
\section{Conclusion and Future Work}
\label{sec:conc}
\vspace{-.2cm}
In this paper, we build a model for building highly concurrent and composable data structures with object level transactions called \lotm{}. We show that higher concurrency can be obtained by considering \lotms as compared to traditional \rwtm{} by leveraging richer object-level semantics. We propose comprehensive theoretical model based on legality semantics and conflict notions for \tab{} based \lotm, \otm. Using these notions we extend the definition of \opty and \coopty for \otm{s} in \secref{opty}. Then, based on this model, we develop a practical implementation of \otm{} \& \ltm{} to verify the gains achieved as demonstrated in \secref{results}. Further, we prove that proposed model is \coop\cite{KuzPer:NI:TCS:2017} thus composable.



\bibliographystyle{plain}
\bibliography{citations}

\begin{thebibliography}{10}

\bibitem{conf/ppopp/DalessandroSS10}
Luke Dalessandro, Michael~F. Spear, and Michael~L. Scott.
\newblock Norec: streamlining stm by abolishing ownership records.
\newblock In R.~Govindarajan, David~A. Padua, and Mary~W. Hall, editors, {\em
  PPOPP}, pages 67--78. ACM, 2010.

\bibitem{Felber:2017:jpdc}
Pascal Felber, Vincent Gramoli, and Rachid Guerraoui.
\newblock Elastic transactions.
\newblock {\em J. Parallel Distrib. Comput.}, 100(C):103--127, February 2017.

\bibitem{Fraser:2007:CPW:1233307.1233309}
Keir Fraser and Tim Harris.
\newblock Concurrent programming without locks.
\newblock {\em ACM Trans. Comput. Syst.}, 25(2), May 2007.

\bibitem{GuerKap:2008:PPoPP}
Rachid Guerraoui and Michal Kapalka.
\newblock On the {C}orrectness of {T}ransactional {M}emory.
\newblock In {\em PPoPP}, pages 175--184. ACM, 2008.

\bibitem{Harris_abstractnested}
Tim Harris and et~al.
\newblock Abstract nested transactions, 2007.

\bibitem{Harretal:2005:PPoPP}
Tim Harris, Simon Marlow, Simon Peyton-Jones, and Maurice Herlihy.
\newblock Composable memory transactions.
\newblock In {\em PPoPP}, pages 48--60, New York, NY, USA, 2005. ACM.

\bibitem{Hassan+:OptBoost:PPoPP:2014}
Ahmed Hassan, Roberto Palmieri, and Binoy Ravindran.
\newblock Optimistic transactional boosting.
\newblock In Jos{\'{e}}~E. Moreira and James~R. Larus, editors, {\em PPoPP},
  pages 387--388. {ACM}, 2014.

\bibitem{Heller+:LazyList:PPL:2007}
Steve Heller, Maurice Herlihy, Victor Luchangco, Mark Moir, William N.~Scherer
  III, and Nir Shavit.
\newblock A lazy concurrent list-based set algorithm.
\newblock {\em Parallel Processing Letters}, 17(4):411--424, 2007.

\bibitem{Herlihy:ArtBook:2012}
M.~Herlihy and N.~Shavit.
\newblock {\em The Art of Multiprocessor Programming}.
\newblock Elsevier Science, 2012.

\bibitem{HerlMoss:1993:SigArch}
Maurice Herlihy and J.~Eliot B.Moss.
\newblock Transactional memory: {A}rchitectural {S}upport for {L}ock-{F}ree
  {D}ata {S}tructures.
\newblock {\em SIGARCH Comput. Archit. News}, 21(2):289--300, 1993.

\bibitem{herlihy2008transactional}
Maurice Herlihy and Eric Koskinen.
\newblock Transactional boosting: a methodology for highly-concurrent
  transactional objects.
\newblock In {\em PPoPP}, pages 207--216. ACM, 2008.

\bibitem{HerlWing:1990:TPLS}
Maurice~P. Herlihy and Jeannette~M. Wing.
\newblock Linearizability: a correctness condition for concurrent objects.
\newblock {\em ACM Trans. Program. Lang. Syst.}, 12(3):463--492, 1990.

\bibitem{KuzPer:NI:TCS:2017}
Petr Kuznetsov and Sathya Peri.
\newblock Non-interference and local correctness in transactional memory.
\newblock {\em Theor. Comput. Sci.}, 688:103--116, 2017.

\bibitem{KR:2011:OPODIS}
Petr Kuznetsov and Srivatsan Ravi.
\newblock On the cost of concurrency in transactional memory.
\newblock In {\em OPODIS}, pages 112--127, 2011.

\bibitem{ni2007open}
Yang Ni, Vijay~S Menon, Ali-Reza Adl-Tabatabai, Antony~L Hosking, Richard~L
  Hudson, J~Eliot~B Moss, Bratin Saha, and Tatiana Shpeisman.
\newblock Open nesting in software transactional memory.
\newblock In {\em PPoPP}. ACM, 2007.

\bibitem{Papad:1979:JACM}
Christos~H. Papadimitriou.
\newblock {The serializability of concurrent database updates}.
\newblock {\em J. ACM}, 26(4), 1979.

\bibitem{ShavTou:1995:PODC}
Nir Shavit and Dan Touitou.
\newblock {S}oftware {T}ransactional {M}emory.
\newblock In {\em PODC}, pages 204--213, 1995.

\bibitem{Singh2017PerformanceCO}
Ajay Singh, Sathya Peri, G.~Monika, and Anila Kumari.
\newblock Performance comparison of various stm concurrency control protocols
  using synchrobench.
\newblock {\em 2017 National Conference on Parallel Computing Technologies
  (PARCOMPTECH)}, pages 1--7, 2017.

\bibitem{WeiVoss:2002:Morg}
Gerhard Weikum and Gottfried Vossen.
\newblock {\em Transactional Information Systems: Theory, Algorithms, and the
  Practice of Concurrency Control and Recovery}.
\newblock Morgan Kaufmann, 2002.

\bibitem{Zhang:2016:LTW:2935764.2935780}
Deli Zhang and Damian Dechev.
\newblock Lock-free transactions without rollbacks for linked data structures.
\newblock SPAA '16, pages 325--336, New York, NY, USA, 2016. ACM.

\end{thebibliography}

\end{document}